\documentclass[journal]{vgtc}

\usepackage{amsmath}
\usepackage{amsfonts}
\usepackage{amsthm}
\usepackage{soul}
\usepackage{xcolor}
\usepackage{placeins}
\usepackage[table]{xcolor}

\newtheorem{theorem}{Theorem}
\newtheorem{definition}{Definition}
\newtheorem{lemma}{Lemma}

\newtheorem{corollary}{Corollary}

\newcommand{\tool}{\emph{FeatureZ}}

\newcommand{\para}[1]{\noindent\textbf{#1}}

\newcommand{\X}{\mathbb{X}}
\newcommand{\M}{\mathbb{M}}
\newcommand{\R}{\mathbb{R}}

\newcommand{\T}{\mathcal{T}}
\newcommand{\Z}{\mathbb{Z}}
\newcommand{\script}{\mathcal}
\newcommand{\St}{\mathrm{St}}
\newcommand{\Lk}{\mathrm{Lk}}
\newcommand{\bigO}{\mathcal{O}}
\newcommand{\ol}{\overline}

\onlineid{2081}

\vgtccategory{Research}

\title{\textbf{FeatureZ}: A General Framework for Feature-Preserving Compression via Pointwise Bounds and Star Classification}

\author{
  \authororcid{Nathaniel Gorski}{0009-0001-8205-5640}, 
  \authororcid{Xin Liang}{0000-0002-0630-1600}, 
  \authororcid{Hanqi Guo}{0000-0001-7776-1834},
  and \authororcid{Bei Wang}{0000-0002-9240-0700}
}

\authorfooter{
\item Nathaniel Gorski is with the University of Utah.
  	E-mail: gorski@sci.utah.edu.
\item Xin Liang is with Oregon State University. 
 	E-mail: lianxin@oregonstate.edu
\item Hanqi Guo is with the Ohio State University. 
	E-mail: guo.2154@osu.edu. 
\item Bei Wang is with the University of Utah. 
	E-mail: beiwang@sci.utah.edu. 
}

\abstract{Geometric and topological features, such as isosurfaces, quantiles, merge trees, and Morse--Smale complexes, are central to the analysis and visualization of scientific data across diverse domains, including medical imaging, climate science, materials science, and astronomy. However, most lossy compressors for scientific data provide only pointwise error guarantees and do not preserve derived features. Existing feature-preserving compressors are often difficult to develop and typically tailored to a single feature. In this paper, we introduce {\tool}, a lossy compression framework for structured volumetric scalar fields that can preserve a wide class of geometric and topological features. In particular, {\tool} frames feature-preserving compression as preserving pointwise upper and lower bounds together with a consistent classification of each point's star (i.e., its incident cells) in the underlying structured mesh. Although not all notions of feature preservation fit this framework, many features and topological descriptors targeted by existing compressors do. {\tool} provides an efficient implementation of feature-preserving compression under this formulation. In particular, {\tool} operates as an augmentation layer that refines the output of an existing lossy compressor. It first applies quantization to enforce pointwise bounds, and then employs an iterative procedure to ensure consistent star classification. We demonstrate that {\tool} preserves diverse features during compression with minimal overhead, achieving compression ratios comparable to or better than methods specialized for individual features.

}

\keywords{Lossy compression, error-bounded compression, generalized feature preservation, topological and geometric features, topological data analysis, scientific visualization, extensible frameworks}

\nocopyrightspace

\teaser{
\centering
\includegraphics[width=0.8\linewidth]{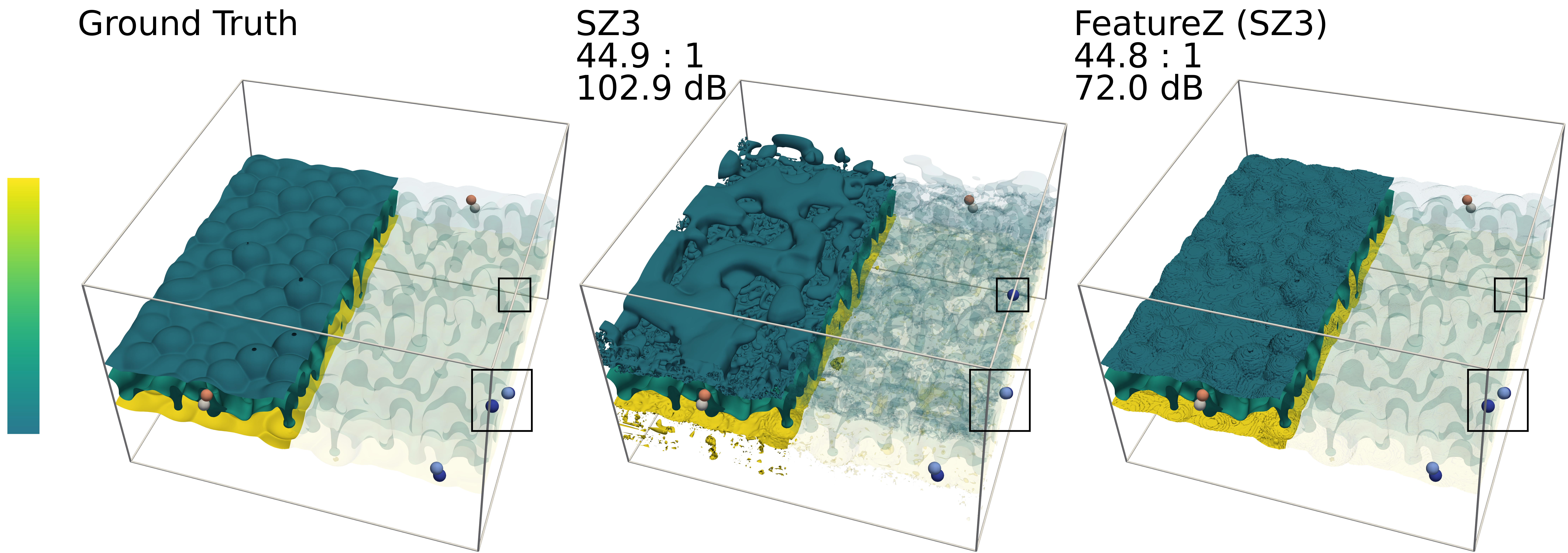}
\vspace{-3mm}
\caption{The Miranda dataset visualized with three isocontours and the vertices of the simplified contour tree with a persistence threshold of $4\%$ of the range. Left: ground truth. Center: reconstruction using SZ3. Right: reconstruction using {\tool}, denoted as \emph{\tool (SZ3)}, which augments SZ3 to preserve isocontours and the simplified contour tree simultaneously. Both SZ3 and {\tool} operate at comparable compression ratios. SZ3 introduces substantial errors in both the isocontours and the contour tree, whereas {\tool} preserves these structures. This example highlights {\tool}'s ability to simultaneously preserve multiple features.}
\label{fig:teaser}
\vspace{-2mm}
}

\begin{document}
\maketitle
\section{Introduction}
\label{sec:introduction}

Geometric and topological features are derived quantities that describe the shape, structure, and connectivity of scientific data, and they play a central role in its analysis and visualization across diverse domains. For example, in medical imaging, isocontours are widely used to reconstruct anatomical structures such as soft tissue and bone surfaces~\cite{lorensen1987marching}. In astronomy, critical points of the matter density field and their connectivity define the skeleton of the cosmic web~\cite{aragoncalvo2010spine}, while contour trees have been employed to denoise ALMA data cubes~\cite{rosen2021using}. In quantum chemistry, Morse--Smale complexes enable detailed analysis of atomic basins as well as the extraction of molecular graphs~\cite{bhatia2018topoms}. In climate science, critical points identify cyclone centers~\cite{yan2024trophy}, whereas isocontours and quantiles support the analysis of temperature, pressure, and humidity fields, facilitating the study of extreme events such as atmospheric rivers~\cite{lan2024topological} and atmospheric blocking~\cite{li2026spatiotemporal}.

Despite their importance, most lossy compressors for scientific data are designed to control pointwise reconstruction error and do not explicitly account for these derived quantities. Consequently, even when global error bounds are small, compression can introduce significant distortions in geometric and topological structures~\cite{gorski2025general}, compromising the fidelity of downstream analysis and visualization. To address this limitation, recent work has proposed feature-preserving compression techniques that explicitly target specific structures, such as contour trees~\cite{gorski2025general, yan2024toposz} and Morse--Smale segmentations~\cite{li2025msz,li2026pmsz}.

However, existing feature-preserving compressors require significant time to develop and are typically specialized to a single feature, with designs that do not generalize across different types of structures. For example, a previous compressor~\cite{gorski2025general} (referred to as TopoA in this paper) computes the contour tree of the decompressed data and applies corrections when deviations from the expected topology are detected. While effective, this approach is intrinsically tied to contour trees and does not readily extend to other feature classes, reflecting a broader limitation of current methods. As a result, feature-preserving compression remains inflexible and fragmented. There is no systematic pathway for scientists to preserve new, application-specific features without designing bespoke solutions from scratch. Furthermore, each feature typically requires a dedicated compressor (e.g., TopoA \cite{gorski2025general} for contour trees or pMSz \cite{li2026pmsz} for Morse--Smale segmentation), forcing users to identify, install, and manage an expanding collection of specialized tools. This fragmented ecosystem underscores the need for a flexible, unified framework that can support user-defined feature preservation within a single, generalizable compression paradigm.

In this work, we introduce {\tool}, a general and extensible compression framework designed to preserve a broad class of features while satisfying user-specified error bounds. {\tool} is \emph{general} in that it provides a novel formulation of feature preservation that unifies a number of  geometric and topological features within a single framework. It is also \emph{extensible}, allowing user-defined features to be incorporated by adding a few new functions without redesigning the compression pipeline.

While {\tool} does not preserve arbitrary derived features, it supports a broad and practically relevant class. Specifically, it supports any feature whose preservation can be expressed through pointwise upper and lower bounds, together with a consistent classification of each point’s star (i.e., the collection of incident cells) in the underlying structured mesh. 
We show that many commonly used features, including isocontours, quantiles, critical points, contour trees \cite{carr2003computing}, merge trees \cite{lukasczyk2024extreem}, and Morse--Smale segmentations \cite{maack2024parallel}, fit naturally within this formulation and can be preserved by {\tool}. 
Notably, all features targeted by existing topology-preserving compressors for scalar-field data can be expressed within this framework. 

We further develop an efficient algorithm that preserves features under this novel formulation. 
Similar to prior topology-preserving approaches, {\tool} operates as an augmentation layer on top of an existing base compressor such as ZFP~\cite{lindstrom2014fixed} or SZ3~\cite{liang2023sz3}. It first applies the base compressor and analyzes the decompressed data, then enforces feature preservation through quantization-based bound correction and an iterative scan procedure that ensures consistent star classification. Because this process relies only on local information---the function values at a point and within its star---{\tool} avoids repeated global feature computation, enabling efficient and scalable compression. Our empirical results demonstrate that {\tool} achieves compression ratios and runtimes comparable to specialized methods. 
We introduce {\tool}, a general framework for feature-preserving compression based on star classification: 
\begin{itemize}[noitemsep,leftmargin=*]
\item We introduce a unified formulation of feature-preserving compression based on pointwise upper and lower bounds together with local star consistency. We show that this formulation captures many geometric and topological features targeted by existing feature-preserving compressors, as well as several previously unsupported feature classes.
\item We develop an efficient framework that augments existing lossy compressors with feature-preservation guarantees. We prove that the resulting compressed data preserves the specified features under the proposed formulation.
\item We conduct a comprehensive evaluation demonstrating that our method preserves diverse combinations of features across multiple base compressors while achieving compression ratios and runtimes that are competitive with, and often superior to, existing feature-preserving approaches.
\end{itemize}

\section{Related Work}
\label{sec:related-work}

\para{Lossy compression for scientific data.} 
Modern scientific instruments and simulations generate massive data volumes, often on the order of terabytes \cite{cappello2019use}, posing significant challenges for storage and transmission. To address these issues, scientists rely on data compression, which can be broadly categorized as \emph{lossless} or \emph{lossy}. Lossless compression guarantees exact reconstruction of the original data but typically achieves limited reduction (often under $2\times$~\cite{son2014data}). In contrast, lossy compression permits controlled distortion in exchange for substantially higher compression ratios and is therefore preferred in practice. 

Scientific workflows, however, demand careful control and quantification of compression-induced errors. As a result, modern lossy compressors for scientific data provide strict guarantees on error bounds \cite{liang2023sz3, lindstrom2014fixed, li2023lossy}. These compressors generally fall into two categories: \emph{prediction-based} and \emph{transformation-based}. Prediction-based methods estimate data values via interpolation and encode only the residuals required to satisfy prescribed error bounds. Early work such as ISABELA \cite{lakshminarasimhan2013isabela} employs B-spline interpolation, while more recent approaches like SZ3 \cite{liang2023sz3} combine multiple predictors, including Lorenzo \cite{ibarria2003out}, cubic splines, and linear interpolation.

Recently, deep learning techniques have been increasingly incorporated into prediction-based compression strategies. Implicit Neural Representations (INRs), such as \cite{dupont2022coin, lu2021compressive}, train compact neural networks that take spatial coordinates as input and approximate the underlying data, with the learned network parameters serving as the compressed representation. Although INRs are computationally expensive to train and evaluate, they have gained popularity due to their ability to provide a continuous representation of the data. In contrast, autoencoder-based approaches \cite{liu2021exploring, liu2023high, wang2022locality} learn paired \emph{encoder} and \emph{decoder} networks, where the encoder maps data to a low-dimensional latent space and the decoder reconstructs the data from this representation.

Complementary to prediction-based methods, transformation-based lossy compressors achieve compression by decorrelating data through mathematical transforms. For instance, ZFP \cite{lindstrom2014fixed} partitions data into small blocks and compresses each block independently using exponent alignment, a near-orthogonal transform, and embedded encoding. Other methods exploit tensor decomposition techniques: TTHRESH \cite{ballester2019tthresh} treats the dataset as a tensor and applies singular value decomposition (SVD) for decorrelation, while SPERR \cite{li2023sperr} combines wavelet transforms with SPECK coding.

\para{Derived features in scientific visualization.}~Derived geometric and topological features play a central role in scientific data analysis and visualization. Among these, isocontours are fundamental primitives that provide intuitive representations of level sets of scalar fields. They are widely used in applications ranging from elevation mapping in cartography to exploring slices of deep neural network loss landscapes \cite{chatzimichailidis2019gradvis}. Global statistical summaries also play an important role; for instance, quantiles are commonly used for uncertainty visualization in direct volume rendering \cite{athawale2021direct} and in medical image analysis \cite{yang2020quantile}.

Topological features further provide powerful tools for understanding complex data. Critical points---locations where the gradient of a function vanishes---have been used to track clouds \cite{li2025tracking} and tropical cyclones \cite{yan2024trophy}. Contour trees support the analysis of scalar field topology and have been applied to transfer function design \cite{zhou2009automatic} as well as to domains such as fluid dynamics \cite{aydogan2014characterization} and medicine \cite{aydogan2013analysis, szymczak2010categorical, wang2018fast}. Similarly, Morse--Smale complexes have been used to study molecular structures \cite{bhatia2018topoms} and to derive meaningful visualizations from histopathology images \cite{ahmad2025histomsc}.

\para{Feature-preserving lossy compression.}~A variety of feature-preserving compressors have been developed for scalar field data, each targeting different structural properties. Building on topological abstractions, Soler et al.~\cite{soler2018topologically} preserved the $0$-dimensional persistence diagram. TopoSZ \cite{yan2024toposz} preserves the contour tree up to a user-specified persistence threshold $\varepsilon$ and was subsequently extended by TopoA \cite{gorski2025general}. 
MSz~\cite{li2025msz} preserves the piecewise-linear Morse--Smale segmentation of a scalar field, with further improvements introduced in pMSz~\cite{li2026pmsz}. Similarly, TopoSZp~\cite{agarwal2026toposzp} preserves extrema (minima and maxima) while retaining a large subset of saddle points. Beyond strictly topological features, Jiao et al.~\cite{pu2022toward} proposed a method that preserves quantities of interest (QoIs)---numerical functions evaluated pointwise---as well as contours. These strategies were improved upon by QPET~\cite{lu2025qpet}. Finally, LOPC \cite{fallin2026fast} adopts a local perspective by preserving the \emph{local order} of a scalar field, i.e., maintaining ``less-than'' relationships between neighboring points. Related feature-preserving approaches have also been developed for vector fields \cite{xia2025tspsz, liang2023toward,xia2024preserving} and tensor fields \cite{gorski2026tfz}.
\section{Technical Background}
\label{sec:background}

\subsection{Piecewise-Linear Scalar Fields}

A \emph{piecewise-linear} (PL) scalar field on a simplicial complex $\X$, $f : \X \rightarrow \R$, is defined by assigning scalar values to the vertices of $\X$ and extending $f$ over each simplex via linear interpolation. It is often convenient to assume that $f$ is injective on the vertices of $\X$. To enforce injectivity, ties between vertices $x, y \in \X$ with identical values are resolved using \emph{simulation of simplicity} (SOS) \cite{edelsbrunner1990simulation}.
Several SOS schemes exist; we adopt the implementation of simplicial complexes and SOS provided by the Topology Toolkit~\cite{tierny2018ttk}. Specifically, if the structured grid coordinates of $x$ precede those of $y$ in column-major (Fortran) order, we define $f(x) < f(y)$; otherwise, $f(x) > f(y)$.

For a vertex $x \in \X$, the \emph{star} $\St(x)$ is the set of simplices containing $x$, along with their faces. The \emph{link} $\Lk(x)$ consists of simplices in $\St(x)$ that do not contain $x$. Given a PL scalar field $f$, the \emph{upper link} $\Lk^+(x)$ is the subset of simplices $\sigma \in \Lk(x)$ such that $f(y) > f(x)$ for all vertices $y$ of $\sigma$. The \emph{lower link} $\Lk^-(x)$ is defined analogously.

\begin{figure}[!ht]
\includegraphics[width=1.0\linewidth]{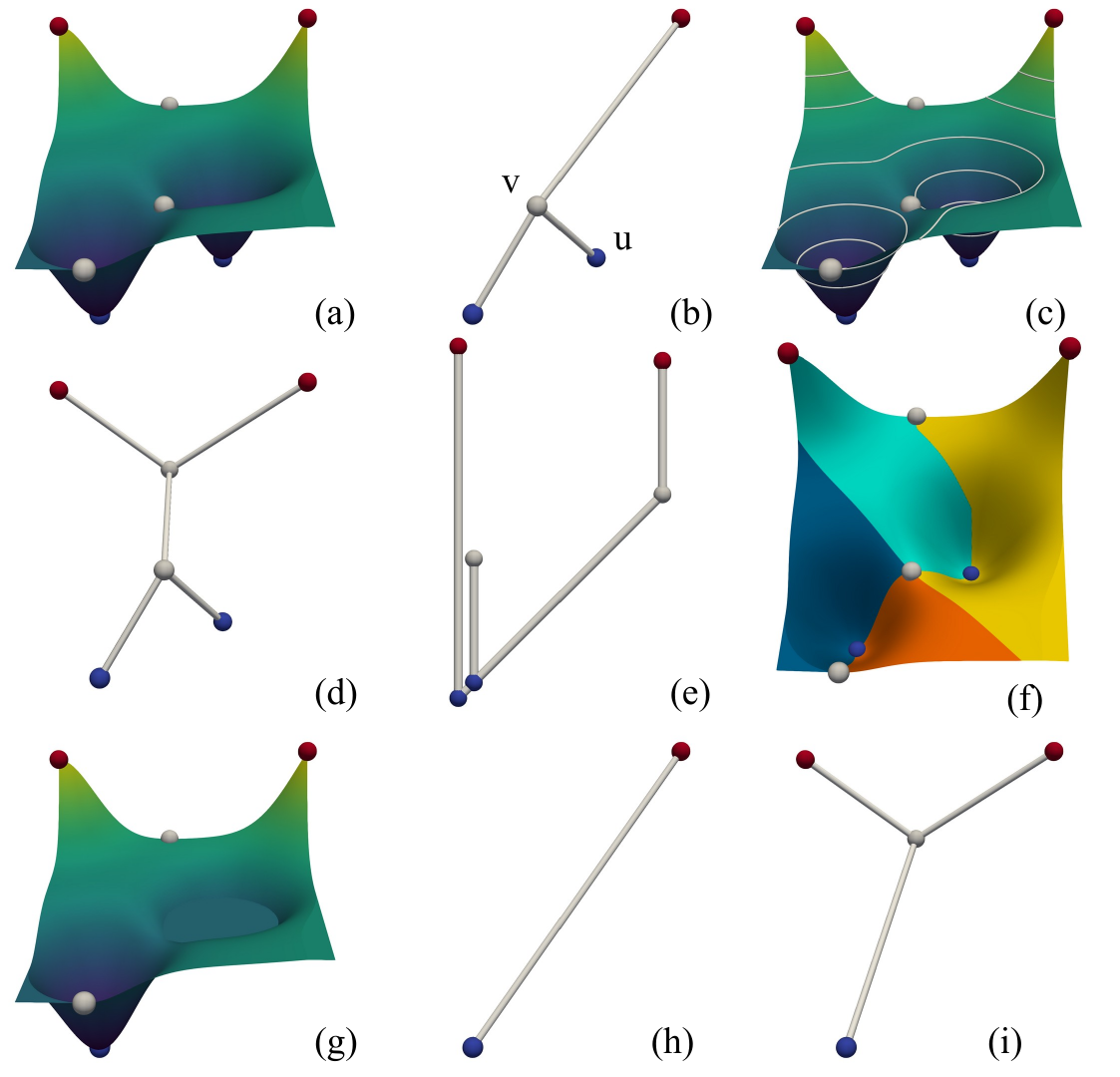}
\vspace{-4mm}
\caption{Various topological descriptors for PL scalar fields. (a) Scalar field with critical points (minima blue, saddles white, maxima red). (b) Merge tree with feature $uv$. (c) Scalar field with selected contours. (d) Contour tree. (e) 0-dimensional persistence diagram: segments denote features; endpoints are colored by criticality types; $y=x$ is the diagonal. (f) Morse--Smale segmentation. (g) Persistence simplification removing $uv$. (h) Simplified merge tree. (i) Simplified contour tree.}
\label{fig:descriptors}
\vspace{-6mm}
\end{figure}

\subsection{Topology of Piecewise-Linear Scalar Fields}

We use {\tool} to preserve several topological features; accordingly, we review common topological descriptors for PL scalar fields. Most of these descriptors originate from Morse theory, where one considers smooth functions. By a slight abuse of notation, we write $f : \M \to \R$ for a smooth function on a manifold and $f : \X \to \R$ for its PL counterpart, where $\X$ is a triangulation of $\M$.

\para{Critical points.} For a smooth function $f:\M \rightarrow \R$, critical points satisfy $\nabla f(x)=0$. An analogous notion exists in the PL setting. Let $f:\X \rightarrow \R$ be a PL scalar field. A vertex $x \in \X$ is \emph{ordinary} if both $\Lk^+(x)$ and $\Lk^-(x)$ are nonempty and simply connected; otherwise, it is \emph{critical}. In 3D, a critical point is \emph{degenerate} if $\Lk^+(x)$ and $\Lk^-(x)$ have more than three total connected components.

Nondegenerate critical points are classified by their \emph{index}, defined as the number of independent directions in $\Lk^-(x)$. In 3D, minima have index $0$, maxima index $3$, and saddles index $1$ or $2$. We visualize a scalar field with its critical points in \cref{fig:descriptors}(a).

\para{Merge trees.} The \emph{sublevel set} at value $t$ is $f^{-1}(-\infty,t] = \{x \in \X : f(x) \le t\}$. As $t$ increases, connected components of sublevel sets are created at minima and merge at saddles, eventually forming a single component at the global maximum. The \emph{merge tree} encodes this evolution, with vertices at critical points and edges tracking the merge of components; see \cref{fig:descriptors}(b).

\para{Persistence diagram.} The evolution of sublevel sets induces a pairing between minima (where components are created) and saddles (where they merge), yielding the 0-dimensional persistence diagram. A component that appears at time $a$ and merges with another component at time $b$ gives rise to a point $(a,b)$ in the diagram with persistence $b-a$, measuring its significance. The diagram also contains analogous components for the evolution of super level sets $f^{-1}[t,\infty)$. A $0$-dimensional persistence diagram is visualized in \cref{fig:descriptors}(e).  

\para{Contours and contour trees.} A \emph{contour} (level set) of $f:\X \rightarrow \R$ at value $t \in \R$ is a connected component of $f^{-1}(t)$. Examples are shown in \cref{fig:descriptors}(c). As $t$ increases, contours may appear, merge, split, or disappear. The \emph{contour tree} records these events (see \cref{fig:descriptors}(d)). 

Components appear and merge at minima and saddles, as captured by the merge tree of $f$, and may split or disappear at saddles and maxima, as captured by the merge tree of $-f$. Accordingly, the contour tree is obtained by combining the merge trees of $f$ and $-f$ \cite{carr2003computing}.

\para{Persistence simplification} removes features whose persistence is below a threshold $\varepsilon$. This corresponds to eliminating branches in the merge trees of $f$ and $-f$ associated with features below the threshold. A simplified contour tree is then obtained by combining the simplified merge trees of $f$ and $-f$. An example of a simplified scalar field is shown in \cref{fig:descriptors}(g), with the corresponding simplified merge and contour trees in \cref{fig:descriptors}(h) and (i).

\para{PL Morse--Smale segmentation.} The \emph{integral lines} of a PL function $f:\X \rightarrow \R$ are discrete monotone paths that follow the local gradient through the vertices. For a vertex $x \in \X$ that is not a maximum, the \emph{forward integral line} from $x$ proceeds to its highest neighbor and ultimately terminates at a maximum. Backward integral lines are defined analogously and terminate at minima. 

The \emph{PL Morse--Smale segmentation} partitions the domain such that two points lie in the same region if their forward and backward integral lines terminate at the same (minimum, maximum) pair. Each region therefore consists of points whose flow originates at a common minimum and terminates at a common maximum; see \cref{fig:descriptors}(f).

\subsection{Quantization in Lossy Compression}

A quantization strategy maps continuous values to a discrete set of representatives to enable compression. Suppose that for some $x \in \X$, a predictive model produces an estimate $g(x)$ of the ground truth value $f(x)$ that may not satisfy a global error bound $\xi$. Quantization-based strategies can then be used to adjust $g(x)$ to a value $f'(x)$ such that $|f'(x) - f(x)| < \xi$.

\begin{figure}
\includegraphics[width=\linewidth]{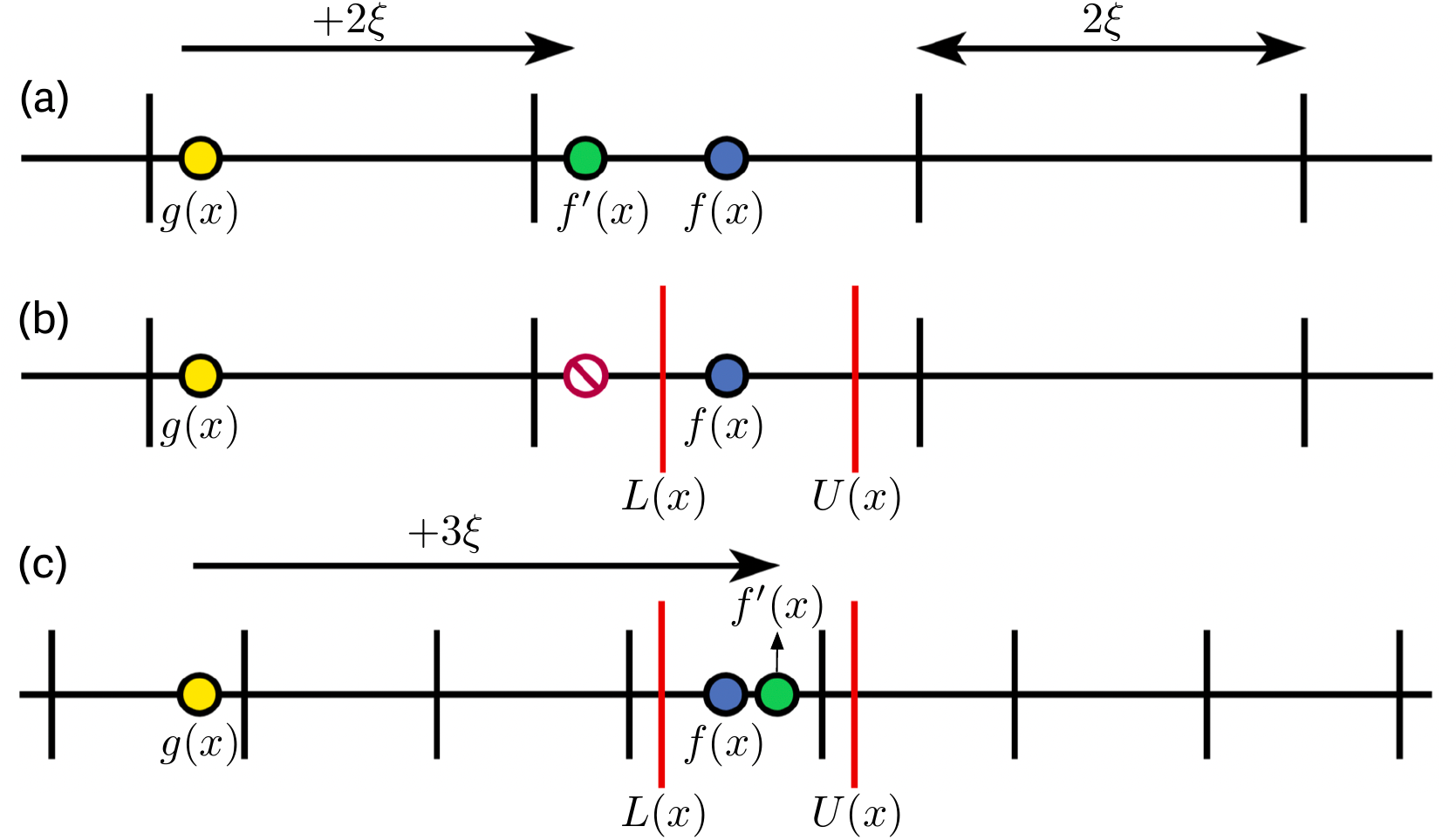}
\label{fig:quantization}
\vspace{-6mm}
\caption{(a) Illustration of standard linear-scaling quantization. Each quantization interval has width $2\xi$. The value $f(x)$ denotes the ground truth, and $g(x)$ is an initial approximation. By shifting $g(x)$ to the right by one quantization level ($n_x \gets +1$), we obtain a final value $f'(x)$ that lies within a distance $\xi$ of $f(x)$.
(b) Example where prescribed upper and lower bounds, $U(x)$ and $L(x)$, cannot be satisfied using standard linear-scaling quantization.
(c) Using logarithmic-scaling quantization, the interval widths are effectively halved. This allows $g(x)$ to be shifted upward by $3\xi$, yielding a final value $f'(x)$ that satisfies the bounds.}
\vspace{-6mm}
\end{figure}

\para{Linear-scaling quantization} was introduced in SZ3 1.4 \cite{tao2017significantly}. It can be viewed as partitioning the real line into intervals of width $2\xi$ centered at $f(x)$. By construction, any value within the same interval as $f(x)$ deviates from it by at most $\xi$. Linear-scaling quantization determines an integer \emph{quantization index} $n_x$ indicating how many such intervals $g(x)$ must be shifted to fall into the interval containing $f(x)$, thereby satisfying the error bound. This process is illustrated in \cref{fig:quantization}(a). The reconstructed value is then given by $f'(x) \gets g(x) + n_x (2\xi)$.

This procedure produces a set of integers $n_x$ for all $x \in \X$. These indices are subsequently compressed using an entropy coding scheme (e.g., Huffman coding) and stored alongside the compressed data. When the predictor ${g(x)}$ is accurate, most $n_x$ values are zero, yielding high compressibility and low storage overhead. 

\para{Logarithmic-scaling quantization.} Consider the additional constraint that each reconstructed value must satisfy pointwise bounds $L(x) < f'(x) < U(x)$ for given functions $L(x), U(x) \in \R$. This setting arises, for example, in topology-preserving compression \cite{gorski2025general, yan2024toposz}.

To address this constraint efficiently, Gorski et al.~\cite{gorski2025general} introduced \emph{logarithmic-scaling quantization}, a multiscale extension of linear-scaling quantization. For each $x \in \X$, the method assigns two integers: a \emph{precision} $p_x$ and a \emph{quantization index} $n_x$. The procedure begins with $p_x \gets 0$ and applies linear-scaling quantization using intervals of width $2\xi / 2^{p_x}$ to obtain $n_x$ and a reconstructed value $f'(x)$. 

If $L(x) < f'(x) < U(x)$, the process terminates; otherwise, the quantization precision is increased and linear-scaling quantization is repeated until the constraint is satisfied.
If no admissible reconstruction is found with $p_x$ at or below a maximum precision $P$, the value $f(x)$ is stored losslessly. For a fixed $p_x$, multiple $n_x$ may yield admissible reconstructions; in this case, the value with minimal magnitude is selected to reduce the total entropy of the set of $\{n_x\}$.

For encoding, a single integer is formed as $a_x \gets n_x \cdot 2^{P - p_x}$, and the sequence $\{a_x\}$ is compressed using an entropy coder. During decompression, the reconstruction is computed as
\[
f'(x) \gets g(x) + \frac{a_x (2\xi)}{2^{P}}.
\]
This construction provably reconstructs the correct value of $f'(x)$. Consequently, only $a_x$ needs to be stored as the quantization index.

\section{Preservation of Features}
\label{sec:features}

{\tool} enforces a global error bound $\xi$ while simultaneously preserving user-defined features whose invariance is characterized by pointwise upper and lower bounds together with a consistent star classification. We formalize the class of features supported by {\tool} in \cref{sec:features-overview}, and subsequently detail how specific feature types can be preserved in \cref{sec:diverse-features}.

\subsection{Overview of Feature Preservation}
\label{sec:features-overview}

{\tool} is designed to preserve a broad class of geometric and topological features that share common sufficient conditions for preservation.

Let $f:\X \rightarrow \R$ denote the original scalar field and $f'$ the reconstructed field obtained after compression and decompression with {\tool}. We model a derived feature as the application of a transformation $\Phi$ to the data, i.e., $\Phi(f)$. For example, if $\Phi$ denotes contour tree extraction, then $\Phi(f)$ is the contour tree of $f$. Under this formulation, our objective is to ensure feature preservation, i.e., $\Phi(f) = \Phi(f')$.

To preserve the feature $\Phi(f)$, we define a star classification function
$C_{\Phi,f} : \X \rightarrow \Z$ that assigns to each point in the domain a categorical label according to the role of its star in the construction of $\Phi(f)$. We also prescribe  pointwise upper and lower bound functions $U_{\Phi,f}, L_{\Phi,f} : \X \rightarrow \R$.

During compression, for each $x \in \X$, {\tool} enforces the following conditions:
\begin{enumerate}[noitemsep]
\item[(a)] $L_{\Phi,f}(x) < f'(x) < U_{\Phi,f}(x)$ or $f'(x) = f(x)$
\item[(b)] $C_{\Phi,f}(x) = C_{\Phi,f'}(x)$.
\end{enumerate}
Condition (a) constrains each reconstructed value to lie within its prescribed bounds (or is stored losslessly if, e.g., $L_{\Phi,f}(x) = U_{\Phi,f}(x)$), while condition (b) ensures consistency of the star classification between the original and reconstructed fields.

The class of features that {\tool} can preserve consists precisely of those whose invariance can be guaranteed under the above conditions. In practice, this class is broad: in particular, all features targeted by existing topology-preserving compressors can be handled within this framework (see \cref{sec:diverse-features}). Moreover, multiple features can be preserved simultaneously by ensuring that these conditions are met for each feature.

\subsection{Preservation of Diverse Features}
\label{sec:diverse-features}

We now describe how different features $\Phi$ can be preserved using {\tool}. To this end, we construct the bound functions $U_{\Phi,f}$ and $L_{\Phi,f}$, together with the star classification function $C_{\Phi,f}$, such that the guarantees enforced by {\tool} imply that $\Phi(f)$ is preserved. In some cases, preservation can be achieved using only the bound functions or the star classification alone; not all features require both conditions. Correctness proofs are provided in the supplement.

\para{Critical points.} In Morse theory, each point in the domain is classified as either regular or critical (i.e., a minimum, maximum, or saddle), referred to as its \emph{criticality type}. For vertices in a structured mesh, the criticality type depends only on the star, $\St(x)$. Accordingly, we define $C_{\Phi,f}(x)$ to be the criticality type of $x$. It follows directly that if $C_{\Phi,f}(x) = C_{\Phi,f'}(x)$ for all $x \in \X$, then the critical points of $f$ coincide with those of $f'$. For efficiency, we compute the criticality type of each point using the lookup table strategy from \cite{lukasczyk2024extreem}, except that we use additional tables for the boundary. We provide details in the supplement.

\para{Simplified merge trees and contour trees.} Denote the contour tree of $f$ by $\T(f)$ and its $\varepsilon$-simplified version by $\T_\varepsilon(f)$. Denote the merge tree of $f$ by $\script{M}(f)$ and its simplified merge tree by $\script{M}_\varepsilon(f)$. Let $\eta \leq \varepsilon$ be a persistence threshold. We construct pointwise upper and lower bounds for each $x \in \X$ using a strategy inspired by TopoA with persistence parameter $\eta$.
Specifically, if $x$ is a vertex of $\T_\eta(f)$, we set
$
L_{\Phi,f}(x)=U_{\Phi,f}(x)=f(x),
$
thereby forcing $f'(x)=f(x)$. Otherwise, $x$ lies on an edge $(a,b)$ of $\script{M}_\eta(f)$ with endpoints $a$ and $b$. Similarly, $x$ lies on an edge $(c,d)$ of $\script{M}_\eta(-f)$. We define $U_{\Phi,f}(x) = \max(f(b),\min(f(d),f(x)))$ and $L_{\Phi,f}(x) = \min(f(c),\max(f(a),f(x)))$.

These bounds guarantee that every branch of $\T_\eta(f)$ is preserved in $\T_\eta(f')$, although $\T_\eta(f')$ may contain additional branches. 
TopoA uses a similar (but not identical) strategy with $\eta=\varepsilon$ and removes additional branches through a progressive bound-tightening procedure.
In contrast, we show that every additional branch has persistence at most $\eta+2\xi$, where $\xi$ is the pointwise error bound (a formal proof is provided in the supplement).
Thus, setting $\eta = \varepsilon - 2\xi$ ensures that any such additional branches have persistence below $\varepsilon$. Consequently, simplifying both $\T_\eta(f)$ and $\T_\eta(f')$ with threshold $\varepsilon$ yields $\T_\varepsilon(f) = \T_\varepsilon(f')$, and no progressive correction is required. This strategy requires $\varepsilon \geq 2\xi$, a condition commonly satisfied in practice. A similar construction can be used to preserve simplified merge trees.

\para{Unsimplified merge trees and contour trees.} To preserve merge and contour trees without persistence simplification, we use the same upper and lower bound construction as in the simplified case, setting $\eta = 0$ so that $T_\eta(f)=T(f)$.  
These bounds ensure that every branch of $\T(f)$ is preserved in $\T(f')$, although $\T(f')$ may contain additional branches with persistence at most $2\xi$.

The number of branches in $\T(f')$ equals the number of extrema of $f'$ minus one. Accordingly, we define $C_{\Phi,f}(x)$ to classify each point as either an extremum or a non-extremum. If $C_{\Phi,f}(x) = C_{\Phi,f'}(x)$ for all $x \in \X$, then the number of extrema---and hence the number of branches---coincides between $f$ and $f'$. It follows that $\T(f') = \T(f)$. 

Merge trees without simplification can be preserved using an analogous construction.

\para{0-dimensional persistence diagram.} For simply connected domains, preservation of the 0-dimensional persistence diagram---whether simplified or unsimplified---follows directly from preservation of the corresponding contour tree.

\para{Piecewise-linear Morse--Smale segmentation.}
The Morse--Smale segmentation is defined by integral lines connecting minima to maxima. At each $x \in \X$, the construction of integral lines is determined by the highest vertex in $\Lk^+(x)$ and the lowest vertex in $\Lk^-(x)$ (assuming they are not empty). We first preserve whether or not each $x \in \X$ is an extremum (i.e., a minimum or maximum). We then define $C_{\Phi,f}(x)$ to encode the identities of the highest and lowest vertices in each of $\Lk^+(x)$ and $\Lk^-(x)$. If $x$ is a maximum, we only encode the smallest neighbor of $x$, and proceed analogously if $x$ is a minimum.

If $C_{\Phi,f}(x) = C_{\Phi,f'}(x)$ for all $x \in \X$, then our construction guarantees that each integral line will travel to the correct neighbors of each $x \in \X$, thus preserving the Morse--Smale segmentation. This strategy is analogous to the C-loops used in pMSz~\cite{li2026pmsz}.

\para{Local order.} For each $x \in \X$, the local order of $\St(x)$ is determined by the partition of $\Lk(x)$ into $\Lk^+(x)$ and $\Lk^-(x)$. We define $C_{\Phi,f}$ accordingly so as to encode this partition, ensuring that the local order is preserved.

\para{Isocontours, sublevel and superlevel sets.} We demonstrate the preservation of isocontours (level sets) for 3D scalar fields; the 2D case follows analogously. In 3D, for a given isovalue $a$, the intersection of an isocontour with a cell $\sigma$ can be detected using the marching cubes algorithm~\cite{lorensen1987marching}, which determines whether each vertex has a function value above or below $a$.

Given a set of isovalues $\{a_i\}$ with $a_1 < a_2 < \cdots < a_n$, we preserve the set of cells intersected by each isocontour using only pointwise upper and lower bounds. For each $x \in \X$ that does not lie on an isocontour, we identify $i$ such that $a_i < f(x) < a_{i+1}$ and set $L_{\Phi,f}(x) \gets a_i$ and $U_{\Phi,f}(x) \gets a_{i+1}$. If $x$ lies on an isocontour, we set $L_{\Phi,f}(x) = f(x) = U_{\Phi,f}(x)$. This construction also preserves the sublevel and superlevel sets for each isovalue.

\para{Quantiles.} Preserving quantiles (e.g., quartiles or percentiles) amounts to preserving the \emph{quantile membership} of each point, i.e., the quantile to which each point belongs relative to the rest of the data. We first compute function value thresholds $\{a_i\}$ that partition the data into quantiles. We then preserve the corresponding sublevel and superlevel sets induced by $\{a_i\}$, ensuring that quantile membership is maintained.

\para{Combinations of features.} Any combination of the above features can be preserved simultaneously by enforcing all corresponding upper and lower bounds and maintaining all associated star classifications. 

\para{Existing feature-preserving capabilities.}
Existing compressors target a limited set of features. Local order is preserved by LOPC~\cite{fallin2026fast}. Critical points are partially preserved by TopoSZp~\cite{agarwal2026toposzp} and fully preserved by LOPC (as a consequence of local order preservation), although {\tool} preserves the critical points without stronger guarantees on local ordering. Simplified contour trees (up to persistence $\varepsilon$) are preserved by TopoSZ~\cite{yan2024toposz} and TopoA~\cite{gorski2025general}, but no method guarantees preservation of the full (unsimplified) contour tree. The piecewise-linear Morse--Smale segmentation is preserved by MSz~\cite{li2025msz} and pMSz~\cite{li2026pmsz}. Isocontours are preserved by QPET \cite{lu2025qpet} and the compressor of Jiao et al. \cite{pu2022toward}. To our knowledge, no existing compressor explicitly targets the preservation of quantiles, nor supports the simultaneous preservation of multiple feature types among features listed above.

\section{Core Compression Algorithm}
\label{sec:pipeline-utilities}

{\tool} is designed as an augmentation layer that operates on top of an existing lossy compressor, referred to as the \textit{base compressor}, similar in spirit to topology-preserving compressors such as TopoA~\cite{gorski2025general}. {\tool} applies corrections to the output of the base compressor to enforce preservation of a user-defined feature $\Phi$. These corrections are stored in a compact representation and included alongside the base compressor output in the final compressed file.

We provide an overview of the pipeline in \cref{sec:algorithm-pipeline}. We give details of the most intricate component, iterative star correction, in \cref{sec:algorithm-star}, followed by how it is specifically applied for individual features in \cref{sec:algorithm-bad-points}. We describe the process of extending {\tool} to preserve new features in the supplement. As before, let $f$ denote the original data and $f'$ the reconstructed data. {\tool} preserves the feature $\Phi$ while enforcing a pointwise error bound $\xi$, i.e., it guarantees that $\Phi(f) = \Phi(f')$ and, for all $x \in \X$, $|f(x) - f'(x)| \leq \xi$.

\subsubsection{Pipeline}
\label{sec:algorithm-pipeline}

We give an overview of each step of {\tool}.

\para{\underline{Step 1: Base compressor}.} The data is compressed and decompressed using the base compressor. The decompressed output is called the \textit{intermediate data}.

\para{\underline{Step 2: Bound computation}.} The pointwise upper and lower bounds $U_{\Phi,f}$ and $L_{\Phi,f}$ are computed.

\para{\underline{Step 3: Initial quantization}.} Logarithmic scaling quantization is used to ensure that the upper and lower bounds, as well as the pointwise error bound $\xi$, are maintained.

\para{\underline{Step 4: Iterative star correction.}} We iterate over each vertex $x \in \X$ and inspect $\St(x)$. We check whether $C_{\Phi,f}(x) = C_{\Phi,f'}(x)$. If the classifications do not match, some vertices will be labeled as \textit{bad}, meaning that they have contributed to the misclassification. If no vertices are bad, this indicates that $C_{\Phi,f}(x) = C_{\Phi,f'}(x)$. The error bounds of a subset of the bad vertices are tightened and quantization is reapplied (see \cref{sec:algorithm-star} for details). This process is repeated until there are no bad vertices, meaning that $\St(x)$ is correctly classified. 

Since modifying a vertex value may affect the classification of neighboring stars that have already been processed, multiple passes are required to ensure that all stars are correctly classified.

\para{\underline{Step 5: Lossless compression}.}~The quantization numbers from logarithmic-scaling quantization are compressed using Huffman coding. The quantization numbers, lossless values, and output of the base compressor are combined into a tar and compressed using XZ~\cite{collin2025xzutils}.

\subsubsection{Iterative Star Correction}
\label{sec:algorithm-star}

During star correction, we check whether $C_{\Phi,f}(x) = C_{\Phi,f'}(x)$ for each $x \in \X$. Rather than directly computing the values of $C_{\Phi,f}(x)$ and $C_{\Phi,f'}(x)$, each point is directly labeled as good or bad. If all points are good, this indicates that $C_{\Phi,f}(x) = C_{\Phi,f'}(x)$. The exact strategy for labeling points as good versus bad varies by feature (see \cref{sec:algorithm-bad-points}). If any points are bad, then we update values in $\St(x)$ as follows. 

Initially, each vertex $x \in \X$ is assigned a \textit{precision} $p_x \gets 0$. If points $y \in \St(x)$ are marked as bad, those with the smallest precision $p_y$ are selected and their precision is incremented by one. These vertices are then re-quantized using logarithmic-scaling quantization with upper bound $\min(U_{\Phi,f}(x),\, f(x) + \xi \cdot 2^{-p_x})$ and lower bound $\max(L_{\Phi,f}(x),\, f(x) - \xi \cdot 2^{-p_x})$. If there are multiple valid choices of quantization intervals, the quantization interval that yields the most accurate reconstructed value is selected. If $p_x$ surpasses the maximum allowable precision, $f(x)$ is stored losslessly. After the values of the bad points are updated, each point in the star is re-classified as either good or bad. This process is repeated until there are no bad points. To ensure termination, we ensure that if any points $\{y_i\}$ are marked as bad, at least one $y_i$ satisfies $f(y_i) \not = f'(y_i)$.

Updating a vertex’s value may affect the classification of neighboring stars. Therefore, the star correction step is performed iteratively, and the entire dataset is scanned multiple times. To improve efficiency, we maintain a set of \textit{dirty} vertices $x$ for which $C_{\Phi,f}(x) \neq C_{\Phi,f'}(x)$ may hold. In each iteration, only dirty vertices are processed. Initially, all vertices are marked as dirty. Whenever a vertex $x$ is updated, both $x$ and its neighbors are marked as dirty for the next iteration. The dirty set is first represented as a bit array; once the number of dirty vertices falls below $20\%$ of the domain, we switch to a sparse representation storing only the indices of dirty vertices. Switching to a sparse representation significantly speeds up star correction when the number of dirty points is small.

\subsubsection{Identification of Bad Points}
\label{sec:algorithm-bad-points}

Let $x \in \X$. For each feature, we classify points in $\St(x)$ as good or bad. We describe this process for all features whose preservation involves star correction.

\para{Critical points.}
We ensure that the criticality type of $x$ agrees for $f$ and $f'$. If the types coincide, then all points in $\St(x)$ are labeled good. Otherwise, $x$ is labeled bad. In this case, for each $y \in \St(x)$, if $y \in \Lk^+(x)$ but should lie in $\Lk^-(x)$ (or vice versa), then $y$ is also labeled bad.

\para{Unsimplified merge and contour trees.}
We preserve whether $x$ is a minimum or maximum using a similar strategy to preserving critical points.

\para{Local order.} We preserve membership in $\Lk^+(x)$ and $\Lk^-(x)$. For each $y \in \Lk(x)$, if $y \in \Lk^+(x)$ but it should lie in $\Lk^-(x)$ (or vice versa), then $y$ is labeled bad. If any such $y$ is bad, then $x$ is also labeled bad.

\para{Morse--Smale segmentation.} We preserve: (a) whether $x$ is a minimum or maximum, and (b) the largest vertex in $\Lk^+(x)$ (if any) and the smallest vertex in $\Lk^-(x)$ (if any). Condition (a) is handled similarly to the critical point case. For condition (b), if $x$ is not a maximum in $f$, let $y$ and $y'$ be the largest elements of $\Lk^+(x)$ under $f$ and $f'$, respectively. If $y \neq y'$, then $x$, $y$, and $y'$ are labeled bad. An analogous rule is applied to $\Lk^-(x)$.

\section{Experimental Results}
\label{sec:results}

We begin with an overview of our experimental setup in \cref{sec:results-overview}. We then evaluate reconstruction quality, including comparisons with both classical and topology-preserving compressors, in \cref{sec:results-quality}. Finally, we analyze runtime performance in \cref{sec:results-time}.

\subsection{Overview of Experiments}
\label{sec:results-overview}

\para{Compression tasks.} We evaluate the capabilities of {\tool} to preserve seven distinct features. We consider four features for which dedicated feature-preserving compressors already exist: the simplified contour tree, Morse--Smale segmentation, local order, and isocontours.

We also examine two additional features for which no dedicated feature-preserving compressor has been explicitly developed, namely critical points and percentiles. Finally, we assess the ability of {\tool} to preserve multiple features simultaneously by jointly preserving the simplified contour tree and Morse--Smale segmentation. We refer to this setting as the preservation of the \emph{combined feature}.

\para{Comparison with existing compressors.} We evaluate {\tool} in conjunction with four base compressors: ZFP, SZ3, SPERR, and HPEZ, denoted as \tool~(ZFP), \tool~(SZ3), etc. We also include results for each base compressor without augmentation by {\tool}. We evaluate HPEZ with maximum optimizations enabled. We also evaluate the neural compressor Neurcomp \cite{lu2021compressive}; its poor performance is reported in the supplement.

Whenever possible, we compare against state-of-the-art compressors designed to preserve specific features, namely TopoA for the simplified contour tree, pMSz for Morse--Smale segmentation, QPET for isocontours, and LOPC for local order. 
We also compare against TopoSZp; since it supports only 2D data, the corresponding results are provided in the supplement. Overall, {\tool} achieves better reconstruction quality and feature preservation, while TopoSZp is significantly faster. 

Like {\tool}, both TopoA, pMSz, and QPET augment a base compressor. For TopoA, we evaluate the same base compressors (ZFP, SZ3, SPERR, and HPEZ), denoted as TopoA~(ZFP) etc. pMSz only supports ZFP and SZ3, while QPET only supports SPERR and HPEZ. We evaluate pMSz and QPET on the base compressors that they support.

\para{Evaluation metrics.} We measure bit-rate (lower is better) and peak signal-to-noise ratio PSNR (higher is better). Bit-rate measures the average number of bits per data point in the compressed representation. It is computed by dividing the number of bits in the compressed file by the number of vertices in $\X$. PSNR quantifies reconstruction quality, where higher values correspond to lower error.
If $I$ is the range of $f$, and MSE is the mean squared error between $f$ and $f'$, then PSNR is defined as:
$
\text{PSNR} = 10 \log_{10}\left( I^2 / \text{MSE} \right).
$
We also evaluate compression and decompression throughput. They are computed by dividing the ground truth file size by the compression time and decompression time, respectively. For {\tool}, reported times exclude Python startup overhead for the dispatcher and begin when the compressor is invoked.

\para{Datasets.} We evaluate {\tool} on nine datasets spanning a range of scientific domains. These datasets are listed in \cref{tab:data} and described in detail in the supplement. 

\begin{table}[!t]
\begin{center}
\caption{Datasets along with their dimensions and file sizes.}
\vspace{-3mm}
\begin{tabular}{c|cc}\hline
Dataset & Dimensions & File Size (MB) \\ \hline
QMCPACK & $69 \times 69 \times 115$ & 4.4 \\
Viscous Fingering & $128 \times 128 \times 128$ & 16.8 \\
Tangaroa & $300 \times 180 \times 120$ & 25.9 \\
Earthquake & $375 \times 188 \times 50$ & 28.2 \\
Ionization & $310 \times 128 \times 128$ & 40.6 \\
Isabel & $500 \times 500 \times 90$ & 90.0 \\
Miranda & $384 \times 384 \times 256$ & 302.0 \\
S3D & $500 \times 500 \times 500$ & 1000.0 \\
SCALE-LETKF & $1200 \times 1200 \times 98$ & 1129.0 \\
\hline
\end{tabular}
\label{tab:data}
\end{center}
\vspace{-10mm}
\end{table}

\para{Parameter configurations.} The global error bound $\xi$ is specified as a fraction of the function range; for example, $\xi = 0.1$ corresponds to $10\%$ of the range. For each feature, we evaluate five values of $\xi$. 

In topology-preserving compression, it has been observed that, for larger error bounds, decreasing the global error bound can, counterintuitively, lead to higher compression ratios \cite{gorski2026tfz, gorski2025general, li2025msz, xia2025tspsz}. A possible explanation is that larger error bounds introduce greater distortion, necessitating more corrections to restore topological structure. 
Such error bounds are suboptimal, since a smaller error bound generally yields both a smaller compressed file size and higher reconstruction quality, whereas the error parameter is intended to control the trade-off between these quantities. To ensure a fair comparison, we evaluate all compressors using only parameter settings that are optimal for each method. 
To do so, we choose $\xi$ sufficiently small for each feature so that all compressors exhibit increasing compression ratios as $\xi$ increases. This choice of bounds is consistent with prior work \cite{gorski2025general, gorski2026tfz}. To maintain this monotonic trend, we use dataset-specific error bounds in some cases and omit outlier configurations that violate this behavior. The values of $\xi$ and the rationale for their selection are provided in the supplement.

For isocontour preservation, we select three isovalues per dataset that partition the volume into approximately four equal regions. For the Miranda dataset, these isovalues produced poor visualizations due to their placement in a noisy, low-gradient region. Similarly, the isocontours of the Viscous Fingering dataset obstructed much of the view. We therefore adjust the isovalues for these datasets to yield more meaningful isocontours. The exact values are given in the supplement. 

For simplified contour tree preservation, we use a persistence threshold of $\varepsilon=0.04$ (4\% of the function range), consistent with the TopoA experiments of Gorski et al.~\cite{gorski2025general}. The supplement includes a sensitivity analysis of the evaluation metrics with respect to the persistence threshold $\varepsilon$ (for simplified contour tree preservation) and the number of isovalues (for isosurface preservation).

\para{Correctness.} 
We have empirically verified that each feature is preserved in every trial. A formal proof of correctness is provided in the supplement.

\para{System.}~All experiments are conducted on a laptop equipped with an Intel Core i9-12900HK CPU, an NVIDIA GeForce RTX 3080 Ti laptop GPU, and 32~GB of RAM.

\para{Additional experiments.} The supplement includes an ablation study comparing linear- and logarithmic-scaling quantization and evaluating the contribution of each pipeline component to feature preservation and compression overhead. 

\begin{figure*}[!ht]
\includegraphics[width=\linewidth]{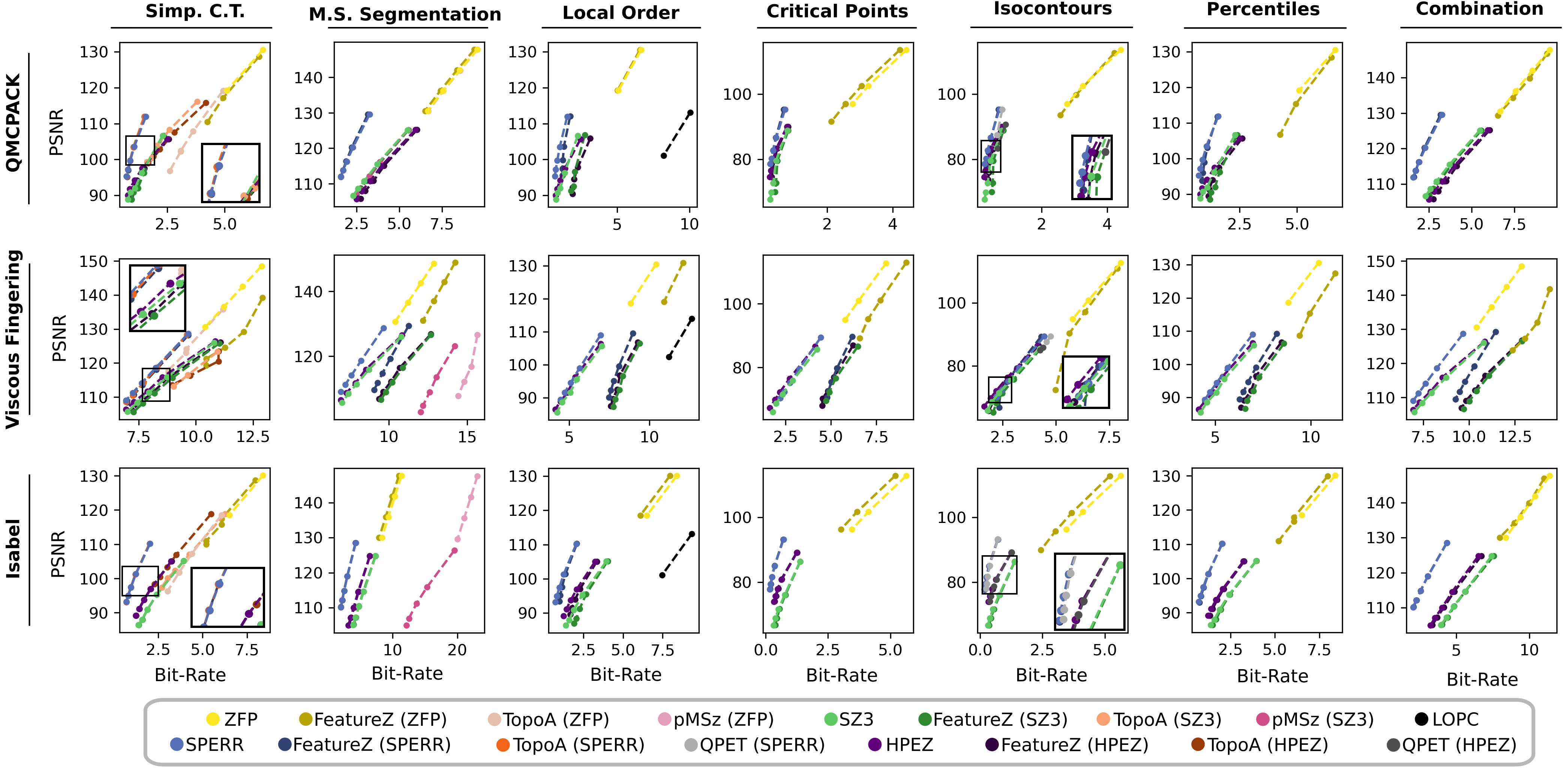}
\vspace{-6mm}
\caption{Rate-distortion curves with selected zoomed-in views illustrating the trade-off between bit-rate and PSNR for features preserved by {\tool}. Results are shown for three datasets---QMCPACK, Viscous Fingering, Isabel---across seven features; the full set of plots is provided in the supplement. For reference, we include the base compressors ZFP, SZ3, TTHRESH, and HPEZ although these do not guarantee feature preservation. Where applicable, we also compare against state-of-the-art topology-preserving compressors, including TopoA, pMSz, QPET, and LOPC.}
\label{fig:rate-distortion}
\vspace{-2mm}
\end{figure*}

\begin{figure*}[!ht]
\begin{center}
\includegraphics[width=0.8\linewidth]{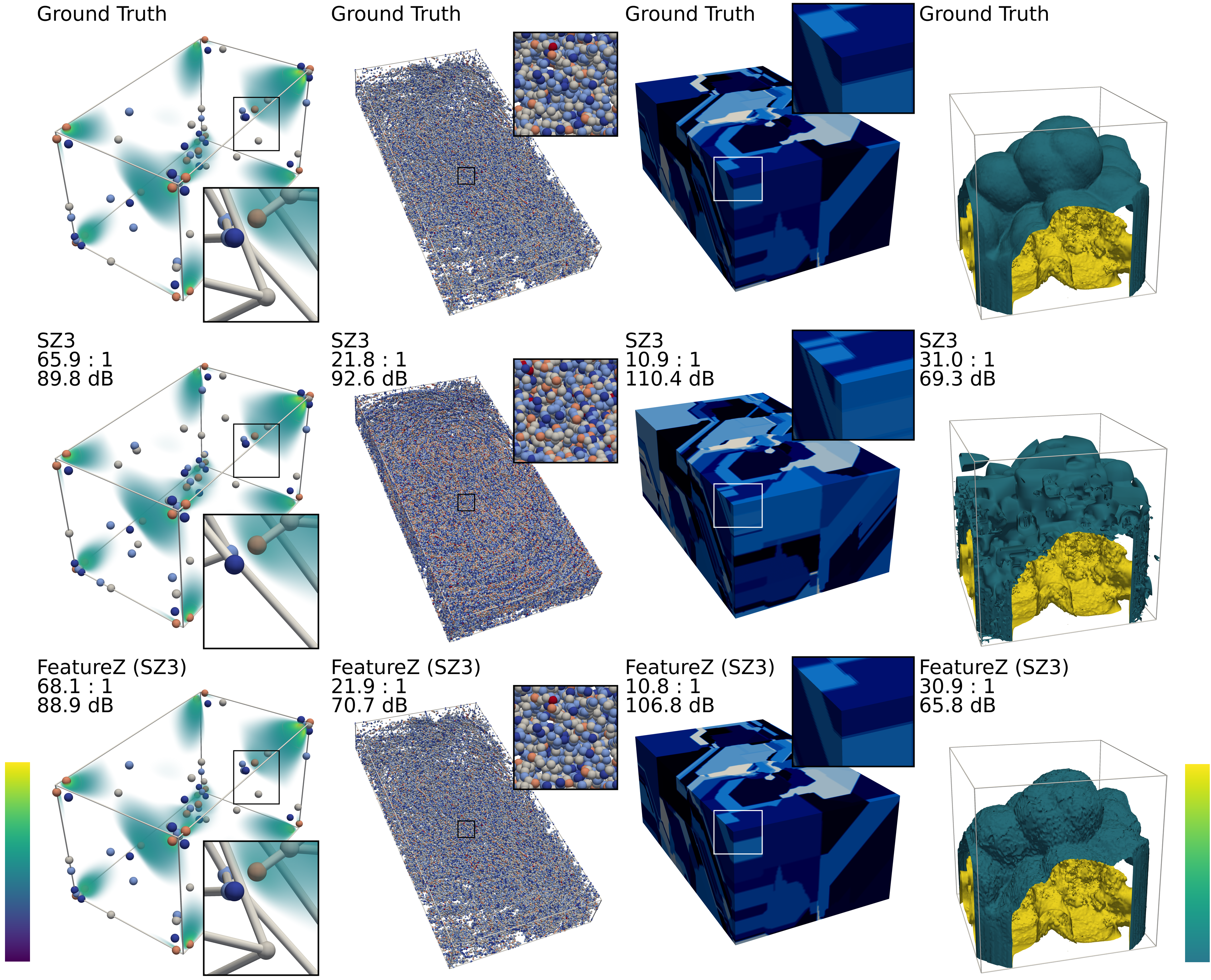}
\vspace{-3mm}
\caption{Volumetric datasets visualized with geometric or topological features: SZ3 compared with {\tool} augmenting SZ3 at similar compression ratios. Zoomed-in views are provided to show extra detail. Top: ground truth. Center: datasets compressed with SZ3. Bottom: datasets compressed with {\tool} (SZ3).
From left to right, the datasets are: QMCPACK, visualized with its simplified contour tree (up to persistence $\varepsilon = 0.04$); Earthquake, with critical points colored by criticality type; Tangaroa, with Morse--Smale segmentation; and Viscous Fingering, with isocontours. In each case, SZ3 introduces errors in the extracted features, whereas {\tool} (SZ3) preserves them without errors.}
\label{fig:equiv-eb}
\vspace{-6mm}
\end{center}
\end{figure*}

\subsection{Reconstruction Quality}
\label{sec:results-quality}

In \cref{fig:rate-distortion}, we present the trade-off between bit-rate and PSNR for three representative datasets across seven features (see the supplement for the complete results). We consider 17 compressor configurations, including (i) general-purpose, unaugmented base compressors such as ZFP, SZ3, SPERR, and HPEZ, (ii) topology-preserving compressors such as TopoA, pMSz, QPET, and LOPC, and (iii) our framework {\tool} augmenting the base compressors. 
We include curves for the unaugmented base compressors under the same error bounds, although these methods do not explicitly preserve any features. 

Across the three base compressor choices, a consistent trend emerges for {\tool}: the SPERR-based variant generally achieves the best performance, while the ZFP-based variant performs the worst.

\para{Comparison with feature-preserving compressors.}
When preserving identical features, {\tool} consistently achieves reconstruction quality that is competitive with, and often superior to, existing feature-preserving compressors (\cref{fig:rate-distortion}). Observations are summarized below:
\begin{itemize}[noitemsep,leftmargin=*]
\item \textbf{Overall performance.} {\tool} attains reconstruction quality comparable to TopoA and QPET and generally outperforms both pMSz and LOPC across features.
    
\item \textbf{Simplified contour tree.} {\tool} (SZ3), {\tool} (SPERR), and {\tool} (HPEZ) generally achieve reconstruction quality on par with or better than TopoA (SZ3) TopoA (SPERR), and TopoA (HPEZ), respectively. {\tool} (ZFP) performs slightly worse than TopoA (ZFP).

\item \textbf{Morse--Smale segmentation.} Compared to pMSz, both {\tool} (SZ3) and {\tool} (ZFP) consistently achieve comparable or better reconstruction quality than pMSz (SZ3) and pMSz (ZFP). {\tool} (SPERR) consistently outperforms both pMSz variants.
    
\item \textbf{Local order.} All variants of {\tool} consistently outperform LOPC, showing strong compression in local order preservation.

\item \textbf{Isocontours.} {\tool} (HPEZ) and {\tool} (SPERR) perform similarly to QPET (HPEZ) and QPET (SPERR).
\end{itemize}
We note that pMSz and LOPC natively support different simplicial complex structures for 3D grids than {\tool}, which may contribute to differences in compression overhead.

In summary, {\tool} achieves reconstruction quality that is highly competitive with existing feature-preserving compression methods. Moreover, while existing approaches typically target a single feature, {\tool} supports a broader range of features while delivering comparable or improved reconstruction quality.

\begin{table}[!h]
\vspace{-2mm}
\caption{Error percentages for SZ3 on different datasets at compression ratios comparable to {\tool} while preserving a given feature.
C.T.: number of vertices that are false positives, false negatives, or misclassified in the reconstructed contour tree divided by number of vertices in the ground truth.
M.S.: fraction of points incorrectly classified in the Morse--Smale segmentation.
L.O.: fraction of points for which the local order among neighbors is violated.
C.P.: number of critical points whose type (ordinary, minimum, saddle, maximum, degenerate) is misclassified divided by number of critical points in the ground truth.
Perc.: fraction of points assigned to incorrect percentiles.}
\vspace{-3mm}
\resizebox{\linewidth}{!}{\begin{tabular}{l|ccccc} \hline
dataset & C.T. & M.S. & L.O. & C.P. & Perc. \\ \hline
QMCPACK & {$125.7\%$} & $20.4\%$ & $5.3\%$ & $108.7\%$ & $2.3\%$\\
Viscous Fingering & {$1.6\%$} & $51.5\%$ & $31.0\%$ & $30.5\%$ & $26.8\%$\\
Tangaroa & {$11.0\%$} & $34.1\%$ & $2.7\%$ & $141.9\%$ & $3.1\%$\\
Earthquake & {$7.1\%$} & $96.3\%$ & $21.2\%$ & $198.5\%$ & $12.0\%$\\
Ionization & {$12.3\%$} & $87.4\%$ & $48.0\%$ & $280.5\%$ & $38.6\%$\\
Isabel & {$153.3\%$} & $17.8\%$ & $5.7\%$ & $222.6\%$ & $1.0\%$\\
Miranda & {$83.3\%$} & $99.5\%$ & $69.2\%$ & $100.6\%$ & $61.9\%$\\
S3D & {$34.7\%$} & $9.8\%$ & $1.9\%$ & $83.2\%$ & $0.9\%$\\
SCALE-LETKF & {$47.3\%$} & $59.6\%$ & $17.4\%$ & $103.0\%$ & $20.9\%$\\
\hline
\end{tabular}}
\vspace{-4mm}
\label{tab:equivalent-eb}
\end{table}

\para{Comparison with base compressors.} In \cref{fig:rate-distortion}, we observe that {\tool} typically introduces only modest overhead relative to its unaugmented base compressors. In some cases, the overhead is negligible (e.g., isocontours for Viscous Fingering), while in others it can reach up to approximately $2\times$. Overall, this overhead is comparable to, or lower than, that of other topology-preserving methods \cite{gorski2026tfz}. Notably, {\tool} (ZFP) occasionally outperforms ZFP alone, as observed for the Isabel dataset when preserving critical points. We attribute this to the combination of {\tool}'s low overhead and the additional post hoc lossless compression (XZ) applied to ZFP outputs.

Although {\tool} introduces additional overhead relative to base compressors, reducing the error bound of a base compressor to match the compressed file size achieved by {\tool} leads to substantial feature distortion. In \cref{tab:equivalent-eb}, we report feature errors for SZ3 under these conditions, which are often significant. For example, on the Earthquake dataset, even when SZ3 achieves a compression ratio comparable to {\tool}, more than $96\%$ of points are assigned to incorrect Morse--Smale segmentation regions after decompression. Notably, cases with large SZ3 errors often coincide with scenarios in which {\tool} incurs higher overhead, reflecting the additional corrections required to preserve features under greater distortion. Corresponding results for ZFP and SPERR, with associated error bounds, are in the supplement, along with a detailed comparison of the trade-off between bit-rate and feature error. Visual comparisons in \cref{fig:teaser} and \cref{fig:equiv-eb} further demonstrate that {\tool} preserves a wide range of features, whereas SZ3 introduces substantial distortions even at comparable compression ratios.

\begin{table}[!ht]
\caption{Runtimes for {\tool} simultaneously preserving the simplified contour tree ($\varepsilon = 0.04$) and Morse--Smale segmentation across datasets. Times are in seconds.
Bounds: upper/lower bound computation; B.C.C.: base compressor (compression); B.C.D.: base compressor (decompression); Quant.: pointwise quantization; Star: iterative star correction; L.L.: lossless compression.}
\vspace{-3mm}
\resizebox{\linewidth}{!}{\begin{tabular}{l|cccccc|c} \hline
dataset & Bounds & B.C.C. & B.C.D. & Quant. & Star & L.L. & total \\ \hline
QMCPACK & 0.46 & 0.03 & 0.01 & 0.00 & 0.06 & 0.06 & 0.61 \\
Viscous Fingering & 1.64 & 0.06 & 0.07 & 0.02 & 0.55 & 1.66 & 4.00 \\
Tangaroa & 1.86 & 0.12 & 0.09 & 0.01 & 1.02 & 0.67 & 3.77 \\
Earthquake & 1.82 & 0.10 & 0.09 & 0.01 & 1.69 & 1.00 & 4.72 \\
Ionization & 1.52 & 0.12 & 0.08 & 0.02 & 3.39 & 1.02 & 6.15 \\
Isabel & 6.16 & 0.46 & 0.36 & 0.19 & 2.21 & 4.48 & 13.87 \\
Miranda & 56.16 & 0.78 & 0.67 & 0.13 & 37.77 & 73.22 & 168.73 \\
S3D & 36.18 & 2.49 & 1.61 & 0.52 & 13.07 & 13.54 & 67.44 \\
SCALE-LETKF & 57.49 & 3.16 & 3.62 & 0.59 & 49.19 & 55.68 & 169.76 \\
\hline
\end{tabular}}
\label{tab:compress-disaggregation}
\vspace{-3mm}
\end{table}

\begin{table}[!ht]
\caption{Compression throughput (MB/s) for {\tool} (SZ3) across datasets and tasks, along with TopoA (SZ3), pMSz (SZ3), and LOPC. C.T.: contour tree; M.S.: Morse--Smale segmentation; L.O.: local order; C.P.: critical points; I.S.: isocontours; Per.: percentiles; Com.: combined feature (simplified contour tree and Morse--Smale segmentation); T.A.: TopoA (SZ3). TopoA failed on SCALE-LETKF for unknown reasons.}
\vspace{-3mm}
\resizebox{\linewidth}{!}{\begin{tabular}{l|ccccccc|ccc}\hline
dataset & C.T. & M.S. & L.O. & C.P. & I.S. & Per. & Com. & T.A. & pMSz & LOPC\\ \hline
QMCPACK & {7.8} & 28.9 & 25.5 & 40.9 & 64.3 & 49.4 & 6.0 & 2.6 & 165.2 & 168.0\\
Viscous Fingering & {8.1} & 6.6 & 6.4 & 5.7 & 62.1 & 12.3 & 3.9 & 2.0 & 157.6 & 5.6\\
Tangaroa & {10.6} & 12.8 & 11.9 & 19.3 & 46.5 & 29.0 & 6.9 & 2.3 & 56.4 & 102.3\\
Earthquake & {12.9} & 9.0 & 11.6 & 10.1 & 84.5 & 41.0 & 6.0 & 2.6 & 207.4 & 150.8\\
Ionization & {21.7} & 5.9 & 7.5 & 7.0 & 53.7 & 43.3 & 6.6 & 2.7 & 119.2 & 57.0\\
Isabel & {11.4} & 12.3 & 14.4 & 26.9 & 61.8 & 25.3 & 4.7 & 2.0 & 126.9 & 104.7\\
Miranda & {5.1} & 1.9 & 2.8 & 2.6 & 91.7 & 4.0 & 1.8 & 0.9 & 63.3 & 65.1\\
S3D & {19.4} & 14.7 & 26.9 & 63.0 & 125.7 & 54.4 & 14.8 & 4.3 & 36.3 & 139.0\\
SCALE-LETKF & {14.4} & 6.4 & 11.4 & 12.1 & 117.1 & 33.0 & 6.7 & -- & 158.2 & 206.5\\
\hline
\end{tabular}}
\vspace{-3mm}
\label{tab:compression-throughput}
\end{table}

\begin{table}[!ht]
\caption{Decompression throughput (MB/s) for {\tool} (SZ3) across datasets and tasks, along with TopoA (SZ3), and LOPC. Labels are consistent with \cref{tab:compression-throughput}.}
\vspace{-3mm}
\resizebox{\linewidth}{!}{\begin{tabular}{l|ccccccc|cc}\hline
dataset & C.T. & M.S. & L.O. & C.P. & I.S. & Per. & Com. & T.A. & LOPC\\ \hline
QMCPACK & {147.6} & 143.3 & 141.2 & 174.5 & 168.2 & 157.8 & 131.4 & 84.0 & 470.9\\
Viscous Fingering & {135.8} & 98.0 & 96.5 & 109.4 & 171.2 & 110.4 & 96.5 & 121.4 & 442.8\\
Tangaroa & {108.5} & 79.3 & 76.1 & 126.5 & 133.0 & 100.2 & 81.9 & 108.5 & 353.3\\
Earthquake & {190.4} & 93.9 & 111.3 & 129.3 & 199.2 & 126.3 & 97.5 & 195.1 & 492.8\\
Ionization & {192.7} & 110.9 & 100.8 & 158.9 & 189.7 & 151.3 & 118.7 & 206.0 & 472.0\\
Isabel & {117.2} & 74.5 & 102.9 & 141.7 & 144.7 & 113.4 & 79.1 & 115.3 & 367.9\\
Miranda & {226.0} & 68.9 & 65.3 & 70.4 & 225.4 & 74.3 & 69.5 & 268.9 & 506.0\\
S3D & {246.3} & 207.8 & 219.5 & 257.0 & 266.0 & 254.3 & 206.3 & 332.1 & 525.9\\
SCALE-LETKF & {174.2} & 85.8 & 110.9 & 145.3 & 248.2 & 139.7 & 105.1 & -- & 542.1\\
\hline
\end{tabular}}
\label{tab:decompression-throughput}
\vspace{-7mm}
\end{table}

\subsection{Analysis of Compression Time}
\label{sec:results-time}

\para{Runtime analysis for combined features.} In \cref{tab:compress-disaggregation}, we report a breakdown of compression times for {\tool} (SZ3) across several datasets while preserving the combined feature. For each dataset, we use the largest error bound for this setting. Results for other base compressors are similar and provided in the supplement. We focus on the combined feature because it exercises all stages of the {\tool} pipeline, including both pointwise quantization and iterative star correction.

Overall, the majority of runtime is spent on bound computation, iterative star correction, and lossless compression. Notably, execution of the base compressor accounts for only a small fraction ($<5\%$) of the total runtime, consistent other feature-preserving compressors \cite{gorski2025general}. In most cases, the dominant cost arises from computing upper and lower bounds, largely due to the expense of constructing the simplified contour tree. We consider this overhead justified, as computing a feature is a natural prerequisite for preserving it during compression.

The star correction step also contributes noticeable overhead. This is partly due to its iterative nature, although in practice the first iteration dominates runtime, with subsequent iterations affecting only a small subset of points. Additionally, per-point operations are relatively expensive: {\tool} must access values from 3D neighborhoods, while the data is stored in a linear layout, leading to poor cache locality. Furthermore, star classification itself is computationally intensive.

Finally, lossless compression incurs nontrivial overhead. We use XZ due to its typically high compression ratios compared to alternatives such as gzip or bzip2 \cite{farncomb2015gzip}. However, this choice comes with increased runtime; faster lossless compressors could reduce this overhead at the expense of lower compression efficiency.

\para{Throughput analysis.} We report compression throughput for {\tool} (SZ3) in \cref{tab:compression-throughput}. For each (dataset, feature) pair, we report the runtime corresponding to the largest tested error bound. A detailed analysis of the trade-off between error bound and runtime is provided in the supplement. Throughput results for {\tool} with other base compressors are similar; no single base compressor consistently yields the fastest performance. Corresponding tables for other base compressors are included in the supplement. We also report throughput for TopoA (SZ3), pMSz (SZ3), and LOPC. 
We do not report QPET throughput in the main paper because it does not support SZ3 as a base compressor. Corresponding results for SPERR and HPEZ are provided in the supplement. Overall, QPET achieves throughput comparable to {\tool}.

Throughput varies significantly across features, reflecting differences in preservation strategies. For example, percentile and isocontour preservation achieve higher throughput, as they do not require star correction (see~\cref{tab:compress-disaggregation}) and have inexpensive bound computations. In contrast, simplified contour tree preservation avoids star correction but requires computing the contour tree, which is costly. Critical point and Morse--Smale segmentation preservation incur lower throughput due to the expense of iterative star correction.

Across base compressors (ZFP, SZ3, SPERR, and HPEZ), no configuration consistently minimizes runtime for {\tool}. This is because the base compressor contributes only a small fraction of the total runtime, limiting the impact of its intrinsic speed. Instead, runtime is largely influenced by the distortions introduced by the base compressor, which determine the extent of corrective processing. Since these distortions affect geometric and topological structures in complex ways, their impact on runtime is difficult to predict.

Compared to other topology-preserving compressors, {\tool} is roughly $5 \times$ faster than TopoA but slower than pMSz and LOPC when preserving the same features. pMSz is GPU-based, and thus expected to outperform {\tool}. {\tool} achieves throughput comparable to other feature-preserving methods not evaluated here \cite{gorski2026tfz, yan2024toposz}.

Decompression throughput for {\tool} SZ3) are reported in \cref{tab:decompression-throughput}. As expected, decompression is significantly faster than compression. There is no clear correlation between compression and decompression throughput, suggesting a complex relationship between feature-preserving edits and reconstruction cost. Compared to existing methods, {\tool} achieves decompression performance comparable to TopoA and slower than LOPC. Decompression times for pMSz are not available, as its experimental implementation does not report them.

\para{Asymptotic analysis.} Because {\tool} supports multiple feature-preservation strategies, a single universal asymptotic bound is not meaningful. Instead, we provide a general characterization. Let $n$ denote the number of points in the domain $\X$. Let the cost of bound computation over the entire domain be $\bigO(T_b(n))$, and the cost of star correction per point be $\bigO(T_s(n))$. Let $k$ denote the maximum precision for logarithmic-scaling quantization (we use $k=8$), so that each point may be updated at most $k+1$ times. Let $S$ denote the largest possible number of points in a star. Let $M$ denote the maximum number of first- and second-order neighbors of any $x \in \X$.

We now analyze the major pipeline stages. Initial quantization visits each point once and may update it up to $k+1$ times, yielding $\bigO(nk)$. During star correction, each point in the star may have values in its star updated up to $k$ times, and its star is re-evaluated whenever a neighboring star (within first- or second-order adjacency) is modified. Consequently, a star may be updated up to $\bigO(MSk)$ times, giving a total cost of $\bigO(nSMk\,T_s(n))$. In practice, both the base compressor and XZ operate in $\bigO(n)$ time. 
Combining these components, the overall runtime of {\tool} is
$
\bigO\!\big(T_{b}(n) + nSMk\,T_s(n)\big).
$

We now instantiate this bound for specific features. For contour tree preservation, bound computation is dominated by contour tree construction, which is $\bigO(n\log n)$. For preserving $c$ isocontours, bound computation requires sorting isovalues and locating each point within the corresponding interval, yielding $T_{b}(n) = \bigO(n\log c + c\log c)$; in practice, since $c \ll n$, this reduces to $\bigO(n)$. Percentile preservation similarly yields $T_{b}(n) = \bigO(n)$ in practice, while other features have $T_{b}(n) = \bigO(1)$. For all implemented features, $T_s(n) \in \bigO(1)$, and $k$, $S$, and $M$ are small constants. Therefore, the overall runtime is $\bigO(n\log n)$ for contour tree preservation and $\bigO(n)$ for all other features.

\section{Limitations and Discussions}
\label{sec:limitations}

{\tool} has several limitations, summarized below.

\para{Storage overhead.}~{\tool} introduces additional storage overhead (up to $\sim 2\times$), comparable to other topology-preserving methods. For a fixed compression ratio, preserving topological correctness may outweigh marginal PSNR gains.

\para{Runtime overhead.} {\tool} incurs substantial runtime overhead compared to the underlying base compressors. While its throughput is comparable to some feature-preserving methods, it remains slower overall and is outperformed by certain specialized (e.g. GPU-based) approaches. We believe there is significant potential for optimization via parallel or GPU implementations. In particular, the star correction step is largely pointwise and could be parallelized with minimal synchronization. Bound computation depends on the user-defined feature: simplified contour tree construction is largely parallelizable \cite{lukasczyk2024extreem}, while percentile and isocontour computations are highly parallel.

\para{Visual artifacts.} We observe noticeable artifacts in isocontour preservation. Since points on an isocontour may shift by up to the length of a grid cell, these artifacts diminish at higher resolutions, which are typical in practice.

\para{Design choices.} {\tool} relies on a number of heuristics that promote computational efficiency but may not always produce optimal augmentation configurations. In particular, feature preservation is formulated entirely through local guarantees—namely, pointwise upper and lower bounds and consistent star classifications—allowing features to be preserved using only neighborhood-level operations and avoiding expensive global geometric or topological computations. This locality, however, comes at the cost of flexibility. Global optimization strategies may identify alternative augmentations that achieve lower storage overheads while preserving the same features. Furthermore, some applications may benefit from allowing small geometric perturbations in preserved features, such as slight shifts in critical point locations, in exchange for improved compression ratios. Such trade-offs are difficult to express within our current framework. Likewise, our star-correction algorithm employs a greedy strategy that corrects vertices one at a time. While computationally efficient, this approach may introduce more corrections than necessary. Alternative strategies that jointly optimize corrections across multiple vertices could potentially achieve the same preservation guarantees with lower storage overhead.

\para{Feature specification.} Defining new features in {\tool} can be nontrivial, particularly for novice users, due to the underlying definitions of star classification. We describe the technical details in the supplement. Preserving a feature requires specifying sufficient conditions within the framework supported by {\tool}; this process can be unintuitive and does not encompass all features of interest.

\para{Parameter selection.} Like other feature-preserving compressors, some error settings for {\tool} may be suboptimal, and decreasing the error bound can sometimes produce both smaller compressed files and higher reconstruction quality. Since there is no systematic strategy for selecting optimal parameter values, initial parameter choices may yield suboptimal performance.

\para{Degenerate cases.} In certain adversarial cases, {\tool} may require near-lossless storage to preserve the requested features, for example when preserving the unsimplified contour tree of an exceptionally topologically complex dataset.

\para{Preserving arbitrary features.} In preliminary experiments, we have explored a variant of {\tool} that preserves arbitrary features that can be defined \emph{simply} as categorical maps $\X \rightarrow \Z$ without star classification; for instance, assigning each point a Morse--Smale region ID. However, without explicit modeling of pointwise dependencies (e.g.,~via star classification), this variant requires  repeated re-computation of the feature, leading to substantial overhead. This underscores the importance of a pipeline that enables efficient, pointwise verification. Developing more intuitive mechanisms for feature specification remains an important direction for future work.

\section{Conclusion}
\label{sec:conclusion}

We introduce {\tool}, a compression framework that preserves a broad class of features in lossy compression while maintaining a user-specified error bound $\xi$. For a given feature $\Phi$, {\tool} first computes pointwise upper and lower bounds enforced via quantization, and then applies an iterative procedure to ensure correct star classification with respect to $\Phi$. Our experiments demonstrate that {\tool} preserves diverse features of interest while incurring moderate overhead in file size and runtime, comparable to existing feature-preserving compressors. Looking ahead, the approach may extend to other data types, such as vector and tensor fields, and could benefit substantially from parallel and GPU implementations.

\acknowledgments{
This work was supported in part by the U.S. National Science Foundation (NSF) under Grants OAC-2313122, OAC-2313123, OAC-2313124, and OAC-2628473. We used ChatGPT 5.2 for grammar correction and language polishing, and ChatGPT 5.2 Codex for limited coding assistance, including integrating base compressors and reformatting arguments for certain internal APIs.}

\bibliographystyle{abbrv-doi-hyperref-narrow}
\bibliography{refs-featurez.bib}

@article{lan2024topological,
	author = {Lan, Fangfei and Gamelin, Brandi and Yan, Lin and Wang, Jiali and Wang, Bei and Guo, Hanqi},
	doi = {10.1111/cgf.15084},
	journal = {Computer Graphics Forum},
	number = {3},
	pages = {e15084:1--e15084:12},
	title = {Topological Characterization and Uncertainty Visualization of Atmospheric Rivers},
	volume = {43},
	year = {2024}}

@article{li2026spatiotemporal,
	author = {Mingzhe Li and Peer Nowack and Bei Wang},
    doi = {10.1109/TVCG.2026.3694459},
	journal = {IEEE Transactions on Visualization and Computer Graphics},
    number = {6},
    pages = {4742--4752},
	title = {Spatiotemporal Detection and Uncertainty Visualization of Atmospheric Blocking Events},
    volume = {32},
	year = {2026}}

@article{yan2024trophy,
	author = {Yan, Lin and Guo, Hanqi and Peterka, Thomas and Wang, Bei and Wang, Jiali},
	doi = {10.1109/TVCG.2023.3326905},
	journal = {IEEE Transactions on Visualization and Computer Graphics},
	number = {1},
	pages = {1249--1259},
	title = {{TROPHY}: A Topologically Robust Physics-Informed Tracking Framework for Tropical Cyclones},
	volume = {30},
	year = {2024}}

@article{bhatia2018topoms,
	author = {Bhatia, Harsh and Gyulassy, Attila G. and Lordi, Vincenzo and Pask, John E. and Pascucci, Valerio and Bremer, Peer-Timo},
	doi = {10.1002/jcc.25181},
	journal = {Journal of Computational Chemistry},
	number = {16},
	pages = {936--952},
	title = {{TopoMS}: Comprehensive topological exploration for molecular and condensed-matter systems},
	volume = {39},
	year = {2018}}

@article{tierny2018ttk,
	author = {Tierny, Julien and Favelier, Guillaume and Levine, Joshua A. and Gueunet, Charles and Michaux, Michael},
	doi = {10.1109/TVCG.2017.2743938},
	journal = {IEEE Transactions on Visualization and Computer Graphics},
	number = {1},
	pages = {832--842},
	title = {The {Topology ToolKit}},
	volume = {24},
	year = {2018}}

@article{gorski2025general,
	author = {Gorski, Nathaniel and Liang, Xin and Guo, Hanqi and Yan, Lin and Wang, Bei},
	doi = {10.1109/TVCG.2025.3567054},
	journal = {IEEE Transactions on Visualization and Computer Graphics},
	number = {6},
	pages = {3693--3705},
	title = {A General Framework for Augmenting Lossy Compressors with Topological Guarantees},
	volume = {31},
	year = {2025}}

@article{yan2024toposz,
	author = {Yan, Lin and Liang, Xin and Guo, Hanqi and Wang, Bei},
	doi = {10.1109/TVCG.2023.3326920},
	journal = {IEEE Transactions on Visualization and Computer Graphics},
	number = {1},
	pages = {1302--1312},
	title = {{TopoSZ}: Preserving Topology in Error-Bounded Lossy Compression},
	volume = {30},
	year = {2024}}

@article{li2025msz,
	author = {Li, Yuxiao and Liang, Xin and Wang, Bei and Qiu, Yongfeng and Yan, Lin and Guo, Hanqi},
	doi = {10.1109/TVCG.2024.3456337},
	journal = {IEEE Transactions on Visualization and Computer Graphics},
	number = {1},
	pages = {130--140},
	title = {{MSz}: An Efficient Parallel Algorithm for Correcting {Morse--Smale} Segmentations in Error-Bounded Lossy Compressors},
	volume = {31},
	year = {2025}}

@article{edelsbrunner1990simulation,
	author = {Edelsbrunner, Herbert and M{\"u}cke, Ernst Peter},
	doi = {10.1145/77635.77639},
	journal = {ACM Transactions on Graphics},
	number = {1},
	pages = {66--104},
	title = {Simulation of Simplicity: A Technique to Cope with Degenerate Cases in Geometric Algorithms},
	volume = {9},
	year = {1990}}

@article{carr2003computing,
	author = {Carr, Hamish and Snoeyink, Jack and Axen, Ulrike},
	doi = {10.1016/S0925-7721(02)00093-7},
	journal = {Computational Geometry},
	number = {2},
	pages = {75--94},
	title = {Computing Contour Trees in All Dimensions},
	volume = {24},
	year = {2003}}

@article{cappello2019use,
	author = {Franck Cappello and Sheng Di and Sihuan Li and Xin Liang and Ali Murat Gok and Dingwen Tao and Chun Hong Yoon and Xin-Chuan Wu and Yuri Alexeev and Frederic T Chong},
	doi = {10.1177/1094342019853336},
	journal = {The International Journal of High Performance Computing Applications},
	number = {6},
	pages = {1201--1220},
	title = {Use cases of lossy compression for floating-point data in scientific data sets},
	volume = {33},
	year = {2019}}

@article{son2014data,
	author = {Son, Seung Woo and Chen, Zhengzhang and Hendrix, William and Agrawal, Ankit and Liao, Wei-keng and Choudhary, Alok},
	doi = {10.14529/jsfi140205},
	journal = {Supercomputing Frontiers and Innovations},
	number = {2},
	pages = {76--88},
	title = {Data Compression for the Exascale Computing Era --- Survey},
	volume = {1},
	year = {2014}}

@article{liang2023sz3,
	author = {Liang, Xin and Zhao, Kai and Di, Sheng and Li, Sihuan and Underwood, Robert and Gok, Ali M. and Tian, Jiannan and Deng, Junjing and Calhoun, Jon C. and Tao, Dingwen and Chen, Zizhong and Cappello, Franck},
	doi = {10.1109/TBDATA.2022.3201176},
	journal = {IEEE Transactions on Big Data},
	number = {2},
	pages = {485--498},
	title = {{SZ3}: A Modular Framework for Composing Prediction-Based Error-Bounded Lossy Compressors},
	volume = {9},
	year = {2023}}

@article{lindstrom2014fixed,
	author = {Lindstrom, Peter},
	doi = {10.1109/TVCG.2014.2346458},
	journal = {IEEE Transactions on Visualization and Computer Graphics},
	number = {12},
	pages = {2674--2683},
	title = {Fixed-Rate Compressed Floating-Point Arrays},
	volume = {20},
	year = {2014}}

@article{athawale2021direct,
	author = {Athawale, Tushar M. and Ma, Bo and Sakhaee, Elham and Johnson, Chris R. and Entezari, Alireza},
	doi = {10.1109/TVCG.2020.3030394},
	journal = {IEEE Transactions on Visualization and Computer Graphics},
	number = {2},
	pages = {1797--1807},
	title = {Direct Volume Rendering with Nonparametric Models of Uncertainty},
	volume = {27},
	year = {2021}}

@article{yang2020quantile,
	author = {Hojin Yang and Veerabhadran Baladandayuthapani and Arvind U.K. Rao and Jeffrey S. Morris},
	doi = {10.1080/01621459.2019.1609969},
	journal = {Journal of the American Statistical Association},
	number = {529},
	pages = {90--106},
	title = {Quantile Function on Scalar Regression Analysis for Distributional Data},
	volume = {115},
	year = {2020}}

@article{zhou2009automatic,
	author = {Zhou, Jianlong and Takatsuka, Masahiro},
	doi = {10.1109/TVCG.2009.120},
	journal = {IEEE Transactions on Visualization and Computer Graphics},
	number = {6},
	pages = {1481--1488},
	title = {Automatic Transfer Function Generation Using Contour Tree Controlled Residue Flow Model and Color Harmonics},
	volume = {15},
	year = {2009}}

@article{aragoncalvo2010spine,
	author = {Arag{\'o}n-Calvo, Miguel A. and Platen, Erwin and {van de Weygaert}, Rien and Szalay, Alexander S.},
	doi = {10.1088/0004-637X/723/1/364},
	journal = {The Astrophysical Journal},
	number = {1},
	pages = {364--382},
	title = {The spine of the cosmic web},
	volume = {723},
	year = {2010}}

@article{gorski2026tfz,
	author = {Gorski, Nathaniel and Liang, Xin and Guo, Hanqi and Wang, Bei},
	doi = {10.1109/TVCG.2025.3634844},
	journal = {IEEE Transactions on Visualization and Computer Graphics},
	number = {1},
	pages = {527--537},
	title = {{TFZ}: Topology-Preserving Compression of {2D} Symmetric and Asymmetric Second-Order Tensor Fields},
	volume = {32},
	year = {2026}}

@article{wang2018fast,
	author = {Wang, Lei and Guo, Quan and Zhao, Jianqiao and Zhang, Shengnan and Yang, Lisu},
	doi = {10.1166/jmihi.2018.2454},
	journal = {Journal of Medical Imaging and Health Informatics},
	number = {7},
	pages = {1451--1455},
	title = {The Fast Contour Tree-Based Medical Volume Rendering Method},
	volume = {8},
	year = {2018}}

@article{ahmad2025histomsc,
	author = {Zahoor Ahmad and Khaled Al-Thelaya and Mahmood Alzubaidi and Faaiz Joad and Nauman Ullah Gilal and William Mifsud and Sabri Boughorbel and Giovanni Pintore and Enrico Gobbetti and Jens Schneider and Marco Agus},
	doi = {10.1016/j.compbiomed.2025.109991},
	journal = {Computers in Biology and Medicine},
	pages = {109991:1--109991:20},
	title = {{HistoMSC}: Density and topology analysis for {AI}-based visual annotation of histopathology whole slide images},
	volume = {190},
	year = {2025}}

@article{aydogan2014characterization,
	author = {Aydogan, Dogu Baran and Hyttinen, Jari},
	doi = {10.1098/rsif.2013.1042},
	journal = {Journal of the Royal Society Interface},
	number = {95},
	pages = {20131042:1--20131042:10},
	title = {Characterization of microstructures using contour tree connectivity for fluid flow analysis},
	volume = {11},
	year = {2014}}

@article{lukasczyk2024extreem,
	author = {Lukasczyk, Jonas and Will, Michael and Wetzels, Florian and Weber, Gunther H. and Garth, Christoph},
	doi = {10.1109/TVCG.2023.3326526},
	journal = {IEEE Transactions on Visualization and Computer Graphics},
	number = {1},
	pages = {1085--1094},
	title = {{ExTreeM}: Scalable Augmented Merge Tree Computation via Extremum Graphs},
	volume = {30},
	year = {2024}}

@article{liang2023toward,
	author = {Liang, Xin and Di, Sheng and Cappello, Franck and Raj, Mukund and Liu, Chunhui and Ono, Kenji and Chen, Zizhong and Peterka, Tom and Guo, Hanqi},
	doi = {10.1109/TVCG.2022.3214821},
	journal = {IEEE Transactions on Visualization and Computer Graphics},
	number = {12},
	pages = {5434--5450},
	title = {Toward Feature-Preserving Vector Field Compression},
	volume = {29},
	year = {2023}}

@article{ibarria2003out,
	author = {Ibarria, Lawrence and Lindstrom, Peter and Rossignac, Jarek and Szymczak, Andrzej},
	doi = {10.1111/1467-8659.00681},
	journal = {Computer Graphics Forum},
    number = {3},
	pages = {343--348},
	title = {Out-of-core compression and decompression of large \emph{n}-dimensional scalar fields},
	volume = {22},
	year = {2003}}

@article{lu2021compressive,
	author = {Lu, Yuzhe and Jiang, Kairong and Levine, Joshua A and Berger, Matthew},
	doi = {10.1111/cgf.14295},
	journal = {Computer Graphics Forum},
	number = {3},
	pages = {135--146},
	title = {Compressive Neural Representations of Volumetric Scalar Fields},
	volume = {40},
	year = {2021}}

@article{liu2023high,
	author = {Liu, Tong and Wang, Jinzhen and Liu, Qing and Alibhai, Shakeel and Lu, Tao and He, Xubin},
	doi = {10.1109/TBDATA.2021.3066151},
	journal = {IEEE Transactions on Big Data},
	number = {1},
	pages = {22--36},
	title = {High-Ratio Lossy Compression: Exploring the Autoencoder to Compress Scientific Data},
	volume = {9},
	year = {2023}}

@article{wang2022locality,
	author = {Nan Wang and Tong Liu and Jinzhen Wang and Qing Liu and Shakeel Alibhai and Xubin He},
	doi = {10.1016/j.jnca.2022.103452},
	journal = {Journal of Network and Computer Applications},
    pages = {103452:1--103452:12},
	title = {Locality-based transfer learning on compression autoencoder for efficient scientific data lossy compression},
	volume = {205},
	year = {2022}}

@article{ballester2019tthresh,
	author = {Ballester-Ripoll, Rafael and Lindstrom, Peter and Pajarola, Renato},
	doi = {10.1109/TVCG.2019.2904063},
	journal = {IEEE Transactions on Visualization and Computer Graphics},
	number = {9},
	pages = {2891--2903},
	title = {{TTHRESH}: Tensor compression for multidimensional visual data},
	volume = {26},
	year = {2020}}

@article{dupont2022coin,
	author = {Emilien Dupont and Hrushikesh Loya and Milad Alizadeh and Adam Golinski and Yee Whye Teh and Arnaud Doucet},
	journal = {Transactions on Machine Learning Research},
	title = {{COIN}++: Neural Compression Across Modalities},
	url = {https://openreview.net/forum?id=NXB0rEM2Tq},
	year = {2022}}

@article{pu2022toward,
	author = {Jiao, Pu and Di, Sheng and Guo, Hanqi and Zhao, Kai and Tian, Jiannan and Tao, Dingwen and Liang, Xin and Cappello, Franck},
	doi = {10.14778/3574245.3574255},
	journal = {Proceedings of the VLDB Endowment},
	number = {4},
	pages = {697--710},
	title = {Toward Quantity-of-Interest Preserving Lossy Compression for Scientific Data},
	volume = {16},
	year = {2022}}

@article{lu2025qpet,
	author = {Liu, Jinyang and Jiao, Pu and Zhao, Kai and Liang, Xin and Di, Sheng and Cappello, Franck},
	doi = {10.14778/3742728.3742739},
	journal = {Proceedings of the VLDB Endowment},
	number = {8},
	pages = {2440--2453},
	title = {{QPET}: A Versatile and Portable Quantity-of-Interest-Preservation Framework for Error-Bounded Lossy Compression},
	volume = {18},
	year = {2025}}

@article{maack2024parallel,
	author = {Maack, Robin G. C. and Lukasczyk, Jonas and Tierny, Julien and Hagen, Hans and Maciejewski, Ross and Garth, Christoph},
	doi = {10.1109/TVCG.2023.3261981},
	journal = {IEEE Transactions on Visualization and Computer Graphics},
	number = {4},
	pages = {1942--1955},
	title = {Parallel Computation of Piecewise Linear Morse--Smale Segmentations},
	volume = {30},
	year = {2024}}

@article{kim2018qmcpack,
	author = {Jeongnim Kim and Andrew D. Baczewski and Todd D. Beaudet and Anouar Benali and M. Chandler Bennett and Mark A. Berrill and Nick S. Blunt and Edgar Josu{\'e} Landinez Borda and Michele Casula and David M. Ceperley and Simone Chiesa and Bryan K Clark and Raymond C. Clay III and Kris T. Delaney and Mark Dewing and Kenneth P. Esler and Hongxia Hao and Olle Heinonen and Paul R. C. Kent and Jaron T. Krogel and Ilkka Kyl{\"a}np{\"a}{\"a} and Ying Wai Li and M. Graham Lopez and Ye Luo and Fionn D. Malone and Richard M. Martin and Amrita Mathuriya and Jeremy {McMinis} and Cody A. Melton and Lubos Mitas and Miguel A. Morales and Eric Neuscamman and William D. Parker and Sergio D. Pineda Flores and Nichols A. Romero and Brenda M. Rubenstein and Jacqueline A. R. Shea and Hyeondeok Shin and Luke Shulenburger and Andreas F. Tillack and Joshua P. Townsend and Norm M. Tubman and Brett Van Der Goetz and Jordan E. Vincent and D. {ChangMo} Yang and Yubo Yang and Shuai Zhang and Luning Zhao},
	doi = {10.1088/1361-648X/aab9c3},
	journal = {Journal of Physics: Condensed Matter},
	number = {19},
    pages = {195901:1--195901:29},
	title = {{QMCPACK}: an open source \emph{ab initio} quantum {Monte Carlo} package for the electronic structure of atoms, molecules and solids},
	volume = {30},
	year = {2018}}

@article{kent2020qmcpack,
	author = {Paul. R. C. Kent and Abdulgani Annaberdiyev and Anouar Benali and M. Chandler Bennett and Edgar Josu{\'e} Landinez Borda and Peter Doak and Hongxia Hao and Kenneth D. Jordan and Jaron T. Krogel and Ilkka Kyl{\"a}np{\"a}{\"a} and Joonho Lee and Ye Luo and Fionn D. Malone and Cody A. Melton and Lubos Mitas and Miguel A. Morales and Eric Neuscamman and Fernando A. Reboredo and Brenda Rubenstein and Kayahan Saritas and Shiv Upadhyay and Guangming Wang and Shuai Zhang and Luning Zhao},
	doi = {10.1063/5.0004860},
	journal = {The Journal of Chemical Physics},
	number = {17},
    pages = {174105:1--174105:20},
	title = {{QMCPACK}: Advances in the development, efficiency, and application of auxiliary field and real-space variational and diffusion quantum {Monte Carlo}},
	volume = {157},
	year = {2020}}

@article{popinet2004experimental,
	author = {Popinet, St{\'e}phane and Smith, Murray and Stevens, Craig},
	doi = {10.1175/1520-0426(2004)021<1575:EANSOT>2.0.CO;2},
	journal = {Journal of Atmospheric and Oceanic Technology},
	number = {10},
	pages = {1575--1589},
	title = {Experimental and numerical study of the turbulence characteristics of airflow around a research vessel},
	volume = {21},
	year = {2004}}

@article{olsen2008terashake2,
	author = {Olsen, Kim Bak and Day, Steven M. and Minster, Jean Bernard and Cui, Yifeng and Chourasia, Amit and Okaya, David and Maechling, Philip and Jordan, Tom},
	doi = {10.1785/0120070148},
	journal = {Bulletin of the Seismological Society of America},
	number = {3},
	pages = {1162--1185},
	title = {{TeraShake2}: Spontaneous rupture simulations of {$M_w$} 7.7 earthquakes on the southern {San Andreas} fault},
	volume = {98},
	year = {2008}}

@article{pont2021wasserstein,
	author = {Pont, Mathieu and Vidal, Jules and Delon, Julie and Tierny, Julien},
	doi = {10.1109/TVCG.2021.3114839},
	journal = {IEEE Transactions on Visualization and Computer Graphics},
	number = {1},
	pages = {291--301},
	title = {{Wasserstein} distances, geodesics and barycenters of merge trees},
	volume = {28},
	year = {2022}}

@article{whalen2008ionization,
	author = {Whalen, Daniel and Norman, Michael L.},
	doi = {10.1086/524400},
	journal = {The Astrophysical Journal},
	number = {2},
	pages = {664--675},
	title = {Ionization Front Instabilities in Primordial {H II} Regions},
	volume = {673},
	year = {2008}}

@article{lien2017near,
	author = {Lien, Guo-Yuan and Miyoshi, Takemasa and Nishizawa, Seiya and Yoshida, Ryuji and Yashiro, Hisashi and Adachi, Sachiho A. and Yamaura, Tsuyoshi and Tomita, Hirofumi},
	doi = {10.2151/sola.2017-001},
	journal = {Scientific Online Letters on the Atmosphere},
	pages = {1--6},
	title = {The near-real-time {SCALE-LETKF} system: A case of the {September} 2015 {Kanto-Tohoku} heavy rainfall},
	volume = {13},
	year = {2017}}

@article { hurrell2013community,
      author = {James W. Hurrell and M. M. Holland and P. R. Gent and S. Ghan and Jennifer E. Kay and P. J. Kushner and J.-F. Lamarque and W. G. Large and D. Lawrence and K. Lindsay and W. H. Lipscomb and M. C. Long and N. Mahowald and D. R. Marsh and R. B. Neale and P. Rasch and S. Vavrus and M. Vertenstein and D. Bader and W. D. Collins and J. J. Hack and J. Kiehl and S. Marshall},
      doi = {10.1175/BAMS-D-12-00121.1},
      journal = {Bulletin of the American Meteorological Society},
      number = {9},
      pages = {1339--1360},
      title = {The Community Earth System Model: A Framework for Collaborative Research},
      volume = {94},
      year = {2013}
}

@incollection{treichler2017s3d,
    address = {New York},
	author = {Treichler, Sean and Bauer, Michael and Bhagatwala, Ankit and Borghesi, Giulio and Sankaran, Ramanan and Kolla, Hemanth and {McCormick}, Patrick S. and Slaughter, Elliott and Lee, Wonchan and Aiken, Alex and Chen, Jacqueline H.},
	booktitle = {Exascale Scientific Applications},
    doi = {10.1201/b21930-12},
    editor = {Tjerk P. Straatsma and Katerina B. Antypas and Timothy J. Williams},
	pages = {257--278},
	publisher = {Chapman and Hall/CRC},
	title = {{S3D-Legion}: An exascale software for direct numerical simulation of turbulent combustion with complex multicomponent chemistry},
	year = {2017}}

@inproceedings{szymczak2010categorical,
    address = {Berlin},
	author = {Szymczak, Andrzej},
    booktitle = {Topological Methods in Data Analysis and Visualization},
	doi = {10.1007/978-3-642-15014-2_17},
    pages = {205--216},
	publisher = {Springer},
	title = {A categorical approach to contour, split and join trees with application to airway segmentation},
	year = {2010}}

@inproceedings{rosen2021using,
	address = {Cham},
	author = {Rosen, Paul and Seth, Anil and Mills, Elisabeth and Ginsburg, Adam and Kamenetzky, Julia and Kern, Jeff and Johnson, Chris R. and Wang, Bei},
	booktitle = {Topological Methods in Data Analysis and Visualization VI},
    doi = {10.1007/978-3-030-83500-2_6},
	editor = {Hotz, Ingrid and Bin Masood, Talha and Sadlo, Filip and Tierny, Julien},
	pages = {87-108},
	publisher = {Springer International Publishing},
	title = {Using Contour Trees in the Analysis and Visualization of Radio Astronomy Data Cubes},
	year = {2021}}

@inproceedings{tao2017significantly,
    address = {New York},
	author = {Tao, Dingwen and Di, Sheng and Chen, Zizhong and Cappello, Franck},
	booktitle = {2017 IEEE International Parallel and Distributed Processing Symposium (IPDPS)},
	doi = {10.1109/IPDPS.2017.115},
	pages = {1129--1139},
    publisher = {IEEE},
	title = {Significantly Improving Lossy Compression for Scientific Data Sets Based on Multidimensional Prediction and Error-Controlled Quantization},
	year = {2017}}

@inproceedings{soler2018topologically,
    address = {New York},
	author = {Soler, Maxime and Plainchault, M{\'e}lanie and Conche, Bruno and Tierny, Julien},
	booktitle = {2017 IEEE International Parallel and Distributed Processing Symposium (IPDPS)},
	doi = {10.1109/PacificVis.2018.00015},
	pages = {46--55},
    publisher = {IEEE},
	title = {Topologically Controlled Lossy Compression},
	year = {2018}}

@inproceedings{lorensen1987marching,
	address = {New York},
	author = {Lorensen, William E. and Cline, Harvey E.},
	booktitle = {SIGGRAPH '87: Proceedings of the 14th annual conference on Computer graphics and interactive techniques},
	doi = {10.1145/37401.37422},
	pages = {163--169},
	publisher = {ACM},
	title = {Marching Cubes: A High Resolution {3D} Surface Construction Algorithm},
	year = {1987}}

@inproceedings{li2023lossy,
    address = {New York},
	author = {Li, Shaomeng and Lindstrom, Peter and Clyne, John},
	booktitle = {2023 IEEE International Parallel and Distributed Processing Symposium (IPDPS)},
	doi = {10.1109/IPDPS54959.2023.00104},
	pages = {1007--1017},
    publisher = {IEEE},
	title = {Lossy Scientific Data Compression With {SPERR}},
	year = {2023}}

@inproceedings{li2025tracking,
    address = {New York},
    author = {Li, Mingzhe and Chatterjee, Dwaipayan and Glassmeier, Franziska and Senf, Fabian and Wang, Bei},
	booktitle = {2025 IEEE Workshop on Topological Data Analysis and Visualization (TopoInVis)},
	doi = {10.1109/TopoInVis68599.2025.00013},
	pages = {89--99},
    publisher = {IEEE},
	title = {Tracking Low-Level Cloud Systems with Topology},
	year = {2025}}

@inproceedings{chatzimichailidis2019gradvis,
    address = {New York},
	author = {Chatzimichailidis, Avraam and Keuper, Janis and Pfreundt, Franz-Josef and Gauger, Nicolas R.},
	booktitle = {2019 IEEE/ACM Workshop on Machine Learning in High Performance Computing Environments (MLHPC)},
	doi = {10.1109/MLHPC49564.2019.00012},
	pages = {66--74},
    publisher = {IEEE},
	title = {{GradVis}: Visualization and Second Order Analysis of Optimization Surfaces during the Training of Deep Neural Networks},
	year = {2019}}

@inproceedings{aydogan2013analysis,
    address = {Berlin},
	author = {Aydogan, Dogu Baran and Moritz, Niko and Aro, Hannu T. and Hyttinen, Jari},
	booktitle = {Medical Image Computing and Computer-Assisted Intervention -- MICCAI 2013},
	doi = {10.1007/978-3-642-40763-5_53},
    pages = {428--435},
	publisher = {Springer},
	title = {Analysis of trabecular bone microstructure using contour tree connectivity},
	year = {2013}}

@inproceedings{xia2025tspsz,
    address = {New York},
	author = {Xia, Mingze and Wang, Bei and Li, Yuxiao and Jiao, Pu and Liang, Xin and Guo, Hanqi},
	booktitle = {2025 IEEE 41st International Conference on Data Engineering (ICDE)},
	doi = {10.1109/ICDE65448.2025.00275},
	pages = {3682-3695},
    publisher = {IEEE},
	title = {{TspSZ}: An Efficient Parallel Error-Bounded Lossy Compressor for Topological Skeleton Preservation},
	year = {2025}}

@inproceedings{lakshminarasimhan2013isabela,
    address = {Berlin},
	author = {Lakshminarasimhan, Sriram and Shah, Neil and Ethier, Stephane and Klasky, Scott and Latham, Rob and Ross, Rob and Samatova, Nagiza F.},
	booktitle = {Euro-Par 2011 Parallel Processing},
	doi = {10.1007/978-3-642-23400-2_34},
    pages = {366--379},
    publisher = {Springer},
	title = {Compressing the Incompressible with {ISABELA}: In-situ Reduction of Spatio-temporal Data},
	year = {2011}}

@inproceedings{liu2021exploring,
    address = {New York},
	author = {Liu, Jinyang and Di, Sheng and Zhao, Kai and Jin, Sian and Tao, Dingwen and Liang, Xin and Chen, Zizhong and Cappello, Franck},
	booktitle = {2021 IEEE International Conference on Cluster Computing (CLUSTER)},
	doi = {10.1109/Cluster48925.2021.00034},
	pages = {294-306},
    publisher = {IEEE},
	title = {Exploring Autoencoder-based Error-bounded Compression for Scientific Data},
	year = {2021}}

@inproceedings{li2023sperr,
    address = {New York},
	author = {Li, Shaomeng and Lindstrom, Peter and Clyne, John},
	booktitle = {2023 IEEE International Parallel and Distributed Processing Symposium (IPDPS)},
	doi = {10.1109/IPDPS54959.2023.00104},
	pages = {1007-1017},
    publisher = {IEEE},
	title = {Lossy Scientific Data Compression With {SPERR}},
	year = {2023}}

@inproceedings{xia2024preserving,
    address = {New York},
	author = {Xia, Mingze and Di, Sheng and Cappello, Franck and Jiao, Pu and Zhao, Kai and Liu, Jinyang and Wu, Xuan and Liang, Xin and Guo, Hanqi},
	booktitle = {2024 IEEE 40th International Conference on Data Engineering (ICDE)},
	doi = {10.1109/ICDE60146.2024.00378},
	pages = {4979--4992},
    publisher = {IEEE},
	title = {Preserving Topological Feature with Sign-of-Determinant Predicates in Lossy Compression: A Case Study of Vector Field Critical Points},
	year = {2024}}

@inproceedings{zhao2020sdrbench,
    address = {New York},
	author = {Zhao, Kai and Di, Sheng and Lian, Xin and Li, Sihuan and Tao, Dingwen and Bessac, Julie and Chen, Zizhong and Cappello, Franck},
	booktitle = {2020 IEEE International Conference on Big Data (Big Data)},
	doi = {10.1109/BigData50022.2020.9378449},
	pages = {2716-2724},
    publisher = {IEEE},
	title = {{SDRBench}: Scientific data reduction benchmark for lossy compressors},
	year = {2020}}

@misc{agarwal2026toposzp,
	author = {Agarwal, Tripti and Di, Sheng and Liang, Xin and Su, Zhaoyuan and Li, Yuxiao and Gopalakrishnan, Ganesh and Guo, Hanqi and Cappello, Franck},
	doi = {10.48550/arXiv.2602.17552},
	howpublished = {arXiv preprint arXiv:2602.17552},
	title = {{TopoSZp}: Lightweight Topology-Aware Error-Controlled Compression for Scientific Data},
	year = {2026}}

@misc{li2026pmsz,
	author = {Li, Yuxiao and Xia, Mingze and Liang, Xin and Wang, Bei and Underwood, Robert and Di, Sheng and Sharma, Hemant and Beniwal, Dishant and Cappello, Franck and Guo, Hanqi},
	doi = {10.48550/arXiv.2601.01787},
	howpublished = {arXiv preprint arXiv:2601.01787},
	title = {{pMSz}: A Distributed Parallel Algorithm for Correcting Extrema and Morse Smale Segmentations in Lossy Compression},
	year = {2026}}

@misc{fallin2026fast,
	author = {Fallin, Alex and Gorski, Nathaniel and Agarwal, Tripti and Wang, Bei and Gopalakrishnan, Ganesh and Burtscher, Martin},
	doi = {10.48550/arXiv.2603.26968},
	howpublished = {arXiv preprint arXiv:2603.26968},
	title = {Fast Topology-Aware Lossy Data Compression with Full Preservation of Critical Points and Local Order},
	year = {2026}}

@misc{collin2025xzutils,
	author = {Lasse Collin and Jia Tan},
	howpublished = {\url{https://www.tukaani.org/xz/}},
	title = {{XZ Utils}},
	year = {2025}}

@misc{farncomb2015gzip,
	author = {Farncomb, Jarrod},
	howpublished = {\url{rootusers.com/gzip-vs-bzip2-vs-xz-performance-comparison/}},
	title = {{Gzip} vs {Bzip2} vs {XZ} Performance Comparison},
	year = {2015}}

@misc{tierny2024ttk,
	author = {Tierny, Julien and Pont, Mathieu and Finken, Tanner and Lukaczyk, Jonas and Masood, Talha Bin and Vidal, Jules and Hofmann, Lutz and Favelier, Guillame and Cl{'\e}mot, Matt{\'e}o and Le Guillou, Eve and Kappe, Christopher and Wetzels, Florian},
	howpublished = {\url{https://github.com/topology-tool-kit/ttk-data}},
	note = {GitHub repository. Accessed 2026-04-01},
	title = {{TTK} data examples},
	year = {2024}}

@misc{ieee2016scivis,
	howpublished = {\url{https://scads.ai/ieee-scivis-contest-2016/}},
	title = {{IEEE} SciVis Contest 2016},
	year = {2016}}

@misc{ieee2006visualization,
	howpublished = {\url{http://sciviscontest.ieeevis.org/2006/}},
	title = {2006 {IEEE} Visualization Design Contest},
	year = {2006}}

@misc{ieee2008visualization,
	howpublished = {\url{http://sciviscontest.ieeevis.org/2008/}},
	title = {2008 {IEEE} Visualizatin Design Contest},
	year = {2008}}

@misc{ieee2004ieee,
	howpublished = {\url{http://sciviscontest.ieeevis.org/2004/}},
	title = {{IEEE} Visualization Design Contest 2004},
	year = {2004}}

\appendix

\newpage

\crefalias{section}{appendix}
\crefalias{subsection}{appendix}

\section{Supplementary Overview}

This supplement is organized as follows. We describe the datasets used in our experiments in \cref{appendix:dataset-descriptions} and provide additional details on our critical point extraction technique in \cref{appendix:critical-point-extraction}. We then describe how users can extend {\tool} with new features in \cref{appendix:extension} and present formal correctness proofs in \cref{appendix:proofs}. Additional experimental results are provided in \cref{appendix:additional-experiments}, followed by the parameter values used in our experiments in \cref{appendix:parameters}.
\section{Descriptions of Datasets}
\label{appendix:dataset-descriptions}

Our experimental study evaluates nine datasets across a range of compressor configurations.

The \textbf{QMCPACK} dataset is derived from the SDR bench~\cite{zhao2020sdrbench} and is based on a QMCPACK continuum quantum Monte Carlo performance test \cite{kim2018qmcpack, kent2020qmcpack}. We select the 145th orbital (out of 288) and normalize it to $[0,1]$.

The \textbf{Viscous Fingering} dataset is obtained from the TTK data repository~\cite{tierny2024ttk} and originates from the 2016 IEEE SciVis Contest~\cite{ieee2016scivis}.

The \textbf{Tangaroa} dataset consists of a single frame from a simulation of wind flow around the Research Vessel Tangaroa~\cite{popinet2004experimental}. Wind speed serves as the scalar field, which we normalize to $[0,1]$.

The \textbf{Earthquake} dataset originates from the TeraShake 2 earthquake simulation \cite{olsen2008terashake2} and was featured in the 2006 IEEE Visualization Design Contest \cite{ieee2006visualization}. It was subsequently processed by Pont et al.\ and released alongside their work \cite{pont2021wasserstein}. The dataset captures a single time step of seismic activity along the San Andreas Fault; specifically, we use time step 011700. Prior to analysis, we normalize the data to the range $[0,1]$.

The \textbf{Ionization} dataset is derived from an ionization front simulation \cite{whalen2008ionization} and was included in the 2008 IEEE Visualization Design Contest \cite{ieee2008visualization}. It is obtained from the same repository as the Earthquake dataset. For our experiments, we use data from cluster 2 at time step 0125 and normalize it to $[0,1]$.

The \textbf{Isabel} dataset originates from a hurricane simulation conducted by the National Center for Atmospheric Research and was featured in the 2004 IEEE Visualization Contest \cite{ieee2004ieee}. Wind speed is used as the scalar field. While the original dataset has dimensions $500 \times 500 \times 100$, we truncate it to $500 \times 500 \times 90$ to exclude land regions with negligible wind activity and normalize the data to $[0,1]$.

The \textbf{Miranda} dataset is accessed from the SDR bench. Density is used as the scalar field. It originates from a hydrodynamics code for large-scale turbulence simulations conducted at Lawrence Livermore National Laboratory.

The \textbf{S3D} dataset is obtained from the S3D turbulence simulation software \cite{treichler2017s3d}. It is a derived field based on data from the SDR bench; specifically, we compute speed from the velocity $x$, $y$, and $z$ component fields.

The \textbf{SCALE-LETKF} dataset is accessed from the SDR bench. It originates from the Local Ensemble Transform Kalman Filter (LETKF) data assimilation package for the SCALE-RM weather model \cite{lien2017near}. We use the QV field, which is upsampled to double precision.

\section{Extraction of Critical Points}
\label{appendix:critical-point-extraction}

We describe our strategy for critical point extraction in 3D structured mesh. In our mesh structure, each point $x$ may have up to 14 distinct neighbors. Interior points have neighbors in all 14 positions, whereas boundary points have fewer. We index these positions consistently. For any $x \in \X$, let $y_i$ denote the neighbor of $x$ at position $i$, when it exists.

For interior points, we adopt a lookup-table strategy similar to \cite{lukasczyk2024extreem}. For each interior point $x \in \X$, we construct a 14-bit integer $n$ such that the $i$-th bit is set to $1$ if $y_i \in \Lk^+(x)$ and to $0$ otherwise. This encoding captures the structure of the upper and lower links of $x$, allowing its criticality type to be determined entirely from $n$. We precompute a lookup table $L$ such that $L[n]$ returns the criticality type of $x$.

This approach implicitly assumes that if the $i$-th bit is $0$, then $y_i$ exists and belongs to $\Lk^-(x)$. However, for boundary points, a zero bit may also indicate that $y_i$ does not exist, leading to ambiguity and incorrect classifications.

To address this limitation, we use an alternative encoding for boundary points. For any $x \in \X$, we represent a subset $S \subseteq \Lk(x)$ using a 14-bit integer $m$, where the $i$-th bit is set to $1$ if $y_i \in S$ and to $0$ otherwise. This representation is unambiguous regardless of whether $x$ lies in the interior or on the boundary. From $m$, we can reconstruct $S$ and use a lookup table to determine (i) the number of connected components in $S$ and (ii) whether $S$ is simply connected.

Accordingly, we define integers $m_+$ and $m_-$ that encode $\Lk^+(x)$ and $\Lk^-(x)$, respectively. Using lookup tables, we compute (i) the number of connected components in each of $\Lk^+(x)$ and $\Lk^-(x)$ and (ii) whether each is simply connected. This information suffices to determine the criticality type of $x$.
\section{Defining and Incorporating New Features}
\label{appendix:extension}

In this section, we describe how to add new features to {\tool}. The high-level logic of {\tool} is implemented in Python, while performance-critical kernels are written in C++ and exposed through custom bindings. Internally, scalar fields are represented as NumPy arrays, and new features must integrate with this representation.

Each feature in {\tool} is represented by an instance of a Python \texttt{Feature} class. A feature instance stores metadata such as the feature description and its parameters. It also provides a function for determining whether the feature matches between two scalar fields, which is used to verify that the feature is preserved during compression. In addition, each instance includes a function for computing and exporting the feature. To add a new feature to {\tool}, a user creates a new instance of this class. Although the matching and export functions are not strictly required for feature-preserving compression, they are convenient utilities that are typically used alongside it.

To enable preservation of a new feature during compression, the user must also implement functions for computing upper and lower bounds and for performing star classification. The bound computation function is written in Python and takes as input a NumPy array containing the ground-truth scalar field. It returns two NumPy arrays of the same shape containing the upper and lower bounds at each point, respectively. For efficiency, this function should be implemented using NumPy vectorized operations or custom C++ code.

Because star correction is performed inside a performance-critical loop during compression, the star classification function must be implemented in C++. It accepts two arrays of floating-point values representing the ground-truth and reconstructed stars, respectively, and returns the set of ``bad'' points. This set is encoded as a 16-bit unsigned integer, where the 15 least significant bits correspond to the points in the star. A bit is set to 1 if the corresponding point is classified as bad and to 0 otherwise.

Finally, the new functions must be integrated into the main pipeline by adding branches to two dispatch statements: one for upper and lower bound computation and another for star classification.

\section{Correctness Proofs}
\label{appendix:proofs}

In this section, we present proofs of {\tool}’s correctness. We begin by introducing the necessary definitions in \cref{appendix:proofs-definitions}. Next, we establish the correctness of the main pipeline in \cref{appendix:proofs-pipeline}, followed by proofs for the preservation of individual features in \cref{appendix:proofs-descriptors}. Finally, we address the correctness of contour tree preservation, considerably more involved than the other features, in \cref{appendix:proofs-contour-tree}.

\subsection{Definitions}
\label{appendix:proofs-definitions}

We first introduce the key definitions needed to establish the correctness of {\tool}. As before, let $\X$ be a simply connected tetrahedral mesh, and let $f:\X \rightarrow \R$ be a piecewise-linear (PL) scalar field, assumed to be injective on the vertices of $\X$. Let $f'$ denote the reconstructed field obtained after compression and decompression with {\tool}. 
Let $U_{\Phi,f}, L_{\Phi,f} : \X \rightarrow \R$ denote upper and lower bounds, respectively, and let $C_{\Phi,f}$ be the star classification function. Finally, let $P$ denote the maximum allowable precision for logarithmic scaling quantization.

During star correction, iteration can arise in two distinct contexts. First, when correcting an individual star, values are updated iteratively until the star is correctly classified. Second, multiple sweeps over the domain $\X$ may be required to ensure that all stars are correctly classified. We use the terms \emph{iterations} and \emph{passes} to distinguish between these two processes.

During pass $t$, let $D_t \subset \X$ denote the set of points marked as \emph{dirty} for that pass. The points that undergo star correction in pass $t$ are therefore exactly those in $D_t$.

\begin{definition}
A point $x \in \X$ is \emph{in-bounds} if it satisfies one of the following conditions:
\begin{itemize}
\item[(a)] $L_{\Phi,f}(x) < f'(x) < U_{\Phi,f}(x)$,
\item[(b)] $f'(x) = f(x)$.
\end{itemize}
\end{definition}

One of the goals of {\tool} is to enforce a pointwise global error bound $\xi$. For simplicity, we assume that $L_{\Phi,f}(x) \geq f(x) - \xi$ and $U_{\Phi,f}(x) \leq f(x) + \xi$. Under this assumption, any in-bounds point $x$ satisfies $|f'(x) - f(x)| \leq \xi$. Consequently, we omit explicit reference to $\xi$ in the remainder of the section.

\begin{definition}
A point $x \in \X$ is \emph{correctly classified} if $C_{\Phi,f}(x) = C_{\Phi,f'}(x)$. Equivalently, $x$ is correctly classified if all points in $\St(x)$ are marked as good.
\end{definition}

\begin{definition}
The classification function $C_{\Phi,f}$ is \emph{nondegenerate} if, whenever a set of vertices $\{y_i\} \subseteq \St(x)$ contains at least one point marked as bad, there exists such a $y_i$ for which $f(y_i) \neq f'(y_i)$.
\end{definition}

\subsection{Correctness of Main Pipeline}
\label{appendix:proofs-pipeline}

We now establish the correctness of the main pipeline. Specifically, we show that the pipeline always terminates and that, upon termination, every $x \in \X$ is in-bounds and correctly classified. This result alone does not guarantee the preservation of any particular features; such guarantees additionally depend on appropriate choices of $C_{\Phi,f}$, $U_{\Phi,f}$, and $L_{\Phi,f}$. We address the preservation of specific features in \cref{appendix:proofs-descriptors}.

\begin{lemma}
When a point $x \in \X$ is visited during logarithmic-scaling quantization, the quantization process terminates, and upon termination, $x$ is in-bounds.
\label{lemma:pipeline-quantization}
\end{lemma}
\begin{proof}
Logarithmic-scaling quantization is an iterative procedure that repeatedly applies linear-scaling quantization with increasing precision levels $p_x$. At each iteration, if $L_{\Phi,f}(x) < f'(x) < U_{\Phi,f}(x)$, the process terminates. Otherwise, the precision $p_x$ is incremented and quantization continues.

If this termination condition is not satisfied after a fixed number of iterations, the process is truncated and $f'(x) \gets f(x)$. In either case, the procedure terminates, and upon termination, $x$ is in-bounds.
\end{proof}

\begin{lemma}
Assume that $C_{\Phi,f}$ is nondegenerate. Let $x \in \X$ be a point visited during iterative star correction. Then the correction process for $x$ terminates, and upon termination, $x$ is correctly classified.
\label{lemma:pipeline-point-correction}
\end{lemma}
\begin{proof}
When $x$ is visited during star correction, the process proceeds through a sequence of iterations. In each iteration, termination occurs only if no points in $\St(x)$ are marked as bad. Consider an iteration in which some points are still marked as bad.

By the nondegeneracy assumption, if a set of points $\{y_i\} \subseteq \St(x)$ is marked as bad, then at least one such $y_i$ satisfies $f(y_i) \neq f'(y_i)$. This implies that $p_{y_i} \leq P$, since otherwise $y_i$ would be stored losslessly.

Among all points marked as bad, let $y_j$ be one with minimal precision $p_{y_j}$. Then $p_{y_j} \leq P$, and during this iteration, its precision is incremented: $p_{y_j} \gets p_{y_j} + 1$. Therefore, at each iteration, the minimum precision among bad points strictly increases. Since $p_{y_j}$ is bounded above by $P$, only finitely many such increments are possible.

Consequently, after a finite number of iterations, all points in $\St(x)$ will have precision exceeding $P$ and hence will be stored losslessly, so no points remain marked as bad. This implies that star correction terminates.

Finally, since termination occurs only when no points in $\St(x)$ are marked as bad, it follows by assumption that $x$ is correctly classified upon termination.
\end{proof}

\begin{lemma}
Assume that $C_{\Phi,f}$ is nondegenerate. Suppose that at the beginning of pass $t$, every incorrectly classified point belongs to $D_t$. Then, at the beginning of pass $t+1$, every incorrectly classified point belongs to $D_{t+1}$.
\label{lemma:pipeline-all-dirty}
\end{lemma}

\begin{proof}
During pass $t$, a point $x$ is added to $D_{t+1}$ if either of the following conditions holds:
\begin{itemize}[noitemsep]
\item[(a)] $f'(x)$ is updated when a neighbor $y$ undergoes star correction,
\item[(b)] $x$ has a neighbor $y$ whose value $f'(y)$ is updated at some point during the pass.
\end{itemize}

Let $x \in \X$. We show that if neither (a) nor (b) holds, then $\St(x)$ is correctly classified at the end of pass $t$. We consider two cases.

\underline{Case 1: $x \notin D_t$.} In this case, $x$ does not undergo star correction during pass $t$. The value $f'(x)$ can only change either when $x$ itself is corrected or when a neighbor undergoes star correction. Since $x \notin D_t$, the former does not occur. Because (a) is false, the latter does not affect $f'(x)$ either. Hence, $f'(x)$ remains unchanged throughout the pass. 

Moreover, since (b) is false, none of the neighbors of $x$ experience a change in value. Therefore, the values of $f'$ over $\St(x)$ remain unchanged during the pass. Because $x$ was not marked dirty at the start of the pass, it is correctly classified initially. Since the values on $\St(x)$ do not change, $x$ remains correctly classified at the end of the pass.

\underline{Case 2: $x \in D_t$.} In this case, $x$ undergoes star correction during pass $t$. By \cref{lemma:pipeline-point-correction}, this ensures that $\St(x)$ is correctly classified immediately after correction. Since both (a) and (b) are false, no values in $\St(x)$ change for the remainder of the pass. Hence, $\St(x)$ remains correctly classified at the end of the pass.

In both cases, we conclude that if $x \notin D_{t+1}$, then $x$ is correctly classified at the start of pass $t+1$. Therefore, every incorrectly classified point must belong to $D_{t+1}$.
\end{proof}

\begin{lemma}
Each $x \in \X$ can be marked as dirty during only finitely many passes.
\label{lemma:pipeline-finite-passes}
\end{lemma}

\begin{proof}
During pass $t$, a vertex $x \in \X$ is added to $D_{t+1}$ only if there exists some $y \in \St(x)$ whose value is updated. By reasoning analogous to that in the proof of \cref{lemma:pipeline-point-correction}, each such $y \in \St(x)$ can undergo only finitely many value updates across all passes.

Since $\St(x)$ contains finitely many vertices, it follows that $x$ can be marked as dirty only finitely many times.
\end{proof}

\begin{theorem}
If $C_{\Phi,f}$ is nondegenerate, then the main pipeline terminates. Upon termination, every $x \in \X$ is correctly classified and in-bounds.
\end{theorem}

\begin{proof}
The first step of the main pipeline is initial quantization. By \cref{lemma:pipeline-quantization}, this step terminates, and upon completion, every $x \in \X$ is in-bounds.

We then proceed to iterative star correction. During the first pass, every vertex $x \in \X$ belongs to $D_1$, so all incorrectly classified points are marked as dirty. By repeatedly applying \cref{lemma:pipeline-all-dirty}, it follows that at the start of each pass $t > 1$, every incorrectly classified point is contained in $D_t$. If, for some pass $T$, we have $|D_T| = 0$, then iterative star correction terminates. By \cref{lemma:pipeline-finite-passes}, each $x \in \X$ can belong to only finitely many sets $D_t$. Therefore, for sufficiently large $t$, we must have $|D_t| = 0$, and star correction terminates. At that point, since no points are marked as dirty, \cref{lemma:pipeline-all-dirty} implies that all points are correctly classified.

It remains to show that all points remain in-bounds. Fix $x \in \X$. After initial quantization, $x$ is in-bounds. During star correction, the value of $f'(x)$ may change only if logarithmic-scaling quantization is applied to $x$. By definition, this procedure preserves the bounds. Hence, $x$ remains in-bounds throughout star correction, and therefore is in-bounds upon termination.

Finally, the remaining step of the pipeline is lossless compression, which does not modify the values of $f'$. This completes the proof.
\end{proof}

\subsection{Correctness of Individual Feature Preservation}
\label{appendix:proofs-descriptors}

We now establish the correctness of individual feature preservation. For each feature $\Phi$, we verify the following:
\begin{itemize}[noitemsep]
\item[(a)] $C_{\Phi,f}$ is nondegenerate (ensuring termination),
\item[(b)] if every $x \in \X$ is correctly classified and in-bounds, then $\Phi(f) = \Phi(f')$ (ensuring that the feature is preserved).
\end{itemize}

\begin{theorem}
{\tool} correctly preserves critical points.
\label{lemma:descriptors-critical-points}
\end{theorem}

\begin{proof}
By construction, $x$ is marked as bad if and only if there exists some $y \in \Lk(x)$ that is marked as bad. A vertex $y$ is marked as bad if and only if 
\[
\bigl(f(x) < f(y)\bigr) \;\oplus\; \bigl(f'(x) < f'(y)\bigr).
\]
If $f(x) = f'(x)$ and $f(y) = f'(y)$, then this xor ($\oplus$) condition is false. Hence, whenever $y$ is marked as bad, it must be that either $f(x) \neq f'(x)$ or $f(y) \neq f'(y)$. Therefore, $C_{\Phi,f}$ is nondegenerate.

Now suppose that some $x \in \X$ has an incorrect criticality type in $f'$. Then there exists a vertex $y \in \Lk(x)$ whose classification relative to $x$ is inconsistent between $f$ and $f'$; that is, $y \in \Lk^+(x)$ in one field and $y \in \Lk^-(x)$ in the other. By construction, such a $y$ is marked as bad. Consequently, if no points in $\St(x)$ are marked as bad, then no such inconsistency exists, and $x$ must have the correct criticality type in $f'$.

Thus, if every $x \in \X$ is correctly classified, then every $x$ has the same criticality type in $f'$ as in $f$, which is equivalent to preserving all critical points.
\end{proof}

\begin{theorem}
{\tool} correctly preserves Morse--Smale segmentation.
\end{theorem}

\begin{proof}
We first show that $C_{\Phi,f}$ is nondegenerate. Suppose that some $y \in \St(x)$ is marked as bad. By construction, this occurs if at least one of the following holds:
\begin{itemize}[noitemsep]
\item[(i)] $x$ is inconsistently classified as a minimum or maximum between $f$ and $f'$,
\item[(ii)] the largest vertex in $\Lk^+(x)$ differs between $f$ and $f'$,
\item[(iii)] the smallest vertex in $\Lk^-(x)$ differs between $f$ and $f'$.
\end{itemize}

We show that in each case there exists a bad $y \in \St(x)$ such that $f(y) \neq f'(y)$. Case (i) follows as in the proof of \cref{lemma:descriptors-critical-points}. For case (ii), let $y$ and $y'$ denote the largest vertices in $\Lk^+(x)$ with respect to $f$ and $f'$, respectively, with $y \neq y'$. Then $f(y) > f(y')$ while $f'(y) < f'(y')$. This reversal is only possible if either $f(y) \neq f'(y)$ or $f(y') \neq f'(y')$. Case (iii) is analogous. Therefore, $C_{\Phi,f}$ is nondegenerate.

We now show that if all points are correctly classified, then the Morse--Smale segmentation is preserved. It suffices to show that all forward and backward integral lines are identical in $f$ and $f'$. We represent an integral line as a sequence of vertices in $\X$. Let $a_1 a_2 \cdots a_N$ and $b_1 b_2 \cdots b_M$ be integral lines of $f$ and $f'$, respectively, with $a_1 = b_1$. We show that these sequences coincide.

Let $Q = \min(N,M)$. We prove by induction that $a_i = b_i$ for all $i \leq Q$. The base case $i=1$ is immediate. Assume $a_k = b_k$ for some $k < Q$. Since $k < Q \leq N$, the vertex $a_k$ is not a maximum in $f$, so $a_{k+1}$ is the largest vertex in $\Lk^+(a_k)$ with respect to $f$. Similarly, $b_{k+1}$ is the largest vertex in $\Lk^+(b_k)$ with respect to $f'$. Because $a_k$ is correctly classified, these largest neighbors agree in $f$ and $f'$. Since $a_k = b_k$, it follows that $a_{k+1} = b_{k+1}$.

Thus, $a_i = b_i$ for all $i \leq Q$. Without loss of generality, assume $Q = N$. Then $a_Q$ is a maximum in $f$. Since $a_Q = b_Q$ and $a_Q$ is correctly classified, $b_Q$ is also a maximum in $f'$, implying $M = N$. Hence, the two integral lines coincide.

Therefore, all forward integral lines agree between $f$ and $f'$. By an analogous argument, all backward integral lines also agree. This establishes that the Morse--Smale segmentation is preserved.
\end{proof}

\begin{theorem}
{\tool} correctly preserves local order. 
\end{theorem}

\begin{proof}
As with critical points, by construction a vertex $x$ is marked as bad if and only if there exists some $y \in \Lk(x)$ that is marked as bad. A vertex $y \in \Lk(x)$ is marked as bad if and only if
\[
\bigl(f(x) < f(y)\bigr) \;\oplus\; \bigl(f'(x) < f'(y)\bigr),
\]
where $\oplus$ denotes exclusive or (xor).

If $f(x) = f'(x)$ and $f(y) = f'(y)$, then the two inequalities agree, and hence the above condition is false. Therefore, no such $y$ is marked as bad, and it follows that $C_{\Phi,f}$ is nondegenerate.

We now verify that all ``less-than'' relationships are preserved. Suppose, for contradiction, that some such relationship is not preserved. Then there exist adjacent vertices $x,y \in \X$ such that $f(x) < f(y)$ but $f'(x) > f'(y)$. Equivalently, $y \in \Lk^+(x)$ with respect to $f$, while $y \in \Lk^-(x)$ with respect to $f'$. In this situation, the xor  condition above holds, so $y$ is marked as bad. 

Consequently, if no vertices are marked as bad, all ``less-than'' relationships are preserved.
\end{proof}

\begin{theorem}
{\tool} correctly preserves isocontours. In particular, let $\sigma \in \X$ be a simplex and let $c$ be an isovalue whose isocontour is preserved. If every $x \in \X$ is in-bounds, then $f^{-1}(c)$ intersects $\sigma$ if and only if $(f')^{-1}(c)$ intersects $\sigma$.
\label{lemma:descriptors-isocontours}
\end{theorem}

\begin{proof}
Since isocontour preservation does not involve star correction, it is not necessary to establish nondegeneracy of $C_{\Phi,f}$.

Let $\{x_i\}$ denote the vertices of $\sigma$. In the PL setting, whether $f^{-1}(c)$ intersects $\sigma$ is determined by whether each $f(x_i)$ is greater than, less than, or equal to $c$.

Fix a vertex $x_i$. If $f(x_i) > c$, then by construction $L_{\Phi,f}(x_i) \geq c$. Since $x_i$ is in-bounds, it follows that $f'(x_i) > L_{\Phi,f}(x_i) \geq c$. Similarly, if $f(x_i) < c$, then $U_{\Phi,f}(x_i) \leq c$, and thus $f'(x_i) < U_{\Phi,f}(x_i) \leq c$. Finally, if $f(x_i) = c$, then by construction $f'(x_i) = f(x_i) = c$.

Therefore, for each vertex $x_i$, the relative ordering of $f(x_i)$ with respect to $c$ agrees with that of $f'(x_i)$. Since this holds for all vertices of $\sigma$, it follows that $f^{-1}(c)$ intersects $\sigma$ if and only if $(f')^{-1}(c)$ intersects $\sigma$.
\end{proof}

\begin{theorem}
{\tool} correctly preserves quantiles. 
\end{theorem}

\begin{proof}
Quantile bounds are preserved by first computing the quantile values and then applying the same isocontour construction procedure using those quantile thresholds. Under this approach, preservation is defined pointwise: for a given quantile level, each point is constrained relative to the corresponding quantile value according to its position in the ground truth distribution—points above the quantile are bounded from below by the quantile value, and points below the quantile are bounded from above by it.

Since quantile preservation does not involve star correction, it is not necessary to establish the nondegeneracy of $C_{\Phi,f}$.

Let $|\X| = n$ and fix $m < n$. Let $s$ be a quantile value selected by {\tool} that separates the $m$ smallest points from the $n-m$ largest points, so that
\[
\bigl|f^{-1}(-\infty,s)\bigr| = m.
\]
We claim that
\[
f^{-1}(-\infty,s) = (f')^{-1}(-\infty,s).
\]

First, suppose that $x \in f^{-1}(-\infty,s)$. Then $f(x) < s$ and $U(x) \le s$, which implies that $f'(x) < s$. Hence $x \in (f')^{-1}(-\infty,s)$. 
Conversely, suppose that $x \notin f^{-1}(-\infty,s)$. Then $f(x) \ge s$ and $L(x) \ge s$, which implies that $f'(x) \ge s$. Hence $x \notin (f')^{-1}(-\infty,s)$.

Therefore,
\[
f^{-1}(-\infty,s) = (f')^{-1}(-\infty,s).
\]
It follows that the threshold $s$ separates the $m$ smallest points from the $n-m$ largest points for both $f$ and $f'$. Thus, quantiles are preserved.
\end{proof}

\subsection{Correctness of Contour Tree Preservation}
\label{appendix:proofs-contour-tree}

We now establish the correctness of contour tree preservation. We begin by introducing notation and preliminary results in \cref{appendix:proof-contour-tree-preliminaries}. Next, we define the upper and lower bounds and present the necessary lemmas in \cref{sec:proofs-contour-tree-lemmas}. Finally, we prove the main correctness result in \cref{appendix:proof-contour-tree-main-result}.

\subsubsection{Preliminaries and Notation}
\label{appendix:proof-contour-tree-preliminaries}

We begin by introducing preliminary definitions and notation related to merge and contour trees. We assume that the reader is familiar with the standard definitions of simplicial complexes, merge trees, contour trees, and persistence simplification. Given a function $f:\X \rightarrow \R$, let $\M(f)$ and $\T(f)$ denote its merge tree and contour tree, respectively. Following~\cite{carr2003computing}, the merge tree of $f$, $\M(f)$, is sometimes referred to as the \emph{join tree}, and the merge tree of $-f$, $\M(-f)$, as the \emph{split tree}.

To avoid ambiguity, we use the term \emph{node} exclusively to refer to a vertex of a merge or contour tree. Vertices of the domain---the simplicial complex $\X$---are referred to as \emph{vertices}. We refer to interior nodes of $\M(f)$ as \emph{join saddles}, and interior nodes of $\M(-f)$ as \emph{split saddles}. We assume that each interior node of a merge tree has exactly two children, although the same logic can easily be extended to the special case where a node has more than two children.

Let $\pi_{\M(f)}:\X \rightarrow \M(f)$ denote the canonical map that sends each $x \in \X$ to its corresponding point in $\M(f)$. If $v$ is a node of $\M(f)$, then there exists a unique vertex $u \in \X$ such that $\pi_{\M(f)}(u) = v$. By a slight abuse of notation, we use the same symbol to refer to both objects; that is, we may write $v$ to denote either the node in $\M(f)$ or its preimage in $\X$.

We represent an edge in the merge tree as an ordered pair $(a,b)$ with $f(a) < f(b)$. The set $\pi^{-1}_{\M(f)}\big((a,b)\big) \subseteq \X$ corresponds to the region of the merge-tree-induced segmentation associated with this edge. For brevity, we denote this region by $\pi^{-1}_{\M(f)}(a,b)$. 

Let $\ol{\M(f)}$ be $\M(f)$ augmented with the vertices of $\M(-f)$. Define $\ol{\M(-f)}$ similarly. Then $\T(f)$ is computed by combining $\ol{\M(f)}$ and $\ol{\M(-f)}$ \cite{carr2003computing}.

Persistence simplification can be defined and implemented in many ways. To prevent ambiguities, we define our implementation here, although we believe that our proofs could be extended to other implementations. Let $\varepsilon$ be a persistence threshold. We construct a persistence simplified merge tree $\M_\varepsilon(f)$ by pruning branches of low persistence. We construct a persistence simplified contour tree $\T_\varepsilon(f)$ by combining $\ol{\M_\varepsilon(f)}$ and $\ol{\M_\varepsilon(-f)}$ using the procedure from \cite{carr2003computing}. 

We now define the persistence simplified merge-tree-induced segmentation. Let $x \in \X$. Let $s_2$ be the lowest ancestor of $\pi_{\M(f)}(x)$ in $\M(f)$ that is also a node in $\M_\varepsilon(f)$. $s_2$ will have two children. Let $s_1$ be the child of $s_2$ that lies in the connected component of $f^{-1}(-\infty,f(s_2))$ as $x$. Then we say that $x \in \pi^{-1}_{\M_\varepsilon(f)}(s_1,s_2)$.

\subsubsection{Supporting Lemmas}
\label{sec:proofs-contour-tree-lemmas}

We first prove some lemmas about the merge-tree-induced segmentation.

\begin{lemma}
Let $(a,b)$ be an edge of $\M_\varepsilon(f)$ and let $x \in \pi^{-1}_{\M_\varepsilon(f)}(a,b)$. Let $y \in \X$. Suppose that there is a path $p$ from $x$ to $y$ that never takes a function value above $f(b)$. Then $\pi_{\M(f)}(y)$ descends from $b$ in $\M(f)$, and $y$ lies in the same connected component of $f^{-1}(-\infty,f(b))$ as $x$.
\label{lemma:combine-paths-descendant}
\end{lemma}

\begin{proof}
Since $x \in \pi^{-1}_{\M_\varepsilon(f)}(a,b)$, this implies that $b$ is an ancestor of $\pi_{\M(f)}(x)$ in $\M(f)$. Thus, there is a path $p'$ connecting $b$ to $x$ that does not exceed $f(b)$. If we combine $p$ with $p'$, we can derive a path connecting $y$ to $b$ that does not exceed a function value of $f(b)$. Thus, $\pi(y)$ must descend from $b$ in $\M(f)$. Since $x$ is part of this path, it also follows that $y$ and $x$ must lie in the same connected component of the sublevel set $f^{-1}(-\infty,f(b))$.
\end{proof}

\begin{lemma}
Let $(a,b)$ be an edge $\M_\varepsilon(f)$ and $x \in \pi^{-1}_{\M_\varepsilon(f)}(a,b)$. Let $y \in \X$, and suppose that there is a path $p:[0,1] \rightarrow \X$ from $y$ to $x$ that does not intersect an isocontour with isovalue $f(a)$ or $f(b)$. Then $y \in \pi^{-1}_{\M_\varepsilon(f)}(a,b)$
\end{lemma}

\begin{proof}
Since $x \in \pi^{-1}_{\M_\varepsilon(f)}(a,b)$, this implies that $\pi_{\M(f)}(x)$ is a descendant of $b$ in $\M_f$. Thus, $f(x) < f(b)$. Since the path $p$ contains $x$ and never intersects a contour with isovalue $f(b)$, it holds that $f$ is less than $f(b)$ on all of $p$. Thus, by \cref{lemma:combine-paths-descendant}, $\pi(y)$ must be a descendant of $b$ in $\M(f)$ and $y$ lies in the same connected component of $f^{-1}(-\infty,f(b))$ as $x$ and $a$. 

We are only left to show that $a$ is the lowest ancestor of $\pi_{\M(f)}(y)$ in $\M(f)$ that is also a vertex in $\M_\varepsilon(f)$. Suppose not. Then $y \not \in \pi^{-1}_{\M_\varepsilon(f)}(a,b)$. In this case, $b$ will have exactly one child in $\M_\varepsilon(f)$ that is an ancestor of $\pi(y)$ in $\M(f)$. Since $y$ lies in the same component of the sub-level set of $f^{-1}(f(b))$ as $a$, this child must be $a$. Thus, $y$ is a descendant of $a$ in $\M(f)$ and $f(y) < f(a)$. The path $p$ from $y$ to $x$ does not cross any contour with isovalue $f(a)$, so $p$ must take a function value below $f(a)$ for its entirety. By \cref{lemma:combine-paths-descendant}, $\pi(x)$ is a descendant of $a$ in $\M(f)$, which is a contradiction.
\end{proof}

\begin{corollary}
Let $(a,b)$ be an edge of $\M_\varepsilon(f)$. Then the boundary of $\pi^{-1}_{\M_\varepsilon(f)}$ is composed of contours with isovalues $f(a)$ and $f(b)$.
\label{corollary:contour-boundary}
\end{corollary}

\begin{figure}
\begin{center}
\includegraphics[width=0.6\linewidth]{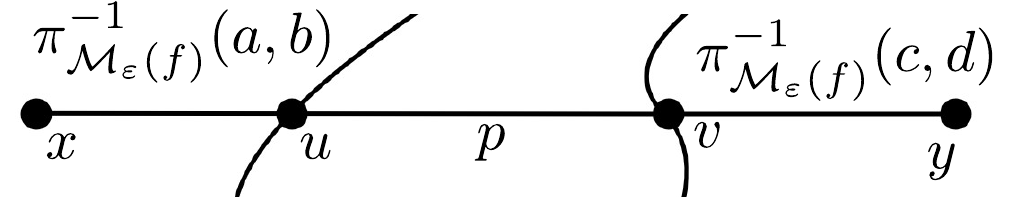}
\end{center}
\caption{Illustration of \cref{lemma:adjacent-increasing}}
\label{fig:adjacent-increasing}
\end{figure}

\begin{lemma}
Suppose that $x$ and $y$ are two adjacent vertices of $\X$ with $f(x) < f(y)$. Suppose that $x$ and $y$ belong to different edges $(a,b)$ and $(c,d)$ of $\M_\varepsilon(f)$. Let $p$ be the path traveling from $x$ to $y$ along the line segment that connects them. Let $u$ and $v$ respectively be the points at which $p$ leaves $\pi^{-1}_{\M_\varepsilon(f)}(a,b)$ and enters $\pi^{-1}_{\M_\varepsilon(f)}(c,d)$. These points are illustrated in \cref{fig:adjacent-increasing}. Then

\[
f(x) \leq f(u) = f(b) \leq f(v) = f(c) < f(y) < f(d)
\]

\label{lemma:adjacent-increasing}
\end{lemma}

\begin{proof}
Since $f$ is PL, then $f$ will strictly increase along $p$ from $f(x)$ to $f(y)$. By assumption, $f(y) < f(d)$.

Notice that the path $p$ connects $x$ to $y$ without taking a function value above $f(d)$. Since $\pi(y)$ descends from $d$ in $\M(f)$, then by \cref{lemma:combine-paths-descendant}, $\pi(x)$ also must descend from $d$ in $\M(f)$ and $x$ lies in the same connected component of $f^{-1}(-\infty,f(d))$ as $y$ and $c$. Since $\pi(x)$ also descends from $b$ in $\M(f)$, this implies that either $b$ descends from $d$ or $d$ descends from $b$ in $\M(f)$. 

If $b$ descends from $d$, then because $a$ is the child of $b$ in $\M(f)$, it must hold that $f(a) < f(b) < f(d)$. Likewise, since $c$ is the child of $d$ in $\M(f)$, and $c$ also lies the same connected component of $f^{-1}(-\infty,f(d))$ as $b$, it would need to follow that $c$ lies between $b$ and $d$ in $\M(f)$, or $b = c$, forcing $f(a) < f(b) \leq f(c) < f(d)$. Using similar reasoning, if it held that $d$ was instead a descendant of $a$ in $\M(f)$, then it would instead hold that $f(c) < f(d) \leq f(a) < f(b)$.

By \cref{corollary:contour-boundary}, either $f(u) = f(a)$ or $f(u) = f(b)$, so $f(a) \leq f(u)$. We also know that $f(u) < f(y)$ because $f$ increases on $p$. And by assumption, $f(y) < f(d)$. Chaining these inequalities yields $f(a) \leq f(u) < f(y) < f(d)$, so $f(a) < f(d)$. This forces $f(a) < f(b) \leq f(c) < f(d)$.

By \cref{corollary:contour-boundary}, either $f(v) = (c)$ or $f(v) = f(d)$. Since $v$ comes before $y$ on $p$, it follows that $f(v) < f(y)$. Since $f(y) < f(d)$, this implies that $f(v) = f(c)$. Thus, $f(c) < f(y)$. And therefore, we have that $f(x) < f(b) \leq f(c) = f(v) < f(y)$.

We now show that $f(u) = f(b)$. We check two cases for $f(a)$. 

\underline{case 1: $f(a) < f(x)$:} In this case, no contour with isovalue $f(a)$ can intersect $p$, forcing $f(u) = f(b)$. 

\underline{Case 2: $f(a) > f(x)$:} since $f(x) < f(b) < f(y)$, there must be some point $w$ on $p$ such that $f(u) = f(a) < f(w) < f(b) < f(y)$. There is a clear path from $x$ to $w$ that does not take a function value above $f(b)$ (namely, along the segment connecting $x$ to $y$). From \cref{lemma:combine-paths-descendant}, this implies that $\pi_{\M(f)}(w)$ is a descendant of $b$ that lies in the same connected component of $f^{-1}(-\infty,f(b))$ as $x$ (and thus $a$). Since $a$ is a child of $b$, if $b$ was not the lowest ancestor of $\pi(w)$ in $\M(f)$ that is also a node in $\M_\varepsilon(f)$, then $\pi(w)$ would also be a descendant of $a$ in $\M(f)$. However, this is impossible since $f(a) < f(w)$. Thus, $w \in \pi^{-1}_{\M_\varepsilon(f)}(a,b)$. Since $x \in \pi^{-1}_{\M_\varepsilon(f)}(a,b)$, $x$ and $w$ must be on the same side of $u$, meaning that $f(u) = f(b)$.

Combining inequalities gives the result.
\end{proof}

We now give a formal definition of our upper and lower bounds.

\begin{definition}
For a given $x \in \X$, let $a_x$ and $b_x$ denote the endpoints of the edge of $\M_\varepsilon(f)$ to which $x$ belongs. Let $c_x$ and $d_x$ denote the endpoints of the edge of $\M_\varepsilon(-f)$ to which $x$ belongs. We say that a function $f'$ is \emph{$\varepsilon$-bounded} if the following conditions hold:

\begin{itemize}[noitemsep]
\item[(a)] If $x$ is a node of $\T_\varepsilon(f)$, then $f'(x) = f(x)$.
\item[(b)] If $x$ is not a node of $\T_\varepsilon(f)$:
\begin{itemize}[noitemsep]
\item[(i)] $f(c_x) < f'(x) < f(b_x)$. 
\item[(ii)] If $f(x) > f(a_x)$ then $f'(x) > f(a_x)$; otherwise, if $f(x) \leq f(a_x)$, then $f'(x) \geq f(x)$.
\item[(ii)] If $f(x) < f(d_x)$, then $f'(x) < f(d_x)$; otherwise, if $f(x) \geq f(d_x)$, then $f'(x) \leq f(x)$.
\end{itemize}
\end{itemize}
\end{definition}

From this point onward, we assume that $f'$ is $\varepsilon$-bounded for some $\varepsilon > 0$. We also assume that $f'$ satisfies a global error bound $\xi$, namely $\|f - f'\|_\infty \leq \xi$.

\begin{lemma}
Let $\sigma \in \X$ be a simplex and $s \in \M_\varepsilon(f)$ be a saddle. Suppose that $\sigma$ intersects a contour with isovalue $f(s)$ that serves as the boundary between two segmentation regions $\pi^{-1}_{\M_\varepsilon(f)}(a,b)$ and $\pi^{-1}_{\M_\varepsilon(f)}(c,d)$. Then $\sigma$ intersects the contour with isovalue $f(s)$ in $f'$.
\label{lemma:contours-are-preserved}
\end{lemma}

\begin{proof}
We check cases based on the dimension of $\sigma$.

\underline{Case 1: $\sigma$ is a $0$-simplex:} In this case, $\sigma$ is a vertex of $\X$. Since $f$ is injective on the vertices of $\X$, this implies that $\sigma = s$. By assumption, $f'(s) = f(s)$.

\underline{Case 2: $\sigma$ is a $1$-simplex:} In this case, let $x$ and $y$ be its endpoints with $f(x) < f(y)$. Without loss of generality, suppose that $x \in \pi^{-1}_{\M_\varepsilon(f)}(a,b)$ and $y \in \pi^{-1}_{\M_\varepsilon(f)}(c,d)$. Let $p$, $u$, and $v$ be as in the statement of \cref{lemma:adjacent-increasing}. By \cref{lemma:adjacent-increasing}, $f(x) < f(u) = f(b) \leq f(v) = f(c) < f(y)$. Let $w$ be the point along $p$ that intersects the contour with isovalue $f(s)$. Since $x \in \pi^{-1}_{\M_\varepsilon(f)}(a,b)$, the path $p$ must intersect $u$ (the boundary of that region) prior to intersecting $w$ (which is itself a region boundary). Thus, $f(u) = f(b) \leq f(s)$. Similarly, $f(s) \leq f(c)$.

From our bounds, $f'(x) < b$. Since $f(y) > c$, it also holds that $f'(y) > c$. Thus, there exists some point $w'$ on the segment between $x$ and $y$ such that $f'(x) < b \leq f'(w') = f(s) \leq c < f'(y)$.

\underline{Case 3: $\sigma$ is a $2$-simplex or $3$-simplex:} Since $f$ is PL, the contour with isovalue $f(s)$ must intersect one of the $0$ or $1$ dimensional faces of $\sigma$. Appealing to cases 1 and 2, that contour must thus intersect the same face in $f'$, and thus it will also intersect $\sigma$.
\end{proof}

\begin{lemma}
Suppose that $s$ is a merge saddle in $\M_\varepsilon(f)$ and $x \in \X$ is a vertex where $\pi_{\M(f)}(x)$ descends from $s$ in $\M_\varepsilon(f)$. Then there is a path $p'$ from $s$ to $x$ that never takes a function value (in $f'$) above $f'(s)$ (except for the endpoint $s$).
\label{lemma:sublevel-sets-are-preserved}
\end{lemma}

\begin{proof}
Let $p$ be a path from $s$ to $x$ that never takes a function value (in $f$) above $f(s)$ except for at the endpoint $s$. We know that such a path exists becase $\pi_{\M(f)}(x)$ descends from $s$. Let $\sigma_1, \sigma_2, \ldots, \sigma_k$ be the (not necessearily distinct) simplices that $p$ passes through, ordered in terms of when $p$ intersects them (so $\sigma_1 = p$ and $\sigma_k = x$). Order the $\sigma_i$ so that $\sigma_i$ is either a face or co-face of $\sigma_{i+1}$. For each $\sigma_i$, let $y_i$ be the vertex of $\sigma_i$ with the lowest function value (in $f$). Then it must follow that $f(y_i) < f(s)$ for all $i$, except for $y_1$ where $f(y_1) = f(s)$. Since each $\sigma_i$ is a face or co-face of $\sigma_{i+1}$, the $y_i$ are all adjacent vertices of $\X$.

Let $p'$ be the path formed by connecting the $y_i$ in order by line segments. Since each $y_i$ is adjacent to $y_{i+1}$, the segment of $p$ between $y_i$ and $y_{i+1}$ will take values of $f$ that interpolate linearly between $f(y_i)$ and $f(y_{i+1})$. Since each $f(y_i) \leq f(s)$, this implies that the entire  path $p$ takes a function value (in $f$) below $f(s)$ except for the endpoint $s$. Thus, each $\pi_{\M(f)}(y_i)$ is a descendant of $s$ in $\M(f)$. Therefore, each $y_i$ belongs to some edge $(a_i,b_i)$ of $\M_\varepsilon(f)$ with $f(b_i) \leq f(s)$. 

From the error bounds, for each $y_i$ except $y_1$, we either have $f'(y_i) < f(b_i) \leq f(s)$ or $f'(y_i) \leq f(y_i) < f(s)$. In either case, $f'(y_i) < f(s)$. And also, $f'(y_1) = f(y_1) = f(s)$. Thus, the path $p'$ will also never take a function value (in $f'$) above $f'(s) = f(s)$, except for at the intial point $y_1$ where $f(y_1) = f'(s) = f(s)$.
\end{proof}

\begin{lemma}
Let $x, y \in \X$ be two points. Let $s$ be the least common ancestor of $\pi(x)$ and $\pi(y)$ in $\M(f)$. Suppose that $s$ is also a node in $\M_\varepsilon(f)$. Suppose that $x$ and $y$ lie in different connected components of the sub-level set $f^{-1}(-\infty,f(s))$. Then $x$ and $y$ lie in different connected components of the sub-level set $(f')^{-1}(-\infty,f'(s))$.
\label{lemma:sublevel-set-complements-are-preserved}
\end{lemma}

\begin{proof}
Let $S_1$ and $S_2$ be the two connected components of $f^{-1}(-\infty,f(s)$ that contain $x$ and $y$. Let $E_1$ be the set of edges that descend from $s$ in $\M_\varepsilon(f)$ whose lower endpoints lie in $S_1$. Clearly, if $e \in E_1$ and $z \in \pi^{-1}_{\M_\varepsilon(f)}(e)$, then $z \in S_1$. It is also easy to show that any point in $S_1$ will belong to at least one edge in $E_1$. Thus, $S_1 = \bigcup_{e \in E_1} \pi^{-1}_{\M_\varepsilon(f)}(e)$.

Let $p$ be a path from $x$ to $y$. The path $p$ must leave $S_1$ at some point. It will cross the boundary of $S_1$ at some point $q$ in some simplex $\sigma$. Clearly, $f(q) \geq f(s)$. Since $S_1 = \bigcup_{e \in E_1} \pi^{-1}_{\M_\varepsilon(f)}(e)$, every point in the boundary of $S_1$ is part of the boundary of some $\pi^{-1}_{\M_\varepsilon(f)}(e)$ for some $e \in E_1$. Thus, there exists some $e \in E_1$ where $q$ lies on the boundary of $\pi^{-1}_{\M_\varepsilon(f)}(e)$. By \cref{lemma:contours-are-preserved}, the path $p$ must cross a point $q'$ with $f(q') > f(s) = f'(s)$. Since this holds for all paths between $x$ and $y$, it  follows that $x$ and $y$ lie in different connected commponents of $(f')^{-1}(-\infty,f'(s))$.
\end{proof}

\begin{lemma}
Let $s$ be a split saddle in $\M_\varepsilon(f)$ that pairs with minimum $m$. Then $s$ is a vertex of $\M_\varepsilon(f')$ that pairs with minimum $m$.
\label{lemma:saddles-preserved-persistence}
\end{lemma}

\begin{proof}
Let $x$ and $y$ be nodes of $\M_\varepsilon(f)$ that descend from $s$ yet are in different connected components of $f^{-1}(-\infty,f(s))$. By \cref{lemma:sublevel-set-complements-are-preserved}, these two points are in different connected components of the sub-level set $(f')^{-1}(-\infty,f(s))$. However, by \cref{lemma:sublevel-sets-are-preserved}, they lie in the same connected component of $(f')^{-1}(-\infty,f(s)]$. This implies that a merge event happens at $s$ in $f'$, so $s$ is a vertex of $\M(f')$.

We now show that $s$ pairs with minimum $m$ in $f'$. Let $S'_1$ and $S'_2$ be the two connected components of $(f')^{-1}(-\infty,f(s))$ incident to $s$. Let $S_1$ and $S_2$ be the corresponding sets for $f$. We show that the lowest point in each $S'_i$ is the same as in the $S_i$. Because $f$ and $f'$ are piecewise-linear, all the minima in their merge trees are vertices of $\X$, and so the lowest points in each set will be vertices of $\X$. Let $m_1$ be the lowest point in $S_1$. Then $m_1$ is a vertex of $\X$. From \cref{lemma:sublevel-sets-are-preserved}, a vertex $z$ of $\X$ satisfies $z \in S'_1$ if and only if $z \in S_1$. Thus, $m_1 \in S'_1$. Let $z$ be a vertex of $S_1$ with $z \not = m_1$. Then $z$ belongs to some edge $(a,b)$ of $\M_\varepsilon(f)$ that is a descendant of $s$. Notice that $f(a) \geq f(m_1)$ and $f(z) > f(m_1)$. If $f(z) > f(a)$, then from the error bounds, $f'(z) > f(a) \geq f(m_1)$. If $f(z) \leq f(a)$, then from the bounds $f'(z) \geq f(z) > f(m_1)$. Thus, each vertex in $S_1$ other than $m_1$ will have a value in $f'$ greater than $f(m_1) = f'(m_1)$. So $m_1$ is the lowest point in $S_1'$. Similarly, $m_2$ will be the lowest point in $S_2'$.

Between the two lowest points in $S_1$ and $S_2$, $s$ will pair (in $f$) with whichever has a higher function value. Similarly, between the two lowest points in $S'_1$ and $S'_2$, $s$ will pair (in $f'$) with whichever has a higher function value. Since the lowest points in $S_1$ and $S_2$ are the same as in $S'_1$ and $S'_2$ (namely, $m_1$ and $m_2$). Since $m_1$ and $m_2$ are minima in $\M_\varepsilon(f)$, by construction $f'(m_1) = f(m_1)$ and $f'(m_2) = f(m_2)$. Thus, $s$ will pair with the same point in $f$ and $f'$. Call it $m$

Since $f'(s) = f(s)$ and $f'(m) = f(m)$, the pair $(s,m)$ will have the same persistence in both $f$ and $f'$, namely, $f(s) - f(m)$. And therefore, the pair $(s,m)$ will have a persistence above $\varepsilon$ in $f'$ and both will be vertices of $\M_\varepsilon(f')$.
\end{proof}

\color{black}
\begin{lemma}
Let $m_1, m_2$ be local minima of $f$ with $f(m_1) > f(m_2)$. Suppose that $m_1$ has persistence $\varepsilon$. Let $p:[0,1] \rightarrow \X$ be a path from $m_1$ to $m_2$. Then
\[
\max_{t \in [0,1]} (f \circ p)(t) - f(m_1) \geq \varepsilon.
\]
\label{lemma:persistence-bound}
\end{lemma}

\begin{proof}
Let $s$ be the least common ancestor of $m_1$ and $m_2$ in $\M(f)$, as illustrated in \cref{fig:persistence-bound}. For any threshold $k \leq f(s)$, nodes $m_1$ and $m_2$ lie in distinct connected components of $f^{-1}(-\infty,k)$. Therefore, any path $p$ from $m_1$ to $m_2$ must pass through some point $q$ with $f(q) \geq f(s)$. Hence,
\[
\max_{t \in [0,1]} (f \circ p)(t) \geq f(q) \geq f(s).
\]

Let $S$ be the connected component of $f^{-1}(-\infty,f(s)]$ that contains $s$. 
Since $f(m_1) > f(m_2)$, $m_1$ is not the global minimum of $S$. Thus, $m_1$ is paired with some saddle $s_1 \in S$ satisfying $f(s_1) \leq f(s)$.

Combining these inequalities, we obtain
\begin{align*}
\max_{t \in [0,1]} (f \circ p)(t) - f(m_1)
&\geq f(s) - f(m_1) \\
&\geq f(s_1) - f(m_1) \\
&= \varepsilon.
\end{align*}

\begin{figure}[!t]
\begin{center}
\includegraphics[width=0.5\linewidth]{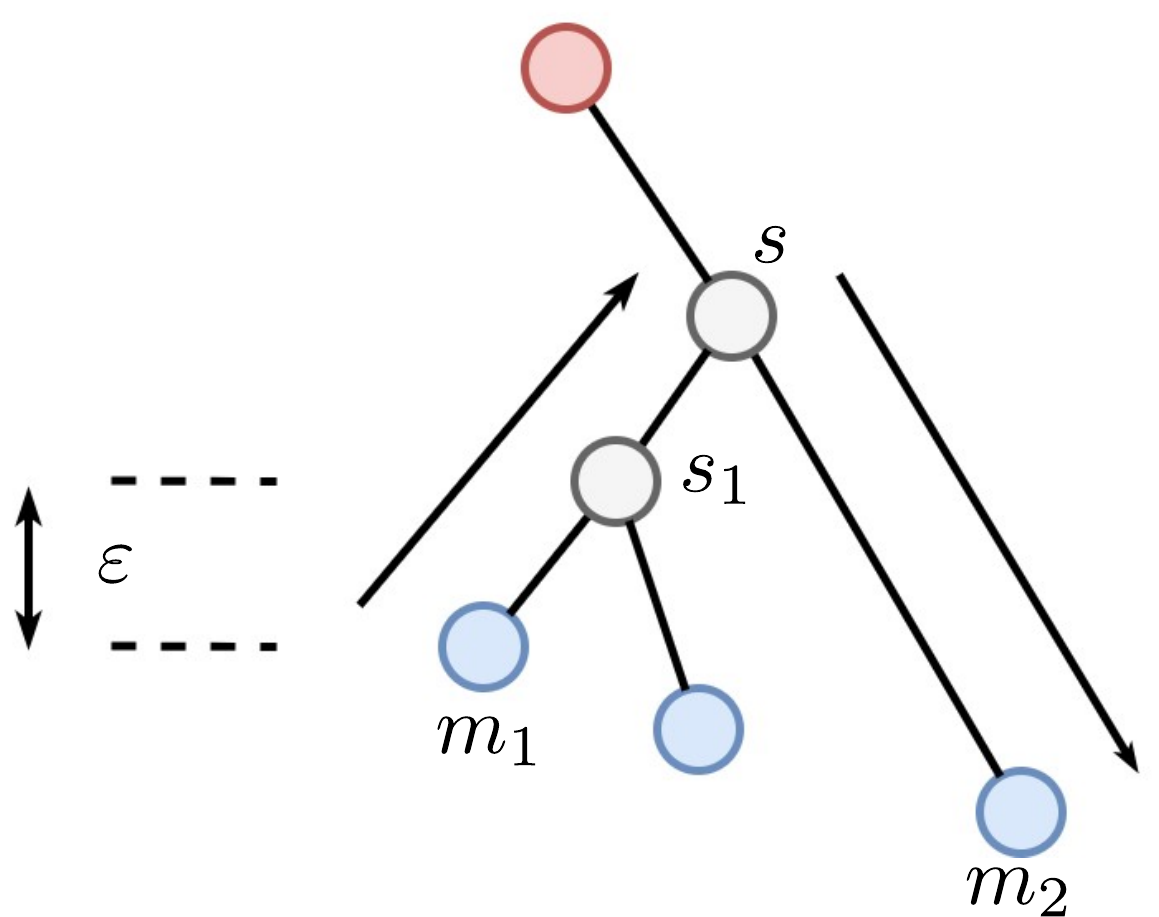}
\end{center}
\vspace{-6mm}
\caption{An illustration of the proof of \cref{lemma:persistence-bound}.}
\label{fig:persistence-bound}
\vspace{-4mm}
\end{figure}
\end{proof}

\begin{lemma}
Suppose that $m$ is a minimum of $f'$ and is not a node of $\T_\varepsilon(f)$. Then $m$ has persistence at most $\varepsilon + 2\xi$ in $f'$.
\label{lemma:false-positive-persistence-bound}
\end{lemma}

\begin{proof}
Let $m_1$ be the lowest descendant of $\pi(m)$ in $\M(f)$, and let $m_2$ be the lowest descendant of $m$ in $\M_\varepsilon(f)$. Then
\[
f(m) \geq f(m_1) \geq f(m_2).
\]
Let $s$ be the least common ancestor of $m_1$ and $m_2$ in $\M(f)$ (see \cref{fig:false-positive-persistence-bounds}).

\begin{figure}[!ht]
\vspace{-2mm}
\begin{center}
\includegraphics[width=0.8\linewidth]{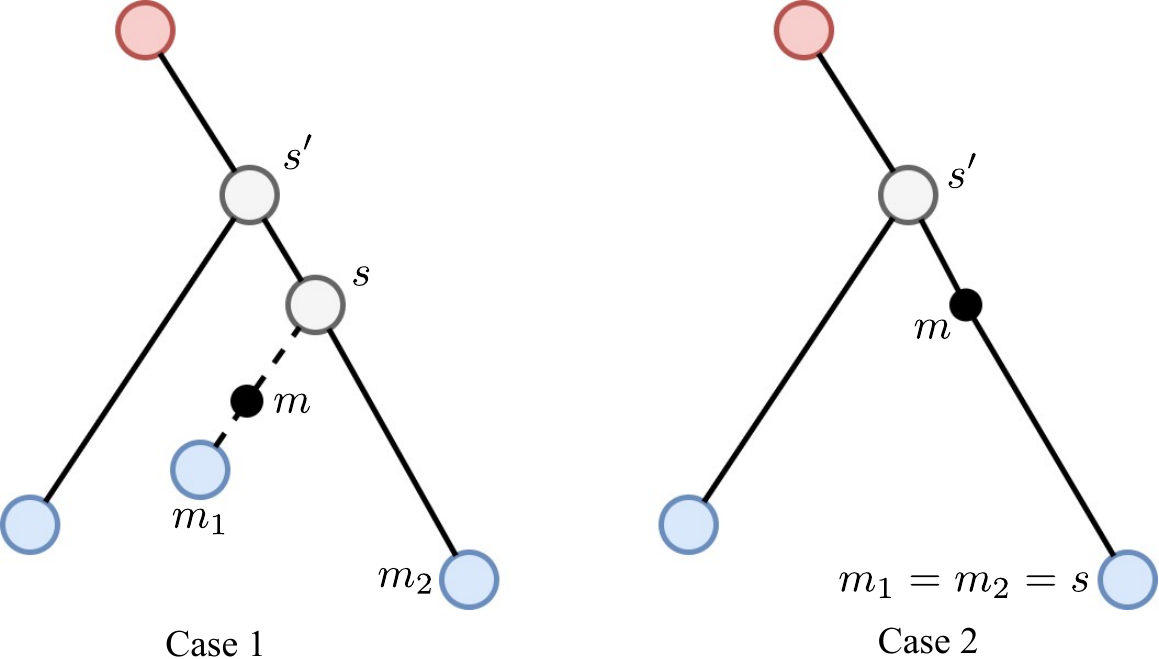}
\end{center}
\vspace{-4mm}
\caption{Illustration of the objects constructed in the proof of \cref{lemma:false-positive-persistence-bound}.}
\label{fig:false-positive-persistence-bounds}
\vspace{-4mm}
\end{figure}

Let $p_1 : [0,1] \rightarrow \X$ be a path from $m_1$ to $m_2$ passing through $s$ that never exceeds the value $f(s)$. Such a path exists because $m_1$ and $m_2$ lie in the same connected component of $f^{-1}(-\infty,f(s)]$. Similarly, let $p_2 : [0,1] \rightarrow \X$ be a path from $m$ to $m_1$ that never exceeds $f(m)$. Concatenating $p_2$ with $p_1$ yields a path $p$ from $m$ to $m_2$, whose maximum function value is $\max\{f(m), f(s)\}$.

Let $s'$ be the parent of $\pi_{\M_\varepsilon(f)}(m)$ in $\M_\varepsilon(f)$. Then $m$ and $m_2$ lie in the same connected component of $f^{-1}(-\infty,f(s')]$, and $m_2$ is the minimum of this component. In $\T_\varepsilon(f)$, $m_2$ lies on some edge $(a,b)$ with $f(a) \geq f(m_2)$. By the $\varepsilon$-boundedness of $f'$, it follows that either $f'(m) > f(a) \geq f(m_2)$ or $f'(m) \geq f(m) > f(m_2)$. Since $f'(m_2) = f(m_2)$, we conclude that $f'(m) > f'(m_2)$.

Moreover, by \cref{lemma:saddles-preserved-persistence}, $m_2$ is a minimum of $f'$. We now consider two cases.

\underline{Case 1: $m_1 \neq m_2$.} In this case, $p$ attains a maximum value of $f(s)$ in $f$, and hence at most $f(s) + \xi$ in $f'$. Also, $f'(m) \geq f(m) - \xi$. Since $p$ connects the minimum $m$ to the lower minimum $m_2$, \cref{lemma:persistence-bound} implies that the persistence of $m$ in $f'$ is at most
\[
f(s) + \xi - (f(m) - \xi) = f(s) - f(m) + 2\xi.
\]

\underline{Case 2: $m_1 = m_2$.} In this case, $s = m_2$, and the path $p$ decreases monotonically from $m$ to $m_2$, never exceeding $f(m)$. Thus, in $f'$, the maximum value along $p$ is at most $f(m) + \xi$, and $f'(m) \geq f(m) - \xi$. Applying \cref{lemma:persistence-bound}, the persistence of $m$ in $f'$ is at most
\[
(f(m) + \xi) - (f(m) - \xi) = 2\xi.
\]

Finally, since $s$ is removed during persistence simplification, the persistence of the pair $(m_1,s)$ is at most $\varepsilon$, i.e., $f(s) - f(m_1) \leq \varepsilon$. Combining the above bounds yields that the persistence of $m$ in $f'$ is at most
\[
f(s) - f(m) + 2\xi \leq f(s) - f(m_1) + 2\xi \leq \varepsilon + 2\xi.
\]

This completes the proof.
\end{proof}

\subsubsection{Proof of Main Result}
\label{appendix:proof-contour-tree-main-result}

\begin{theorem}
{\tool} correctly preserves simplified contour trees.
\end{theorem}

\begin{proof}
Since preservation of the simplified contour tree does not involve star correction, it is not necessary to establish nondegeneracy of $C_{\Phi,f}$.

We construct $f'$ to be $(\varepsilon - 2\xi)$-bounded. By \cref{lemma:saddles-preserved-persistence}, all extrema and saddles in $\M_\varepsilon(f)$ will be present in $\M_\varepsilon(f')$, while \cref{lemma:sublevel-sets-are-preserved} and \cref{lemma:sublevel-set-complements-are-preserved} demonstrates that all ancestor-descendant relationships will be preserved between them. Conversely, by \cref{lemma:false-positive-persistence-bound}, any extremum of $f'$ that is not a node of $\T_\varepsilon(f)$ has persistence at most $(\varepsilon - 2\xi) + 2\xi = \varepsilon$. Thus, each saddle in $\T_\varepsilon(f)$ pairs with an extremum in $\T_\varepsilon(f)$, and no additional extremum--saddle pairs with persistence exceeding $\varepsilon$ are introduced.

Moreover, we can use \cref{lemma:sublevel-sets-are-preserved} and \cref{lemma:sublevel-set-complements-are-preserved} to show that $\ol{\M_\varepsilon(f')}$ will match $\ol{\M_\varepsilon(f)}$ and $\ol{\M_\varepsilon(-f')}$ will match $\ol{\M_\varepsilon(-f)}$. It follows that $\T_\varepsilon(f') = \T_\varepsilon(f)$.
\end{proof}

\begin{theorem}
{\tool} correctly preserves contour trees without simplification.
\end{theorem}

\begin{proof}
The function $C_{\Phi,f}$ ensures that the set of extrema is identical for $f$ and $f'$. The correctness and nondegeneracy of $C_{\Phi,f}$ follow by arguments analogous to those in \cref{lemma:descriptors-critical-points}.

We construct $f'$ to be $0$-bounded. By \cref{lemma:saddles-preserved-persistence}, every node of $\T(f)$ remains a node of $\T(f')$. Since $f$ and $f'$ have exactly the same set of extrema, $\T(f')$ cannot contain any additional nodes beyond those in $\T(f)$. Therefore, $\T(f)$ and $\T(f')$ have identical node sets.

Moreover, \cref{lemma:sublevel-sets-are-preserved} and \cref{lemma:sublevel-set-complements-are-preserved} ensures that $\ol{\M(f)}$ and $\ol{\M(-f)}$ will match $\ol{\M(f')}$ and $\ol{\M(-f')}$, so $\T(f)$ will match $\T(f')$.
\end{proof}

\section{Additional Experiments}
\label{appendix:additional-experiments}

In this section, we present additional experiments. Specifically:

\begin{itemize}[noitemsep,leftmargin=*]
\item We present bit-rate versus PSNR curves for every preserved feature and dataset (i.e., the counterparts of \cref{fig:rate-distortion} for all datasets) in \cref{appendix:additional-experiments-reconstruction-quality}.

\item We present a runtime breakdown of {\tool} when augmenting each base compressor, together with a throughput comparison between existing feature-preserving compressors and {\tool} augmenting each base compressor in \cref{appendix:additional-experiments-runtime-analysis}.

\item We present an ablation study comparing linear and logarithmic quantization, as well as the contribution of each pipeline stage to storage overhead and feature preservation, in \cref{appendix:additional-experiments-ablation-study}.

\item We analyze the sensitivity of {\tool} to feature-specific parameters, namely the persistence threshold $\varepsilon$ for simplified contour tree preservation and the number of isovalues for isocontour preservation, in \cref{appendix:additional-experiments-varying-feature-specific-parameters}.

\item We present bit-rate versus feature-error curves for each base compressor, together with the feature-error percentages achieved by the base compressors at bit-rates comparable to those of {\tool}, in \cref{appendix:additional-experiments-errors-yielded-by-base-compressors}.

\item We evaluate {\tool} when augmenting the neural compressor Neurcomp \cite{lu2021compressive} in \cref{appendix:additional-experiments-testing-featurez-augmenting-neurcomp}.

\item We compare {\tool} with TopoSZp \cite{agarwal2026toposzp} in \cref{appendix:additional-experiments-comparison-with-toposzp}.
\end{itemize}

To improve readability, some very large figures and tables were moved to the end of the document.

\subsection{Reconstruction Quality}
\label{appendix:additional-experiments-reconstruction-quality}

We present rate-distortion curves, showing the trade-off between bit rate and PSNR, for every compressor and feature evaluated in \cref{fig:rate-distortion-appendix}. We include four variants of {\tool}: {\tool} (ZFP), {\tool} (SZ3), {\tool} (SPERR), and {\tool} (HPEZ). Where applicable, we also compare against state-of-the-art topology-preserving compressors TopoA, pMSz, QPET, and LOPC. For reference, we include the baseline compressors ZFP, SZ3, SPERR, and HPEZ, which do not enforce feature preservation. For each (feature, dataset, compressor) combination, we remove dominated points---i.e., those with both higher bit rate and lower PSNR than another point---thereby filtering out strictly suboptimal configurations; see \cref{appendix:parameters}.

\subsection{Runtime Analysis}
\label{appendix:additional-experiments-runtime-analysis}

\subsubsection{Runtime for {\tool} (ZFP), {\tool} (SPERR), and {\tool} (HPEZ)}

We report the runtime of {\tool} (ZFP), {\tool} (SPERR), and {\tool} (HPEZ) across different stages of the pipeline in \cref{tab:time-disaggregation-others}. 
Compression and decompression throughput for ZFP is reported in \cref{tab:compression-throughput-zfp} and \cref{tab:decompression-throughput-zfp}, respectively, together with those of other feature-preserving compressors that augment ZFP, as well as LOPC. Corresponding results for SPERR and HPEZ are reported in \cref{tab:compression-throughput-sperr} and \cref{tab:decompression-throughput-sperr}, and \cref{tab:compression-throughput-hpez} and \cref{tab:decompression-throughput-hpez}, respectively. Since pMSz is an experimental implementation that supports only compression, we do not report decompression throughput. We also note that pMSz is available only as a GPU implementation.

\begin{center}

\begin{minipage}{\linewidth}
\captionof{table}{Compression throughput for {\tool} (ZFP) across nine datasets while preserving seven features. \textbf{C.T.}: simplified contour tree ($\varepsilon = 0.04$). \textbf{M.S.}: Morse--Smale segmentation. \textbf{L.O.}: local order. \textbf{C.P.}: critical points. \textbf{Per.}: percentiles. \textbf{Com.}: combined feature of the simplified contour tree ($\varepsilon = 0.04$) and Morse--Smale segmentation.}
\vspace{-3mm}
\resizebox{\linewidth}{!}{\begin{tabular}{l|ccccccc|ccc}\hline
dataset & C.T. & M.S. & L.O. & C.P. & I.S. & Per. & Com. & T.A. & pMSz & LOPC\\ \hline
QMCPACK & 7.8 & 28.1 & 28.8 & 43.1 & 66.5 & 37.6 & 5.5 & 2.5 & 863.2 & 168.0\\
Viscous Fingering & 7.7 & 7.0 & 7.2 & 10.2 & 54.7 & 11.9 & 4.0 & 2.2 & 183.7 & 5.6\\
Tangaroa & 9.8 & 14.9 & 12.8 & 25.4 & 49.3 & 25.7 & 7.2 & 2.3 & 188.0 & 102.3\\
Earthquake & 11.4 & 8.8 & 13.0 & 13.3 & 52.7 & 28.9 & 5.8 & 2.6 & 649.5 & 150.8\\
Ionization & 19.9 & 6.8 & 9.3 & 8.0 & 57.1 & 37.3 & 7.1 & 3.0 & 91.6 & 57.0\\
Isabel & 7.6 & 9.1 & 11.6 & 20.1 & 33.0 & 14.8 & 4.3 & 2.0 & 1104.5 & 104.7\\
Miranda & 5.1 & 1.9 & 2.9 & 3.0 & 115.7 & 3.9 & 1.8 & 1.1 & 400.2 & 65.1\\
S3D & 17.0 & 12.5 & 27.2 & 42.6 & 69.5 & 39.9 & 12.4 & 4.1 & 55.0 & 139.0\\
SCALE-LETKF & 11.8 & 7.2 & 11.7 & 14.3 & 47.7 & 26.8 & 6.9 & 2.6 & 476.5 & 206.5\\
\hline
\end{tabular}}
\vspace{3mm}
\label{tab:compression-throughput-zfp}
\end{minipage}

\begin{minipage}{\linewidth}
\captionof{table}{Decompression throughput for {\tool} (ZFP) across nine datasets while preserving seven features. Labels are consistent with \cref{tab:compression-throughput-zfp}.}
\vspace{-3mm}
\resizebox{\linewidth}{!}{\begin{tabular}{l|ccccccc|cc}\hline
dataset & C.T. & M.S. & L.O. & C.P. & I.S. & Per. & Com. & T.A. & LOPC\\ \hline
QMCPACK & 107.8 & 99.4 & 121.0 & 136.8 & 140.8 & 124.4 & 91.1 & 85.6 & 470.9\\
Viscous Fingering & 125.5 & 97.4 & 85.5 & 97.5 & 121.0 & 93.1 & 94.2 & 118.8 & 442.8\\
Tangaroa & 84.6 & 73.1 & 76.7 & 105.0 & 104.2 & 83.0 & 76.2 & 118.6 & 353.3\\
Earthquake & 118.0 & 116.4 & 100.0 & 126.9 & 149.6 & 109.1 & 120.7 & 143.7 & 492.8\\
Ionization & 163.3 & 111.0 & 122.1 & 167.1 & 169.0 & 142.8 & 106.2 & 222.0 & 472.0\\
Isabel & 68.2 & 49.9 & 65.6 & 95.6 & 99.8 & 65.4 & 53.4 & 115.4 & 367.9\\
Miranda & 230.0 & 75.3 & 70.7 & 77.9 & 218.9 & 83.6 & 82.8 & 295.1 & 506.0\\
S3D & 165.8 & 166.4 & 158.3 & 193.4 & 203.7 & 170.5 & 164.6 & 246.4 & 525.9\\
SCALE-LETKF & 138.0 & 80.6 & 80.5 & 127.9 & 177.9 & 136.6 & 79.1 & 203.0 & 542.1\\
\hline
\end{tabular}}
\vspace{3mm}
\label{tab:decompression-throughput-zfp}
\end{minipage}

\begin{minipage}{\linewidth}
\captionof{table}{Compression throughput for {\tool} (SPERR) across nine datasets while preserving seven features. Labels are consistent with \cref{tab:compression-throughput-zfp}.}
\vspace{-3mm}
\resizebox{\linewidth}{!}{\begin{tabular}{l|ccccccc|ccc}\hline
dataset & C.T. & M.S. & L.O. & C.P. & I.S. & Per. & Com. & T.A. & LOPC & QPET\\ \hline
QMCPACK & 8.1 & 30.9 & 29.4 & 50.8 & 52.8 & 43.2 & 6.2 & 2.5 & 168.0 & 89.3\\
Viscous Fingering & 7.6 & 6.3 & 6.4 & 6.5 & 45.5 & 11.3 & 3.8 & 1.9 & 5.6 & 84.4\\
Tangaroa & 10.3 & 13.6 & 12.0 & 22.5 & 41.5 & 22.5 & 7.1 & 2.3 & 102.3 & 53.5\\
Earthquake & 11.8 & 8.2 & 11.0 & 11.1 & 55.9 & 32.0 & 5.6 & 2.6 & 150.8 & 55.0\\
Ionization & 19.5 & 6.1 & 8.4 & 8.5 & 47.0 & 34.0 & 6.7 & 2.5 & 57.0 & 36.5\\
Isabel & 11.0 & 13.9 & 14.9 & 26.0 & 55.9 & 25.3 & 5.1 & 2.1 & 104.7 & 66.2\\
Miranda & 5.1 & 1.8 & 2.8 & 2.7 & 63.4 & 4.0 & 1.8 & 0.9 & 65.1 & 53.7\\
S3D & 19.8 & 14.1 & 33.2 & 58.7 & 120.0 & 60.8 & 15.1 & 4.3 & 139.0 & 134.5\\
SCALE-LETKF & 14.2 & 6.8 & 11.9 & 15.0 & 98.7 & 37.1 & 6.7 & 2.1 & 206.5 & 52.6\\
\hline
\end{tabular}}
\vspace{3mm}
\label{tab:compression-throughput-sperr}
\end{minipage}

\begin{minipage}{\linewidth}
\captionof{table}{Decompression throughput for {\tool} (SPERR) across nine datasets while preserving seven features. Labels are consistent with \cref{tab:compression-throughput-zfp}.}
\vspace{-3mm}
\resizebox{\linewidth}{!}{\begin{tabular}{l|ccccccc|ccc}\hline
dataset & C.T. & M.S. & L.O. & C.P. & I.S. & Per. & Com. & T.A. & LOPC & QPET\\ \hline
QMCPACK & 148.3 & 122.4 & 139.1 & 173.0 & 173.8 & 143.8 & 121.2 & 89.8 & 470.9 & 259.2\\
Viscous Fingering & 77.1 & 63.3 & 67.9 & 88.3 & 119.7 & 75.6 & 61.6 & 72.1 & 442.8 & 128.4\\
Tangaroa & 78.5 & 64.3 & 63.3 & 92.3 & 91.5 & 76.5 & 67.1 & 90.4 & 353.3 & 124.1\\
Earthquake & 110.0 & 56.9 & 82.5 & 120.1 & 154.7 & 88.6 & 61.5 & 113.9 & 492.8 & 99.7\\
Ionization & 146.7 & 91.1 & 99.8 & 137.7 & 147.7 & 122.1 & 97.8 & 157.1 & 472.0 & 43.9\\
Isabel & 84.5 & 60.0 & 79.1 & 100.9 & 101.7 & 80.9 & 61.6 & 105.8 & 367.9 & 149.8\\
Miranda & 163.3 & 63.6 & 60.2 & 64.0 & 168.3 & 64.3 & 62.7 & 205.5 & 506.0 & 92.7\\
S3D & 190.3 & 169.6 & 182.8 & 208.7 & 214.8 & 194.2 & 167.7 & 254.8 & 525.9 & 286.6\\
SCALE-LETKF & 140.8 & 71.5 & 72.6 & 122.1 & 192.9 & 130.3 & 71.5 & 173.6 & 542.1 & 83.2\\
\hline
\end{tabular}}
\vspace{3mm}
\label{tab:decompression-throughput-sperr}
\end{minipage}

\arrayrulecolor{black}

\begin{minipage}{\linewidth}
\captionof{table}{Compression throughput for {\tool} (HPEZ) across nine datasets while preserving seven features. Labels are consistent with \cref{tab:compression-throughput-zfp}. QPET did not finish on the Ionization dataset for unknown reasons.}
\vspace{-3mm}
\resizebox{\linewidth}{!}{\begin{tabular}{l|ccccccc|cc}\hline
dataset & C.T. & M.S. & L.O. & C.P. & I.S. & Per. & Com. & LOPC & QPET\\ \hline
QMCPACK & 7.9 & 25.8 & 25.6 & 46.9 & 64.5 & 44.1 & 7.0 & 168.0 & 87.3\\
Viscous Fingering & 8.0 & 7.0 & 7.0 & 6.4 & 63.9 & 13.3 & 4.0 & 5.6 & 122.4\\
Tangaroa & 10.8 & 14.5 & 15.5 & 22.5 & 48.5 & 28.7 & 7.1 & 102.3 & 97.3\\
Earthquake & 12.4 & 9.8 & 12.6 & 10.6 & 75.6 & 42.2 & 6.0 & 150.8 & 109.9\\
Ionization & 20.8 & 8.8 & 8.5 & 8.1 & 52.1 & 42.3 & 6.7 & 57.0 & --\\
Isabel & 11.1 & 13.8 & 16.2 & 30.5 & 62.8 & 26.6 & 7.0 & 104.7 & 93.2\\
Miranda & 5.1 & 2.7 & 3.0 & 2.8 & 109.5 & 4.2 & 1.8 & 65.1 & 128.4\\
S3D & 19.6 & 33.2 & 30.5 & 67.0 & 119.8 & 54.8 & 14.4 & 139.0 & 188.2\\
SCALE-LETKF & 14.0 & 6.7 & 11.3 & 27.4 & 109.9 & 40.5 & 6.6 & 206.5 & 163.4\\
\hline
\end{tabular}}
\vspace{3mm}
\label{tab:compression-throughput-hpez}
\end{minipage}

\begin{minipage}{\linewidth}
\captionof{table}{Decompression throughput for {\tool} (HPEZ) across nine datasets while preserving seven features. Labels are consistent with \cref{tab:compression-throughput-zfp}. QPET did not finish on the Ionization dataset for unknown reasons.}
\vspace{-3mm}
\label{tab:decompression-throughput-hpez}
\resizebox{\linewidth}{!}{\begin{tabular}{l|ccccccc|cc}\hline
dataset & C.T. & M.S. & L.O. & C.P. & I.S. & Per. & Com. & LOPC & QPET\\ \hline
QMCPACK & 154.5 & 132.2 & 135.4 & 182.2 & 177.5 & 158.7 & 96.9 & 470.9 & 336.0\\
Viscous Fingering & 137.7 & 102.6 & 99.2 & 119.0 & 178.5 & 112.3 & 95.9 & 442.8 & 243.7\\
Tangaroa & 113.4 & 89.4 & 88.5 & 129.7 & 132.2 & 112.1 & 87.5 & 353.3 & 366.2\\
Earthquake & 191.1 & 105.1 & 120.9 & 131.9 & 187.5 & 133.6 & 104.9 & 492.8 & 212.2\\
Ionization & 191.4 & 118.4 & 111.3 & 162.9 & 181.8 & 158.3 & 116.5 & 472.0 & --\\
Isabel & 114.7 & 101.4 & 109.4 & 133.2 & 130.7 & 117.8 & 99.0 & 367.9 & 361.3\\
Miranda & 207.6 & 69.5 & 69.6 & 73.7 & 227.3 & 79.1 & 71.4 & 506.0 & 273.1\\
S3D & 224.4 & 193.5 & 212.3 & 242.9 & 241.6 & 230.6 & 191.3 & 525.9 & 611.9\\
SCALE-LETKF & 168.0 & 93.9 & 110.3 & 162.7 & 230.1 & 166.4 & 107.6 & 542.1 & 337.9\\
\hline
\end{tabular}}
\vspace{3mm}
\end{minipage}
\end{center}

\begin{table}[!h]
\caption{Runtimes (in seconds) for {\tool} (ZFP), {\tool} (SPERR), and {\tool} (HPEZ) across different stages of the pipeline. 
\textbf{Bounds}: upper and lower bound computation. 
\textbf{B.C.C.}: base compressor compression. 
\textbf{B.C.D.}: base compressor decompression. 
\textbf{Quant.}: initial quantization. 
\textbf{Star}: iterative star correction. 
\textbf{L.L.}: lossless storage.}
\resizebox{\linewidth}{!}{\begin{tabular}{l|cccccc|c} \hline
\textbf{Dataset} & \textbf{Bounds} & \textbf{B.C.C.} & \textbf{B.C.D.} & \textbf{Quant.} & \textbf{Star} & \textbf{L.L.} & \textbf{Total} \\ \hline
\multicolumn{8}{c}{ZFP} \\ \hline
QMCPACK & 0.47 & 0.02 & 0.01 & 0.00 & 0.05 & 0.07 & 0.62 \\
Viscous Fingering & 1.72 & 0.05 & 0.05 & 0.01 & 0.86 & 1.47 & 4.15 \\
Tangaroa & 1.87 & 0.14 & 0.08 & 0.01 & 0.61 & 0.87 & 3.58 \\
Earthquake & 1.83 & 0.08 & 0.08 & 0.01 & 1.36 & 1.50 & 4.87 \\
Ionization & 1.50 & 0.11 & 0.07 & 0.02 & 2.76 & 1.28 & 5.74 \\
Isabel & 6.22 & 0.50 & 0.31 & 0.07 & 2.05 & 6.55 & 15.70 \\
Miranda & 55.94 & 0.89 & 0.47 & 0.12 & 35.47 & 73.79 & 166.69 \\
S3D & 37.13 & 3.11 & 1.71 & 0.54 & 11.79 & 26.31 & 80.62 \\
SCALE-LETKF & 58.77 & 3.49 & 2.42 & 0.60 & 29.10 & 69.17 & 163.57 \\ \hline
\multicolumn{8}{c}{SPERR} \\ \hline
QMCPACK & 0.47 & 0.04 & 0.02 & 0.00 & 0.06 & 0.04 & 0.62 \\
Viscous Fingering & 1.72 & 0.15 & 0.14 & 0.01 & 1.04 & 1.31 & 4.37 \\
Tangaroa & 1.86 & 0.33 & 0.22 & 0.01 & 0.75 & 0.49 & 3.67 \\
Earthquake & 1.84 & 0.34 & 0.25 & 0.01 & 1.63 & 0.97 & 5.03 \\
Ionization & 1.51 & 0.28 & 0.17 & 0.02 & 3.16 & 0.94 & 6.09 \\
Isabel & 6.16 & 1.17 & 0.79 & 0.07 & 2.25 & 1.79 & 12.23 \\
Miranda & 56.12 & 2.21 & 1.48 & 0.13 & 38.24 & 71.31 & 169.49 \\
S3D & 37.40 & 4.76 & 2.79 & 0.52 & 12.05 & 8.73 & 66.27 \\
SCALE-LETKF & 58.20 & 10.24 & 7.83 & 0.59 & 44.23 & 47.50 & 168.62 \\ \hline
\multicolumn{8}{c}{HPEZ} \\ \hline
QMCPACK & 0.45 & 0.04 & 0.02 & 0.00 & 0.07 & 0.05 & 0.63 \\
Viscous Fingering & 1.70 & 0.08 & 0.07 & 0.01 & 1.06 & 1.30 & 4.22 \\
Tangaroa & 1.85 & 0.16 & 0.10 & 0.01 & 0.88 & 0.64 & 3.64 \\
Earthquake & 1.81 & 0.14 & 0.09 & 0.01 & 1.68 & 0.94 & 4.67 \\
Ionization & 1.51 & 0.16 & 0.09 & 0.02 & 3.29 & 1.00 & 6.08 \\
Isabel & 6.20 & 0.70 & 0.43 & 0.07 & 2.40 & 3.05 & 12.86 \\
Miranda & 56.14 & 1.22 & 0.82 & 0.13 & 37.92 & 72.68 & 168.91 \\
S3D & 37.16 & 3.93 & 1.87 & 0.52 & 13.30 & 12.70 & 69.53 \\
SCALE-LETKF & 58.09 & 4.75 & 3.94 & 0.59 & 49.13 & 54.18 & 170.71 \\ \hline
\end{tabular}}
\label{tab:time-disaggregation-others}
\end{table}

\subsubsection{Error-Bound Versus Compression Time}
\label{additional-experiments-error-bound-versus-compression-time}

In \cref{fig:eb-throughput}, we visualize the trade-off between error bound and compression throughput for {\tool} augmenting ZFP, SZ3, SPERR, and HPEZ while preserving each feature on three representative datasets: QMCPACK, Viscous Fingering, and Isabel. Overall, varying the error bound can significantly impact compression throughput. However, no consistent trend emerges: increasing the error bound may lead to higher, lower, or largely unchanged throughput depending on the configuration.

\subsection{Ablation Study}
\label{appendix:additional-experiments-ablation-study}

\subsubsection{Linear- versus Logarithmic-Scaling Quantization}

We compare the current implementation of {\tool}, which uses logarithmic-scaling quantization, with a simpler variant based on linear-scaling quantization. Throughout the {\tool} pipeline, each point is assigned upper and lower bounds that may be progressively tightened. Our linear-scaling implementation is similar to that of TopoSZ \cite{yan2024toposz}. Specifically, for each point $x$, we employ standard linear-scaling quantization using intervals of width $2\xi$. If the produced value does not satisfy the upper and lower bounds for $x$, then we store $f(x)$ losslessly. TopoSZ uses intervals of width $\xi$ instead of the standard $2\xi$, but in preliminary testing we found that using the standard $2\xi$ performed better.

We evaluate this variant against the standard implementation of {\tool} while augmenting all four base compressors and preserving the seven features considered in the main paper. Results for three representative datasets are shown in \cref{fig:rate-distortion-ablation}. Overall, logarithmic-scaling quantization provides better rate--distortion performance than linear-scaling quantization. However, the two approaches perform nearly identically for Morse--Smale segmentation, local order preservation, and critical point preservation. These are precisely the features that rely solely on star correction, suggesting that the additional precision afforded by logarithmic-scaling quantization is largely unnecessary for this stage. This observation further suggests that the advantages of {\tool} over pMSz (for Morse--Smale segmentation preservation) and LOPC (for local order preservation) arise primarily from the star correction algorithm rather than the use of logarithmic-scaling quantization. In contrast, simplified contour tree preservation and isosurface preservation depend on maintaining feature-specific upper and lower bounds rather than star correction, allowing logarithmic-scaling quantization to improve compression efficiency and enabling {\tool} to remain competitive with TopoA and QPET.

An exception is {\tool} augmenting ZFP, for which linear-scaling quantization often performs slightly better than logarithmic-scaling quantization.

Compression and decompression throughput for both versions of {\tool} augmenting SZ3 is reported in \cref{tab:ablation-compression-throughput} and \cref{tab:ablation-decompression-throughput}, respectively. Compression throughput is generally comparable. Decompression throughput is also similar.

\begin{table}[!h]
\caption{Compression throughput of {\tool} augmenting SZ3 using linear- and logarithmic-scaling quantization for all preserved features on three representative datasets. Throughput is reported in MB/s. For each (feature, dataset) pair, the largest standard error bound is used. Abbreviations: C.T., contour tree ($\varepsilon = 0.04$); M.S., Morse--Smale segmentation; L.O., local order; C.P., critical points; I.S., isosurfaces; Per., percentiles; Com., combined feature (contour tree and Morse--Smale segmentation).}
\label{tab:ablation-compression-throughput}

\resizebox{\linewidth}{!}{\begin{tabular}{c|ccccccc} \hline
Dataset & C.T. & M.S. & L.O. & C.P. & I.S. & Per. & Com.\\ \hline
\multicolumn{8}{c}{Linear-Scaling Quantization Comp. Throughput (MB/s)}\\ \hline
QMCPACK & 8.5 & 28.4 & 27.3 & 42.5 & 49.2 & 37.0 & 6.9 \\
Viscous Fingering & 7.6 & 6.8 & 6.9 & 6.1 & 21.4 & 9.2 & 4.0 \\
Isabel & 6.1 & 12.7 & 16.0 & 28.2 & 61.3 & 24.4 & 6.6 \\
\hline
\multicolumn{8}{c}{Logarithmic-Scaling Quantization Comp. Throughput (MB/s)}\\ \hline
QMCPACK & 7.8 & 28.9 & 25.5 & 40.9 & 64.3 & 49.4 & 6.0 \\
Viscous Fingering & 8.1 & 6.6 & 6.4 & 5.7 & 62.1 & 12.3 & 3.9 \\
Isabel & 11.4 & 12.3 & 14.4 & 27.2 & 61.8 & 25.3 & 4.7 \\
\hline
\end{tabular}}

\end{table}

\begin{table}[!h]

\caption{Decompression throughput of {\tool} augmenting SZ3 using linear- and logarithmic-scaling quantization for all preserved features on three representative datasets. Throughput is reported in MB/s. For each (feature, dataset) pair, the largest standard error bound is used. Abbreviations: C.T., contour tree ($\varepsilon = 0.04$); M.S., Morse--Smale segmentation; L.O., local order; C.P., critical points; I.S., isosurfaces; Per., percentiles; Com., combined feature (contour tree and Morse--Smale segmentation).}

\resizebox{\linewidth}{!}{\begin{tabular}{c|ccccccc} \hline
Dataset & C.T. & M.S. & L.O. & C.P. & I.S. & Per. & Com.\\ \hline
\multicolumn{8}{c}{Linear-Scaling Quantization Decomp. Throughput (MB/s)}\\ \hline
QMCPACK & 128.4 & 128.0 & 135.3 & 182.3 & 145.9 & 122.0 & 118.3 \\
Viscous Fingering & 134.8 & 95.7 & 99.0 & 109.8 & 137.8 & 98.1 & 92.9 \\
Isabel & 55.3 & 79.6 & 105.0 & 143.2 & 143.6 & 112.3 & 75.0 \\
\hline
\multicolumn{8}{c}{Logarithmic-Scaling Quantization Decomp. Throughput (MB/s)}\\ \hline
QMCPACK & 147.6 & 143.3 & 141.2 & 174.5 & 168.2 & 157.8 & 131.4 \\
Viscous Fingering & 135.8 & 98.0 & 96.5 & 109.4 & 171.2 & 110.4 & 96.5 \\
Isabel & 117.2 & 74.5 & 102.9 & 143.2 & 144.7 & 113.4 & 79.1 \\
\hline
\end{tabular}}

\label{tab:ablation-decompression-throughput}

\end{table}

\subsubsection{Contribution of Individual Pipeline Components}

We conduct an ablation study to quantify the contributions of feature-specific upper and lower bounds, star correction, and lossless compression to the overall bit-rate and feature error. Since the unsimplified contour tree is the only feature supported by {\tool} that requires both feature-specific upper and lower bounds and star correction, we evaluate this feature while selectively disabling different stages of the pipeline.

\cref{tab:ablation-bitrate} reports the bit-rates obtained with different subsets of the pipeline enabled. Overall, maintaining the feature-specific upper and lower bounds for contour tree preservation introduces more storage overhead than star correction alone. However, the overheads of these two stages are far from additive: the combined bit-rate is substantially lower than the sum of their individual overheads. This observation suggests that many edits made to maintain the upper and lower bounds also preserve star classification, and vice versa. Lossless compression consistently reduces the final bit-rate, although the magnitude of the reduction varies across datasets.

The relative overheads of maintaining feature-specific upper and lower bounds and performing star correction depend on the feature being preserved. For the unsimplified contour tree, the upper and lower bounds ensure that every topological feature in the contour tree of $f$ appears in the contour tree of $f'$, while star correction eliminates false-positive maxima and minima. Maintaining these feature-specific bounds is inherently more restrictive than star correction in this setting, which explains its larger standalone storage overhead.

\cref{tab:ablation-error} reports the corresponding feature errors. Using only the feature-specific upper and lower bounds sometimes produces higher error than the base compressor alone. We attribute this to logarithmic-scaling quantization, which introduces many small jagged artifacts that manifest as low-persistence topological features in the contour tree.

In contrast, star correction yields substantially lower error. If the reconstructed scalar field contains the same set of minima and maxima as the ground truth, then the reconstructed and ground-truth contour trees contain the same number of vertices. Following TopoSZ \cite{yan2024toposz}, an extremum--saddle connection in the reconstructed contour tree is counted as a false positive if it is absent from the ground-truth contour tree, while a missing ground-truth connection is counted as a false negative. Consequently, each extremum can contribute at most two errors (e.g., if it is connected to the wrong saddle), implying a maximum feature error of $200\%$ when only star correction is applied. As expected, lossless compression does not affect the feature error.

\begin{table}

\caption{Bit-rates achieved by {\tool} while preserving the unsimplified contour tree on three representative datasets with different pipeline components enabled. For each dataset, the largest error bound that preserves the contour tree is used. Abbreviations: BC, unaugmented base compressor; BC/ULB, contour-tree-specific upper and lower bounds without lossless compression; BC/ULB/LL, contour-tree-specific upper and lower bounds with lossless compression; BC/St, star correction without lossless compression; BC/St/LL, star correction with lossless compression; BC/ULB/St, full {\tool} pipeline without lossless compression; Full, complete {\tool} pipeline.}

\resizebox{\linewidth}{!}{\begin{tabular}{c|ccccccc} \hline
Dataset & BC & BC/ULB & BC/ULB/LL & BC/St & BC/St/LL & Bc/ULB/St & Full\\ \hline
QMCPACK & 0.3 & 1.6 & 0.9 & 1.3 & 0.3 & 1.6 & 1.0 \\
Tangaroa & 0.5 & 2.4 & 2.0 & 1.6 & 0.6 & 2.8 & 2.3 \\
Isabel & 0.3 & 4.4 & 4.0 & 1.3 & 0.3 & 4.4 & 4.0 \\
\hline
\end{tabular}}
\label{tab:ablation-bitrate}
\end{table}

\begin{table}

\caption{Feature error percentages of {\tool} while preserving the unsimplified contour tree with different pipeline components enabled. Feature error is defined as the number of incorrectly predicted vertices in the reconstructed contour tree (false positives, false negatives, and false types) divided by the number of vertices in the ground-truth contour tree. For each dataset, the largest error bound that preserves the contour tree is used. Labels are consistent with \cref{tab:ablation-bitrate}.}

\resizebox{\linewidth}{!}{\begin{tabular}{c|ccccccc} \hline
Dataset & BC & BC/ULB & BC/ULB/LL & BC/St & BC/St/LL & Bc/ULB/St & Full\\ \hline
QMCPACK & 433.3 & 3382.4 & 3382.4 & 135.2 & 135.2 & 0.0 & 0.0 \\
Tangaroa & 138.0 & 4204.7 & 4204.7 & 66.7 & 66.7 & 0.0 & 0.0 \\
Isabel & 380.9 & 255.0 & 255.0 & 153.4 & 153.4 & 0.0 & 0.0 \\
\hline
\end{tabular}}
\label{tab:ablation-error}
\end{table}

\arrayrulecolor{black}
\color{black}

\subsection{Varying Feature-Specific Parameters}
\label{appendix:additional-experiments-varying-feature-specific-parameters}

In \cref{fig:vary-eps-and-iso-num}, we evaluate how varying the persistence threshold $\varepsilon$ for simplified contour tree preservation and the number of isovalues for isocontour preservation affects the compression ratio, compression throughput, and decompression throughput of {\tool} augmenting all four base compressors (ZFP, SZ3, SPERR, and HPEZ). For isocontour preservation, the isovalues are chosen to partition each dataset into regions of equal volume. In every experiment, we use the largest standard error bound for each (feature, dataset) pair. Overall, we observe no clear correlation between either the persistence threshold or the number of isovalues and the compression ratio, compression throughput, or decompression throughput. This suggests that the relationship between these feature-specific parameters and the resulting performance metrics is complex.

\subsection{Errors Introduced by Base Compressors}
\label{appendix:additional-experiments-errors-yielded-by-base-compressors}

In \cref{fig:eb-error}, we show the tradeoff between bit-rate and feature error for the unaugmented base compressors. Overall, the base compressors often produce high feature error at low bit-rates (e.g., below 1 bpp). As the bit-rate increases, the feature error decreases rapidly but often remains substantial even at relatively high bit-rates. As expected, features defined primarily by pointwise error bounds, such as percentiles and especially isocontours, are better preserved by the base compressors, particularly at higher bit-rates. In contrast, features defined by neighborhood-level properties (Morse--Smale segmentation, local order, and critical points) or global topological structure (contour trees) are generally less well preserved.

ZFP consistently achieves higher bit-rates and correspondingly lower feature error than the other base compressors. This behavior arises because ZFP is conservative: in practice, it often produces reconstruction errors well below the requested error bound, at the cost of higher bit-rates.

We also report error percentages for ZFP, SPERR, and HPEZ in \cref{tab:equiv-eb-others}, where their error bounds are adjusted to achieve compression ratios similar to those of {\tool}. To compute these values, we use the largest error bound for {\tool} in each trial (see \cref{tab:main-error-bounds}). Since ZFP supports only a limited set of compression ratios, we omit trials in which no error bound yields a compression ratio within a factor of two of that achieved by {\tool}. The corresponding error bounds for the base compressors are listed in \cref{tab:equivalent-eb-bounds}.

\begin{table}

\caption{Error percentages for ZFP, SPERR, and HPEZ across nine datasets at compression ratios comparable to {\tool} while preserving a given feature. 
\textbf{C.T.}: fraction of vertices that are false positives, false negatives, or misclassified in the reconstructed contour tree, relative to the number of vertices in the ground truth.
\textbf{M.S.}: fraction of points incorrectly classified in the Morse--Smale segmentation.
\textbf{L.O.}: fraction of points where the local ordering among neighbors is violated.
\textbf{C.P.}: fraction of critical points whose type (ordinary, minimum, saddle, maximum, or degenerate) is misclassified, relative to the number of critical points in the ground truth.
\textbf{Perc.}: fraction of points assigned to incorrect percentiles.}

\resizebox{\linewidth}{!}{\begin{tabular}{l|ccccc} \hline
dataset & C.T. & M.S. & L.O. & C.P. & Perc. \\ \hline
\multicolumn{6}{c}{ZFP} \\ \hline
QMCPACK & $57.1\%$ & $0.0\%$ & $1.4\%$ & $28.8\%$ & $0.3\%$\\
Viscous Fingering & $0.1\%$ & $49.4\%$ & $22.3\%$ & $43.6\%$ & $23.4\%$\\
Tangaroa & $0.5\%$ & $26.5\%$ & $0.4\%$ & $48.8\%$ & $0.2\%$\\
Earthquake & $0.0\%$ & $96.1\%$ & $14.7\%$ & $187.3\%$ & $2.9\%$\\
Ionization & $0.4\%$ & $73.7\%$ & $15.4\%$ & $189.3\%$ & $36.3\%$\\
Isabel & $0.0\%$ & $0.8\%$ & $0.6\%$ & $23.1\%$ & $0.1\%$\\
Miranda & $33.3\%$ & $98.0\%$ & $65.6\%$ & $100.5\%$ & $57.2\%$\\
S3D & $6.1\%$ & $0.5\%$ & $1.2\%$ & $13.9\%$ & $0.1\%$\\
SCALE-LETKF & $4.0\%$ & $59.5\%$ & $10.4\%$ & $51.1\%$ & $2.8\%$\\
\hline
\multicolumn{6}{c}{SPERR} \\ \hline
QMCPACK & $94.3\%$ & $6.4\%$ & $2.3\%$ & $67.9\%$ & $1.2\%$\\
Viscous Fingering & $1.2\%$ & $51.0\%$ & $30.6\%$ & $61.2\%$ & $25.8\%$\\
Tangaroa & $10.5\%$ & $27.6\%$ & $1.4\%$ & $106.8\%$ & $0.9\%$\\
Earthquake & $7.1\%$ & $96.1\%$ & $19.3\%$ & $179.6\%$ & $10.0\%$\\
Ionization & $14.8\%$ & $81.0\%$ & $46.5\%$ & $390.9\%$ & $38.2\%$\\
Isabel & $106.7\%$ & $10.4\%$ & $2.6\%$ & $115.7\%$ & $0.4\%$\\
Miranda & $33.3\%$ & $99.2\%$ & $68.7\%$ & $101.6\%$ & $60.3\%$\\
S3D & $24.5\%$ & $4.5\%$ & $1.2\%$ & $54.2\%$ & $0.3\%$\\
SCALE-LETKF & $38.8\%$ & $58.5\%$ & $16.5\%$ & $371.8\%$ & $9.4\%$\\
\hline
\multicolumn{6}{c}{HPEZ} \\ \hline
QMCPACK & $114.3\%$ & $29.0\%$ & $5.4\%$ & $78.7\%$ & $2.2\%$\\
Viscous Fingering & $1.4\%$ & $50.9\%$ & $30.7\%$ & $34.3\%$ & $25.9\%$\\
Tangaroa & $12.4\%$ & $32.1\%$ & $2.4\%$ & $74.2\%$ & $1.8\%$\\
Earthquake & $3.5\%$ & $96.1\%$ & $20.2\%$ & $191.6\%$ & $11.6\%$\\
Ionization & $9.5\%$ & $84.3\%$ & $47.7\%$ & $333.2\%$ & $37.9\%$\\
Isabel & $140.0\%$ & $30.5\%$ & $4.2\%$ & $146.2\%$ & $0.7\%$\\
Miranda & $33.3\%$ & $99.5\%$ & $69.1\%$ & $101.3\%$ & $60.3\%$\\
S3D & $31.6\%$ & $14.7\%$ & $1.9\%$ & $66.6\%$ & $0.6\%$\\
SCALE-LETKF & $55.7\%$ & $60.8\%$ & $17.3\%$ & $90.1\%$ & $3.1\%$\\
\hline
\end{tabular}}
\label{tab:equiv-eb-others}
\end{table}

\subsection{Augmenting Neurcomp with {\tool}}
\label{appendix:additional-experiments-testing-featurez-augmenting-neurcomp}

{\tool} is designed to augment both classical and neural compressors. To evaluate its compatibility with neural compression, we augment the implicit neural representation (INR)-based compressor Neurcomp. Neurcomp requires the user to specify a target compression ratio. We perform an initial grid search to identify the most effective target ratio and, consistent with previous work \cite{gorski2025general}, find that a value of 100 generally yields the best results. Accordingly, we use a target compression ratio of 100 in all experiments. In practice, however, Neurcomp does not always achieve this ratio because the compressed size of its network weights varies after lossless compression. Furthermore, augmenting Neurcomp with {\tool} introduces additional variation in the final compression ratio depending on the error bound $\xi$ and the feature being preserved.

In \cref{fig:rate-distortion-neurcomp}, we compare the rate--distortion performance of {\tool} augmenting ZFP, SZ3, SPERR, HPEZ, and Neurcomp. Across all experiments, Neurcomp consistently exhibits the worst rate--distortion performance by a substantial margin. Compression and decompression throughput of {\tool} augmented with Neurcomp on three datasets is reported in \cref{tab:compression-throughput-neurcomp} and \cref{tab:decompression-throughput-neurcomp}, respectively. Compression throughput never exceeds 0.3 MB/s, likely because training the INR dominates the compression time. Likewise, decompression throughput is typically about an order of magnitude lower than that of the other base compressors. Overall, we observe little benefit to using Neurcomp as the underlying (base) compressor for {\tool}.

\begin{table}
\caption{Compression throughput for {\tool} augmenting Neurucomp across all features on three select datasets. Throughput is reported in MB/s. C.T.: contour tree $(\varepsilon = 0.04)$, M.S.: Morse--Smale segmentation. L.O.: local order. C.P.: Critical points. I.S.: Isosurfaces. Per.: Percentiles. Com.: Combined feature (contour tree and Morse--Smale segmentation).}
\resizebox{\linewidth}{!}{\begin{tabular}{l|ccccccc|}\hline
dataset & C.T. & M.S. & L.O. & C.P. & I.S. & Per. & Com.\\ \hline
QMCPACK & 0.2 & 0.2 & 0.2 & 0.2 & 0.2 & 0.2 & 0.2\\
Viscous Fingering & 0.3 & 0.3 & 0.3 & 0.3 & 0.3 & 0.3 & 0.2\\
Isabel & 0.1 & 0.1 & 0.1 & 0.1 & 0.1 & 0.1 & 0.1\\
\hline
\end{tabular}}
\label{tab:compression-throughput-neurcomp}
\end{table}

\begin{table}
\caption{Decompression throughput for {\tool} augmenting Neurucomp across all features on three select datasets. Throughput is reported in MB/s. C.T.: contour tree $(\varepsilon = 0.04)$, M.S.: Morse--Smale segmentation. L.O.: local order. C.P.: Critical points. I.S.: Isosurfaces. Per.: Percentiles. Com.: Combined feature (contour tree and Morse--Smale segmentation).}
\resizebox{\linewidth}{!}{\begin{tabular}{l|ccccccc|}\hline
dataset & C.T. & M.S. & L.O. & C.P. & I.S. & Per. & Com.\\ \hline
QMCPACK & 3.2 & 3.3 & 3.2 & 3.3 & 3.2 & 3.3 & 3.2\\
Viscous Fingering & 9.0 & 8.4 & 8.9 & 9.3 & 9.6 & 8.9 & 8.6\\
Isabel & 10.1 & 9.3 & 10.0 & 11.1 & 10.9 & 10.0 & 8.9\\
\hline
\end{tabular}}
\label{tab:decompression-throughput-neurcomp}
\end{table}

\arrayrulecolor{black}

\subsection{Comparison with TopoSZp}
\label{appendix:additional-experiments-comparison-with-toposzp}

We additionally compare {\tool} with TopoSZp. Since TopoSZp only supports compression of 2D single-precision scalar fields, we perform this comparison on three additional 2D datasets. For each dataset, we evaluate every augmented variant of {\tool} and TopoSZp using the same set of error bounds employed for critical point preservation throughout the paper. As in our other experiments, we measure bit-rate, PSNR, compression time, and decompression time. Because TopoSZp preserves only a subset of the critical points, we also report its feature error.

The implementation of TopoSZp available to us occasionally terminated with segmentation faults. When this occurred, we reran the experiment until it completed successfully. In some cases, a single dataset/error-bound combination required more than 20 attempts before TopoSZp produced a valid result.

\para{Datasets.} A summary of the three additional datasets is provided in \cref{tab:data-toposzp}.

\begin{table}[!h]
\begin{center}
\caption{Datasets used to compare {\tool} with TopoSZp, together with their dimensions and file sizes.}
\begin{tabular}{c|cc}\hline
Dataset & Dimensions & File Size (MB) \\ \hline
S3D & $500 \times 500$ & 2.0 \\
CESM CLDHGH & $3600 \times 1800$ & 25.9 \\
CESM FLDS & $3600 \times 1800$ & 25.9 \\
\hline
\end{tabular}
\label{tab:data-toposzp}
\end{center}
\end{table}

The \textbf{S3D} dataset is obtained by slicing our 3D S3D dataset along the plane $z = 250$. We also resample the data to single precision to ensure compatibility with TopoSZp. The \textbf{CESM CLDHGH} and \textbf{CESM FLDS} datasets are obtained from the Community Earth System Model (CESM) \cite{hurrell2013community}, which is primarily funded by the U.S. National Science Foundation. Specifically, we use the CLDHGH and FLDS fields from the CESM-ATM dataset available through SDRBench \cite{zhao2020sdrbench}.

\para{Reconstruction Quality.} \cref{fig:rate-distortion-toposzp} compares the rate--distortion performance of each {\tool} variant with TopoSZp. Overall, every augmented version of {\tool} consistently outperforms TopoSZp.

\cref{fig:feature-error-toposzp} shows the tradeoff between bit-rate and the feature error reported by TopoSZp. Overall, TopoSZp exhibits considerable variation in feature error across datasets and error bounds, although the reported errors are generally low. We note, however, that TopoSZp does not define critical points according to the Morse-theoretic definition commonly used in topological data analysis. Consequently, the feature errors reported here are based on TopoSZp's definition and are not directly comparable to the feature errors reported elsewhere in this paper. A comparison of TopoSZp with unaugmented base compressors using its own error metric is provided in \cite{agarwal2026toposzp}.

\begin{figure}
\includegraphics[width=\linewidth]{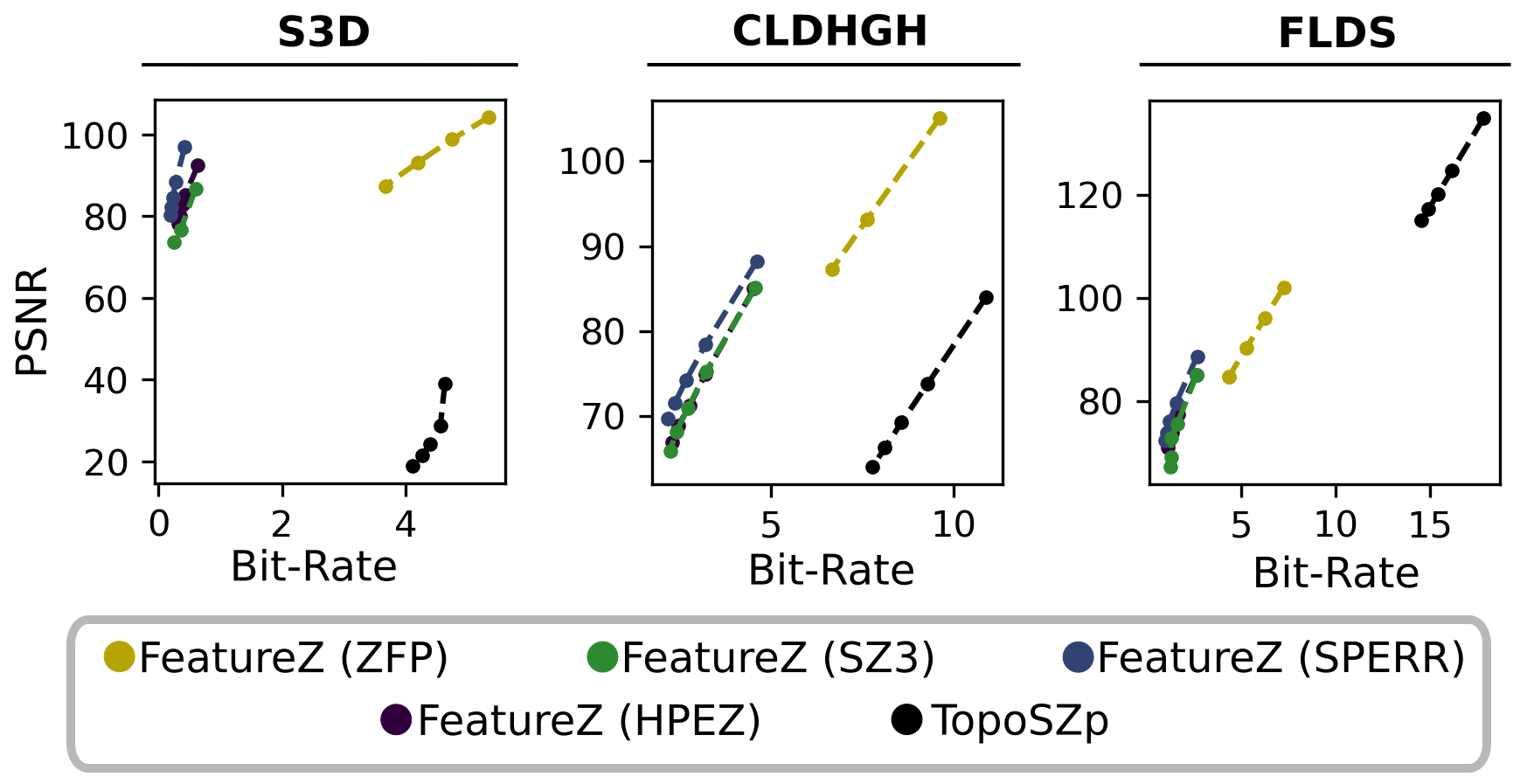}
\caption{Rate distortion curves illustrating the trade-off between bit-rate and PSNR for each {\tool} model as well as TopoSZp when preserving the critical points of three datasets.}
\label{fig:rate-distortion-toposzp}
\end{figure}

\begin{figure}
\includegraphics[width=\linewidth]{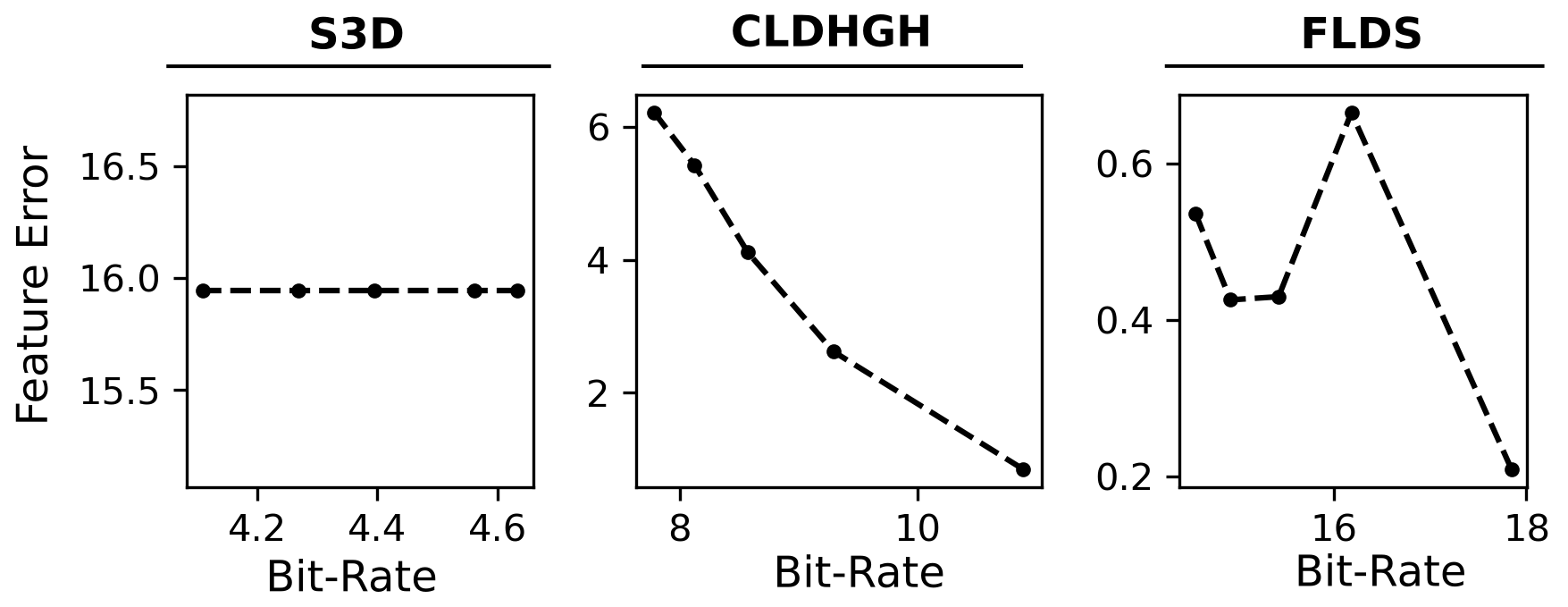}
\caption{Bit-rate versus feature-error curves for TopoSZp on each dataset. Feature error is defined as the total number of false positives, false negatives, and false types in the reconstructed data divided by the number of ground-truth critical points. For TopoSZp, all feature errors correspond to false-negative saddles.}
\label{fig:feature-error-toposzp}
\end{figure}

\para{Run Time.} We report the compression and decompression throughput of each {\tool} variant alongside TopoSZp in \cref{tab:throughputs-toposzp}. Overall, TopoSZp achieves substantially higher compression and decompression throughput than {\tool}. This result is expected, as TopoSZp is specifically designed to prioritize throughput \cite{agarwal2026toposzp}. We also observe that TopoSZp's decompression throughput is lower than its compression throughput, which is uncommon for compression algorithms. This behavior arises because TopoSZp performs much of its topological correction during decompression rather than during compression.

\begin{table}[!h]
\caption{Compression and decompression throughput (in MB/s) for each {\tool} variant along with TopoSZp on each dataset. F.Z. refers to {\tool}. }
\resizebox{\linewidth}{!}{\begin{tabular}{l|ccccc} \hline
dataset & F.Z. (ZFP) & F.Z. (SZ3) & F.Z. (SPERR) & F.Z. (HPEZ) & TopoSZp \\ \hline 
\multicolumn{6}{c}{Compression Throughput (MB/s)} \\ \hline 
S3DSlice & 35.4 & 41.0 & 44.9 & 43.9 & 991.6 \\
CLDHGH & 12.8 & 15.8 & 13.5 & 14.9 & 1045.7 \\
FLDS & 14.3 & 16.1 & 15.8 & 16.6 & 886.9 \\ \hline 
\multicolumn{6}{c}{Decompression Throughput (MB/s)} \\ \hline
S3DSlice & 131.7 & 149.9 & 161.1 & 161.0 & 264.3 \\
CLDHGH & 57.0 & 102.3 & 59.6 & 97.7 & 200.5 \\
FLDS & 71.7 & 111.8 & 79.6 & 109.9 & 223.3 \\
\hline
\end{tabular}}
\label{tab:throughputs-toposzp}
\end{table}

\arrayrulecolor{black}

\section{Parameter Values for Experiments}
\label{appendix:parameters}

We list the isovalues used for isocontour preservation for each dataset in \cref{tab:isovalues} and the error bounds used in the main experiments for {\tool} in \cref{tab:main-error-bounds}.

In lossy compression, the error parameter typically controls the trade-off between compressed file size and reconstruction quality. However, for feature-preserving compressors that augment base compressors, large error bounds can produce the counterintuitive effect that increasing the error bound simultaneously decreases the bit-rate and improves PSNR \cite{gorski2025general,gorski2026tfz,li2025msz,xia2025tspsz}. This behavior likely occurs because larger error bounds permit greater distortion by the base compressor, requiring more extensive corrections to preserve the target feature and thereby increasing the storage overhead. Such parameter settings are inherently suboptimal, since a smaller error bound yields both a lower bit-rate and higher reconstruction quality.

To enable fair comparisons across feature-preserving compressors, we select error bounds that are optimal for every compressor and every feature. For each feature, we choose an integer $n$ and evaluate every compressor using five error bounds evenly spaced between $1\times10^{-n}$ and $10\times10^{-n}$. We initially set $n=4$ for all features. If many of the resulting error bounds are suboptimal for any compressor, we increase $n$ until all tested error bounds are optimal. We emphasize that increasing $n$ is frequently necessary for existing feature-preserving compressors as well as {\tool}. When a single value of $n$ is suitable for most, but not all, datasets, we increase $n$ only for the affected datasets. Finally, if only a limited number of error bounds remain suboptimal for a particular compressor, we omit those configurations for that compressor rather than increasing $n$ for all compressors.

For selected experiments, we identify error bounds for the unaugmented base compressors that yield compression ratios comparable to {\tool}. These comparisons are presented in \cref{fig:teaser}, \cref{fig:equiv-eb}, \cref{tab:equivalent-eb}, and \cref{tab:equiv-eb-others}. For each (task, dataset) pair, we run {\tool} using the largest error bound listed in \cref{tab:main-error-bounds}. The corresponding error bounds for the base compressors are reported in \cref{tab:equivalent-eb-bounds}. This table does not include the error bound used for the teaser image, which uses SZ3 with an error bound of
$2.067 \times 10^{-5}$.

\begin{table}[!]
\caption{Isovalues used for isocontour preservation across the nine datasets.}
\vspace{-3mm}
\begin{center}
\resizebox{\columnwidth}{!}{
\begin{tabular}{l|ccc} \hline
dataset & isovalue 1 & isovalue 2 & isovalue 3 \\ \hline
QMCPACK & $3.969 \times 10^{-1}$ & $4.069 \times 10^{-1}$ & $4.156 \times 10^{-1}$ \\
Viscous Fingering & $4.000 \times 10^{-4}$ & $1.000 \times 10^{-3}$ & $8.920 \times 10^{-2}$ \\
Tangaroa & $6.174 \times 10^{-1}$ & $6.262 \times 10^{-1}$ & $6.374 \times 10^{-1}$ \\
Earthquake & $1.328 \times 10^{-3}$ & $2.771 \times 10^{-3}$ & $6.565 \times 10^{-3}$ \\
Ionization & $6.063 \times 10^{-2}$ & $6.605 \times 10^{-2}$ & $8.162 \times 10^{-2}$ \\
Isabel & $8.268 \times 10^{-2}$ & $1.418 \times 10^{-1}$ & $2.250 \times 10^{-1}$ \\
Miranda & $9.000 \times 10^{-6}$ & $2.072 \times 10^{-1}$ & $1.000 \times 10^{0}$ \\
S3D & $3.638 \times 10^{-3}$ & $4.548 \times 10^{-3}$ & $5.147 \times 10^{-3}$ \\
SCALE-LETKF & $1.011 \times 10^{-5}$ & $1.003 \times 10^{-3}$ & $9.321 \times 10^{-3}$ \\
\hline
\end{tabular}
}
\end{center}
\vspace{-6mm}
\label{tab:isovalues}
\end{table}

\begin{figure*}
\begin{center}
\includegraphics[width=\linewidth]{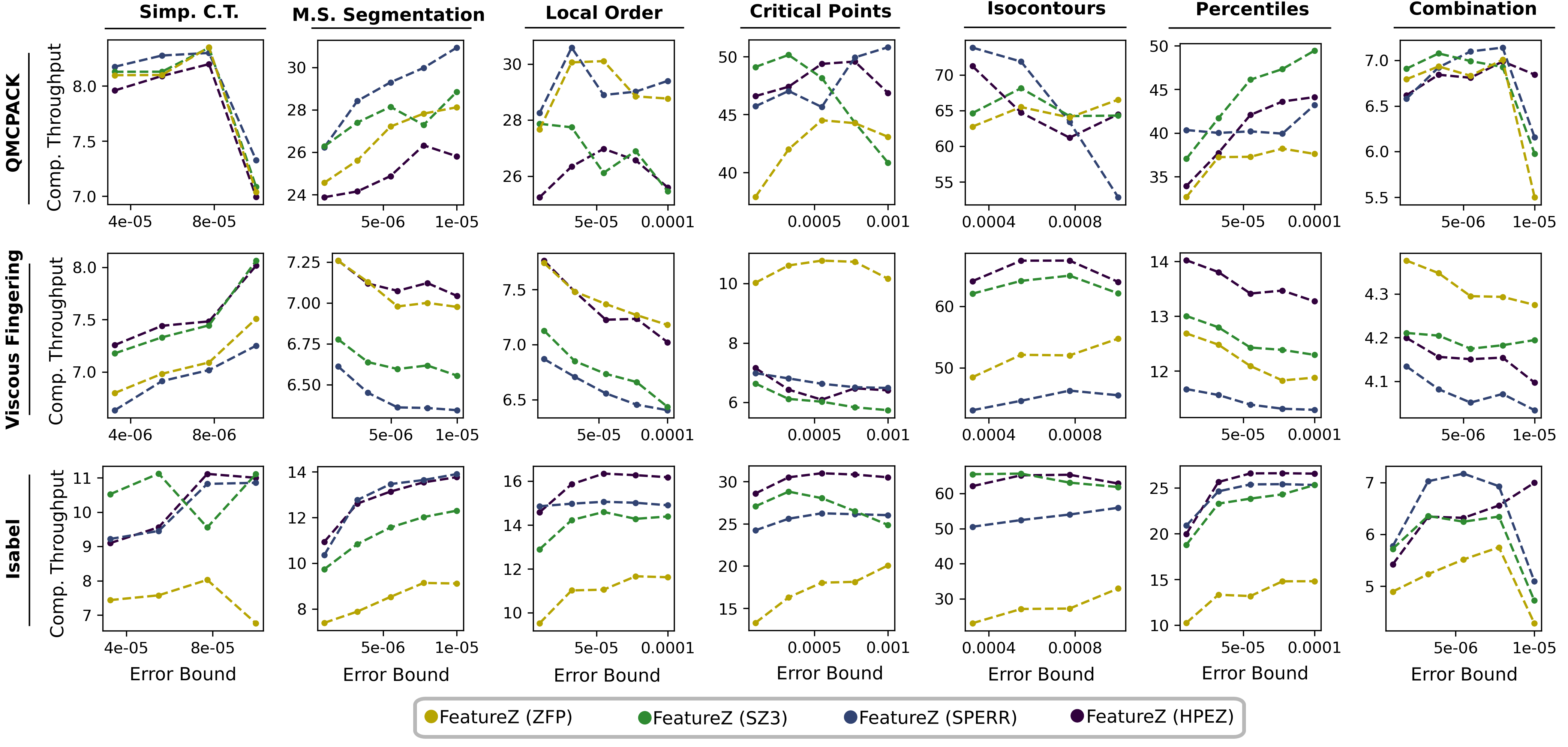}
\end{center}
\vspace{-4mm}
\caption{Compression throughput versus error bound for {\tool}, using ZFP, SZ3, SPERR, and HPEZ as base compressors, while preserving seven features on the QMCPACK, Viscous Fingering, and Isabel datasets.}
\label{fig:eb-throughput}
\vspace{-4mm}
\end{figure*}

\begin{figure*}

\includegraphics[width=\linewidth]{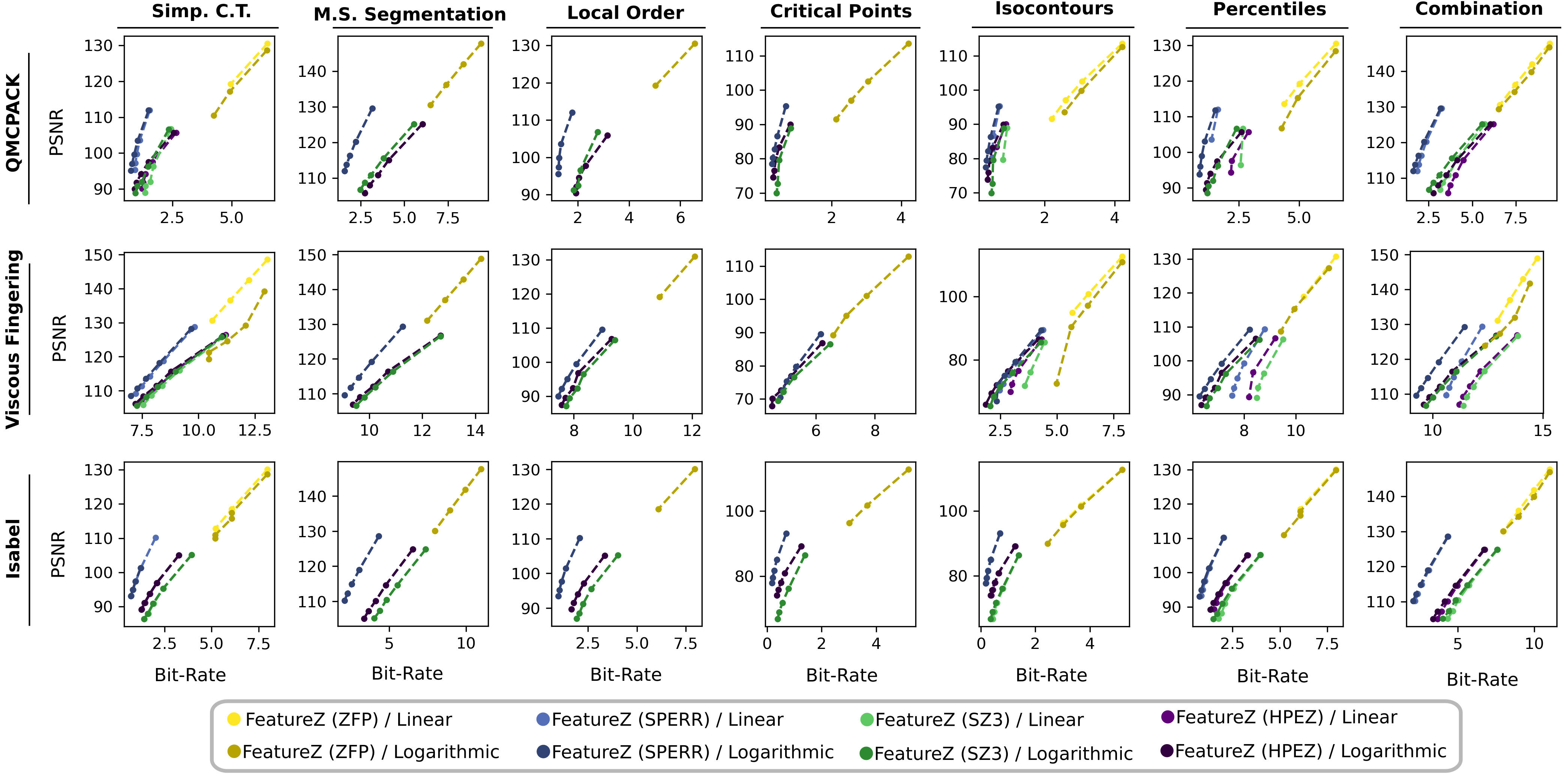}

\caption{Rate-distortion curves with illustrating the trade-off between bit-rate and PSNR for features preserved by {\tool} using linear- and logarithmic-scaling quantization. Results are shown for three representative datasets---QMCPACK, Viscous Fingering, Isabel---across seven features. The standard error bounds for each (compressor,feature) pair are used. For Morse--Smale segmentation and local ordering, using linear- and logarithmic-scaling quantization yielded nearly identical results.}
\label{fig:rate-distortion-ablation}

\end{figure*}

\begin{figure*}
\includegraphics[width=\linewidth]{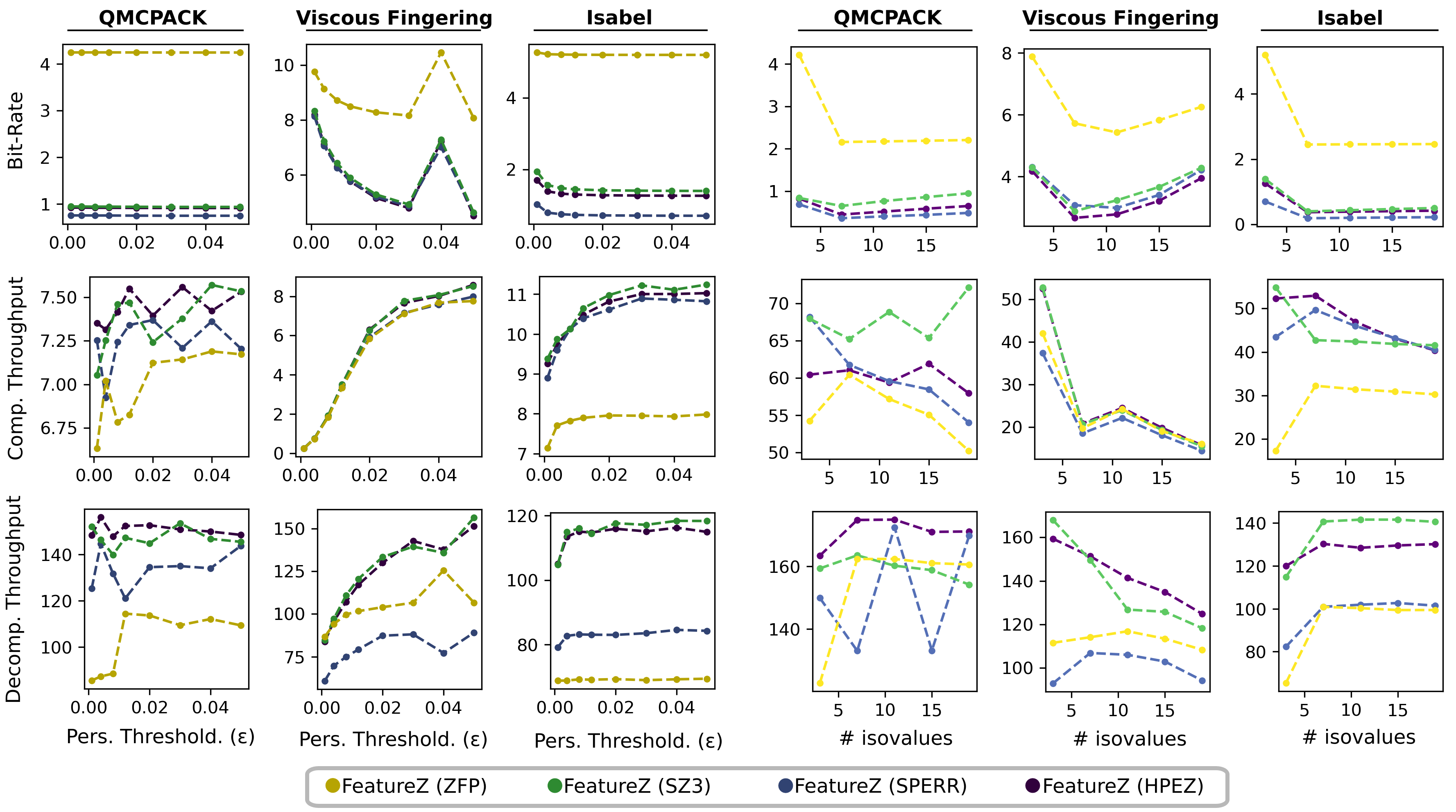}
\caption{Bit-rates, compression throughput, and decompression throughput produced by {\tool} augmenting ZFP, SZ3, SPERR, and HPEZ on three datasets when preserving the simplified contour tree and isocontours. For the simplified contour tree, the persistence threshold $\varepsilon$ is varied and results are shown for different values of $\varepsilon$. Likewise, for isocontour preservation, the number of isocontours is varied. In each case isovalues were selected to divide the datasets into equally sized regions (by volume). The highest standard error bound for each given (dataset,feature) pair is used. Compression and decompression throughput is reported in MB/s.}
\label{fig:vary-eps-and-iso-num}

\end{figure*}

\begin{figure*}

\includegraphics[width=\linewidth]{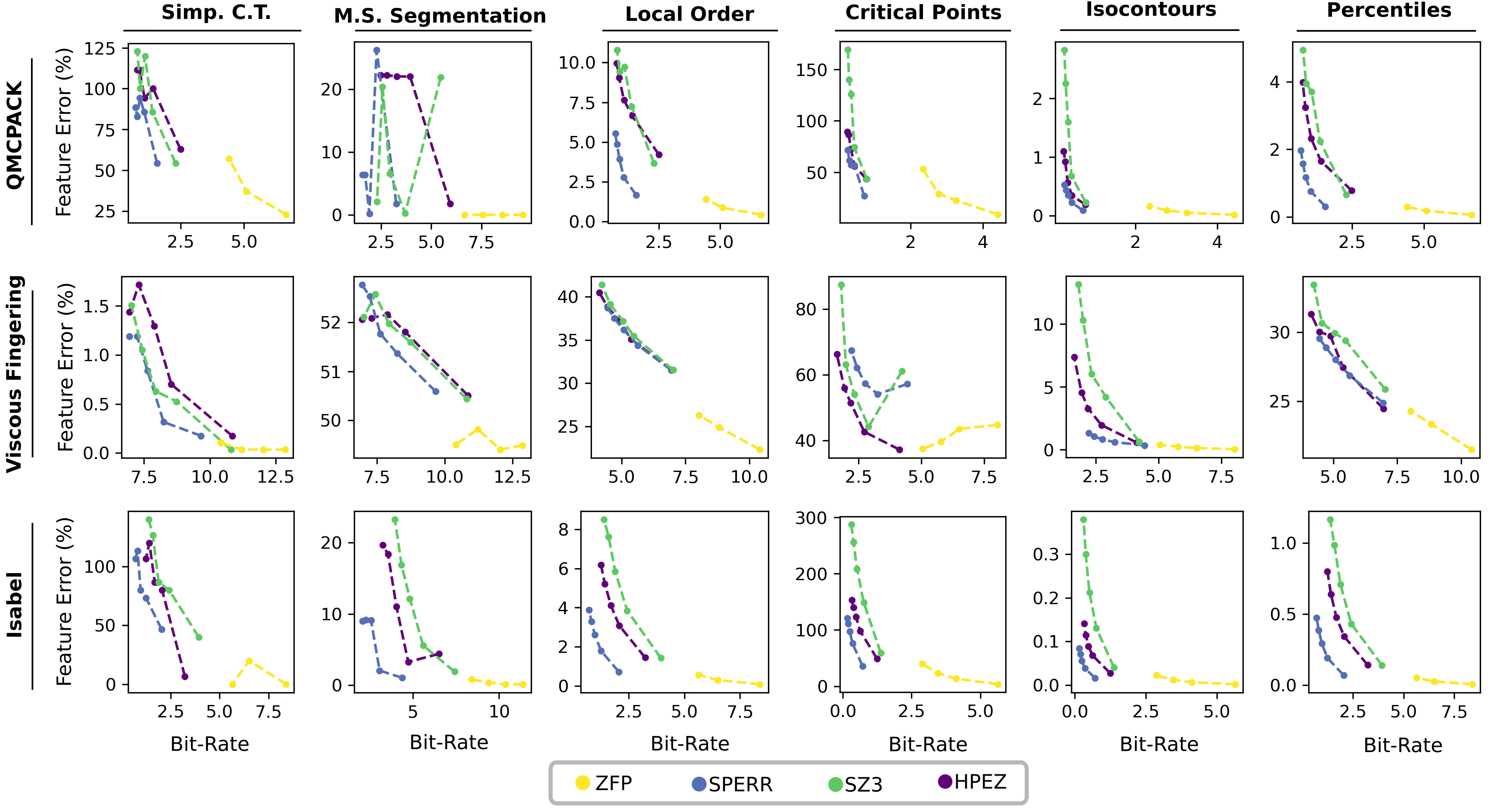}

\caption{Feature errors induced by each base compressor (not augmented by {\tool}) for select features at different bit-rates. Results are shown for three representative datasets---QMCPACK, Viscous Fingering, and Isabel---across six features. The standard error bounds used for each (compressor,feature) pair are used. For contour tree preservation, feature error is defined as the number of falsely predicted vertices (false positive, false negative, and false types) in the contour tree of the decompressed data divided by the number of vertices in the ground truth contour tree. For critical points, feature error is defined as the number of falsely predicted critical points (false positive, false negative, and false types) divided by the number of ground truth critical points.}
\label{fig:eb-error}

\end{figure*}

\begin{figure*}
\includegraphics[width=\linewidth]{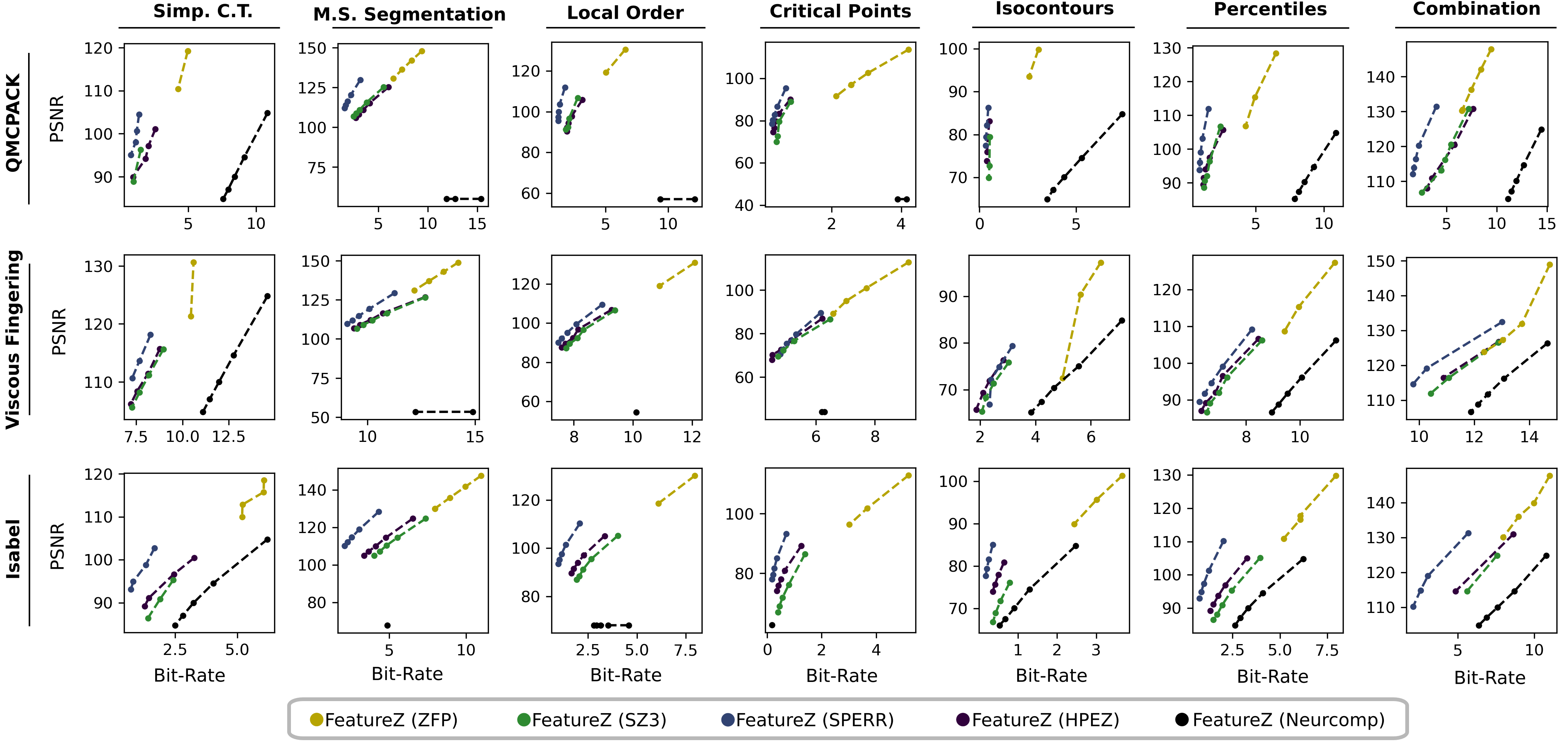}

\caption{Rate-distortion curves illustrating the trade-off between bit rate and PSNR, along with selected zoomed-in views, for seven features across all nine datasets. We include {\tool} augmenting SPERR, SZ3, ZFP, HPEZ and Neurucomp.}
\label{fig:rate-distortion-neurcomp}
\end{figure*}

\begin{figure*}
\begin{center}
\includegraphics[width=0.9\linewidth]{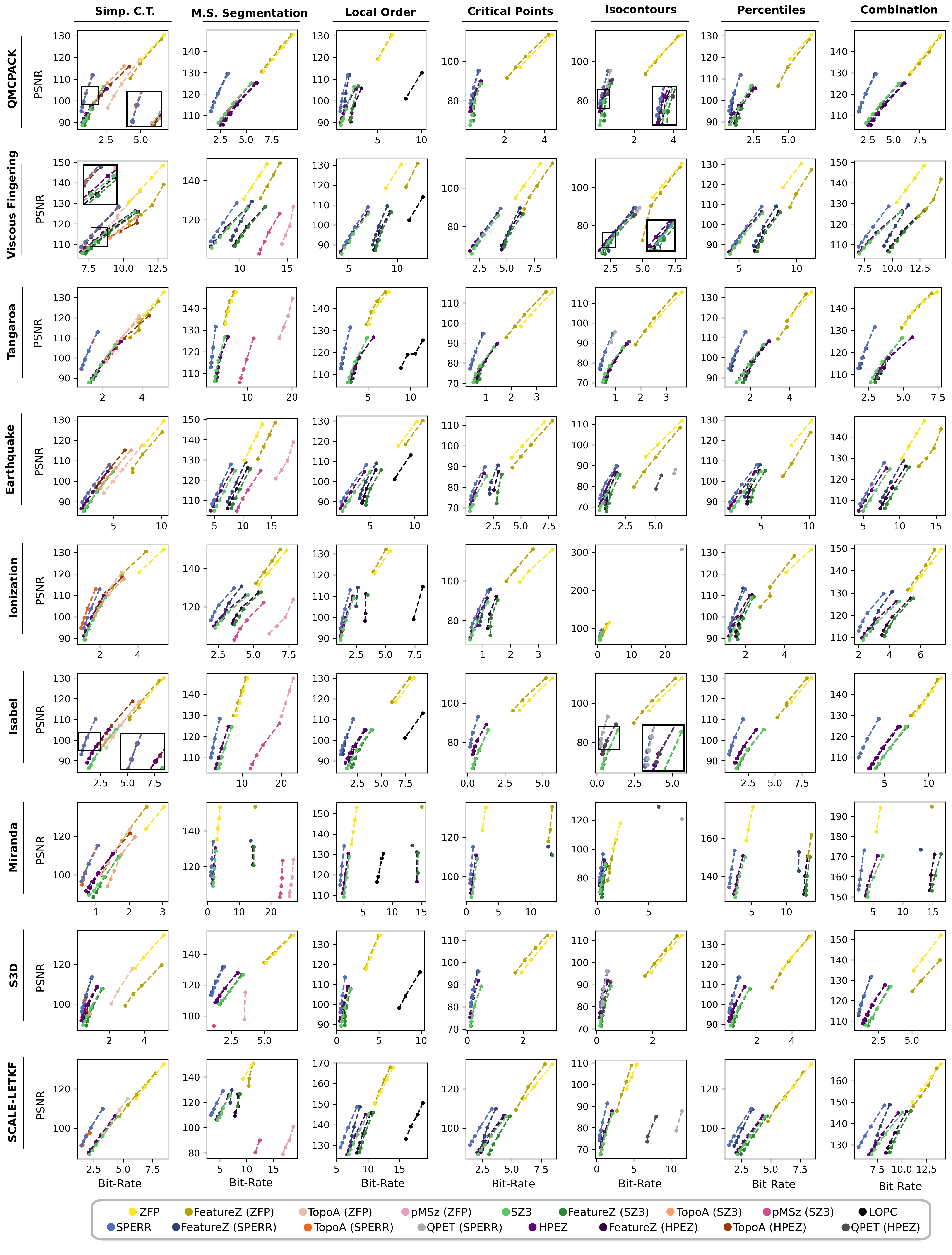}
\end{center}
\vspace{-4mm}
\caption{Rate-distortion curves illustrating the trade-off between bit rate and PSNR, along with selected zoomed-in views, for seven features across all nine datasets. We include four variants of {\tool}, namely {\tool} (ZFP), {\tool} (SZ3), {\tool} (SPERR), and {\tool} (HPEZ). Where applicable, we also compare against state-of-the-art feature-preserving compressors TopoA, pMSz, QPET, and LOPC. For reference, we include the baseline compressors ZFP, SZ3, SPERR, and HPEZ which do not enforce feature preservation.}
\label{fig:rate-distortion-appendix}
\end{figure*}

\begin{table*}
\begin{center}
\caption{Error bounds used for each base compressor to achieve compression ratios comparable to {\tool} augmenting that compressor while preserving the specified feature at the given error bound. \textbf{E.B.}: error bounds. \textbf{B.C.} denotes the base compressor. }
\begin{tabular}{l|cccc} \hline
Dataset & ZFP E.B. & SZ3 E.B. & SPERR E.B. & HPEZ E.B. \\ \hline
\multicolumn{4}{c}{Contour Tree ($\varepsilon = 0.04$)} \\ \hline
QMCPACK & $8.750 \times 10^{-5}$ & $7.170 \times 10^{-5}$ & $8.828 \times 10^{-5}$ & $7.578 \times 10^{-5}$\\
Viscous Fingering & $1.000 \times 10^{-5}$ & $8.750 \times 10^{-6}$ & $1.000 \times 10^{-5}$ & $8.750 \times 10^{-6}$\\
Tangaroa & $8.750 \times 10^{-5}$ & $9.688 \times 10^{-5}$ & $9.688 \times 10^{-5}$ & $9.063 \times 10^{-5}$\\
Earthquake & $8.750 \times 10^{-5}$ & $9.531 \times 10^{-5}$ & $9.844 \times 10^{-5}$ & $9.297 \times 10^{-5}$\\
Ionization & $8.750 \times 10^{-5}$ & $9.063 \times 10^{-5}$ & $1.000 \times 10^{-4}$ & $8.901 \times 10^{-5}$\\
Isabel & $8.750 \times 10^{-5}$ & $9.688 \times 10^{-5}$ & $9.999 \times 10^{-5}$ & $9.688 \times 10^{-5}$\\
Miranda & $8.750 \times 10^{-5}$ & $8.203 \times 10^{-5}$ & $8.672 \times 10^{-5}$ & $8.195 \times 10^{-5}$\\
S3D & $9.375 \times 10^{-5}$ & $5.771 \times 10^{-5}$ & $7.187 \times 10^{-5}$ & $5.518 \times 10^{-5}$\\
SCALE-LETKF & $8.750 \times 10^{-5}$ & $9.999 \times 10^{-5}$ & $9.844 \times 10^{-5}$ & $9.609 \times 10^{-5}$\\
\hline
\multicolumn{4}{c}{Morse-Smale Segmentation} \\ \hline
QMCPACK & $9.375 \times 10^{-6}$ & $8.359 \times 10^{-6}$ & $9.844 \times 10^{-6}$ & $8.516 \times 10^{-6}$\\
Viscous Fingering & $2.500 \times 10^{-6}$ & $1.797 \times 10^{-6}$ & $1.797 \times 10^{-6}$ & $2.207 \times 10^{-6}$\\
Tangaroa & $9.375 \times 10^{-6}$ & $6.328 \times 10^{-6}$ & $7.813 \times 10^{-6}$ & $6.870 \times 10^{-6}$\\
Earthquake & $2.344 \times 10^{-6}$ & $1.523 \times 10^{-6}$ & $1.523 \times 10^{-6}$ & $1.523 \times 10^{-6}$\\
Ionization & $9.375 \times 10^{-6}$ & $2.275 \times 10^{-6}$ & $2.207 \times 10^{-6}$ & $2.207 \times 10^{-6}$\\
Isabel & $9.375 \times 10^{-6}$ & $9.375 \times 10^{-6}$ & $9.687 \times 10^{-6}$ & $9.686 \times 10^{-6}$\\
Miranda & $4.883 \times 10^{-9}$ & $4.883 \times 10^{-9}$ & $4.883 \times 10^{-9}$ & $4.883 \times 10^{-9}$\\
S3D & $8.750 \times 10^{-6}$ & $8.984 \times 10^{-6}$ & $1.000 \times 10^{-5}$ & $8.654 \times 10^{-6}$\\
SCALE-LETKF & $3.437 \times 10^{-6}$ & $2.716 \times 10^{-7}$ & $2.332 \times 10^{-7}$ & $3.882 \times 10^{-7}$\\
\hline
\multicolumn{4}{c}{Local Order} \\ \hline
QMCPACK & $1.000 \times 10^{-4}$ & $1.711 \times 10^{-5}$ & $2.070 \times 10^{-5}$ & $1.729 \times 10^{-5}$\\
Viscous Fingering & $1.250 \times 10^{-5}$ & $6.250 \times 10^{-6}$ & $6.250 \times 10^{-6}$ & $6.250 \times 10^{-6}$\\
Tangaroa & $9.375 \times 10^{-6}$ & $6.250 \times 10^{-6}$ & $7.813 \times 10^{-6}$ & $6.797 \times 10^{-6}$\\
Earthquake & $6.250 \times 10^{-5}$ & $1.318 \times 10^{-5}$ & $1.729 \times 10^{-5}$ & $1.250 \times 10^{-5}$\\
Ionization & $6.250 \times 10^{-5}$ & $2.043 \times 10^{-6}$ & $2.908 \times 10^{-6}$ & $2.235 \times 10^{-6}$\\
Isabel & $8.750 \times 10^{-5}$ & $5.352 \times 10^{-5}$ & $5.469 \times 10^{-5}$ & $5.703 \times 10^{-5}$\\
Miranda & $4.883 \times 10^{-9}$ & $4.883 \times 10^{-9}$ & $4.883 \times 10^{-9}$ & $4.883 \times 10^{-9}$\\
S3D & $9.375 \times 10^{-5}$ & $3.145 \times 10^{-5}$ & $5.020 \times 10^{-5}$ & $2.654 \times 10^{-5}$\\
SCALE-LETKF & $1.000 \times 10^{-6}$ & $2.500 \times 10^{-7}$ & $2.500 \times 10^{-7}$ & $2.969 \times 10^{-7}$\\
\hline
\multicolumn{4}{c}{Critical Points} \\ \hline
QMCPACK & $8.750 \times 10^{-4}$ & $3.965 \times 10^{-4}$ & $9.082 \times 10^{-4}$ & $6.250 \times 10^{-4}$\\
Viscous Fingering & $2.617 \times 10^{-4}$ & $5.396 \times 10^{-5}$ & $6.982 \times 10^{-5}$ & $7.715 \times 10^{-5}$\\
Tangaroa & $8.750 \times 10^{-4}$ & $5.313 \times 10^{-4}$ & $8.828 \times 10^{-4}$ & $6.797 \times 10^{-4}$\\
Earthquake & $4.961 \times 10^{-4}$ & $4.166 \times 10^{-5}$ & $6.342 \times 10^{-5}$ & $3.882 \times 10^{-5}$\\
Ionization & $8.750 \times 10^{-4}$ & $7.715 \times 10^{-5}$ & $2.822 \times 10^{-4}$ & $7.603 \times 10^{-5}$\\
Isabel & $8.750 \times 10^{-4}$ & $7.678 \times 10^{-4}$ & $8.730 \times 10^{-4}$ & $8.701 \times 10^{-4}$\\
Miranda & $4.883 \times 10^{-8}$ & $4.883 \times 10^{-8}$ & $4.883 \times 10^{-8}$ & $4.883 \times 10^{-8}$\\
S3D & $8.750 \times 10^{-4}$ & $8.057 \times 10^{-4}$ & $1.000 \times 10^{-3}$ & $8.979 \times 10^{-4}$\\
SCALE-LETKF & $1.000 \times 10^{-4}$ & $3.906 \times 10^{-5}$ & $2.344 \times 10^{-5}$ & $7.267 \times 10^{-5}$\\
\hline
\multicolumn{4}{c}{Percentiles} \\ \hline
QMCPACK & $8.750 \times 10^{-5}$ & $4.258 \times 10^{-5}$ & $5.781 \times 10^{-5}$ & $4.844 \times 10^{-5}$\\
Viscous Fingering & $3.438 \times 10^{-5}$ & $1.387 \times 10^{-5}$ & $1.797 \times 10^{-5}$ & $1.523 \times 10^{-5}$\\
Tangaroa & $8.750 \times 10^{-5}$ & $5.625 \times 10^{-5}$ & $7.187 \times 10^{-5}$ & $6.290 \times 10^{-5}$\\
Earthquake & $6.250 \times 10^{-5}$ & $1.865 \times 10^{-5}$ & $2.559 \times 10^{-5}$ & $1.831 \times 10^{-5}$\\
Ionization & $8.750 \times 10^{-5}$ & $3.555 \times 10^{-5}$ & $8.096 \times 10^{-5}$ & $4.141 \times 10^{-5}$\\
Isabel & $8.750 \times 10^{-5}$ & $8.594 \times 10^{-5}$ & $8.906 \times 10^{-5}$ & $8.672 \times 10^{-5}$\\
Miranda & $4.883 \times 10^{-10}$ & $4.883 \times 10^{-10}$ & $4.883 \times 10^{-10}$ & $4.883 \times 10^{-10}$\\
S3D & $9.375 \times 10^{-5}$ & $6.719 \times 10^{-5}$ & $9.434 \times 10^{-5}$ & $6.689 \times 10^{-5}$\\
SCALE-LETKF & $8.750 \times 10^{-5}$ & $6.895 \times 10^{-5}$ & $5.234 \times 10^{-5}$ & $8.438 \times 10^{-5}$\\
\hline
\end{tabular}
\end{center}
\label{tab:equivalent-eb-bounds}
\end{table*}

\begin{table*}
\caption{Error bounds used for the main experiments with {\tool}. In general, the same error bounds are used for each task; exceptions are specified for particular (task, dataset) pairs.}
\begin{center}
\begin{tabular}{l|l} \hline
Setting & Error Bounds \\ \hline
Contour Tree ($\varepsilon = 0.04$) & $1 \times 10^{-5}$, $3.25 \times 10^{-5}$, $5.5 \times 10^{-5}$, $7.75 \times 10^{-5}$, $0.0001$ \\
Contour Tree ($\varepsilon = 0.04$) / Viscous Fingering & $1 \times 10^{-6}$, $3.25 \times 10^{-6}$, $5.5 \times 10^{-6}$, $7.75 \times 10^{-6}$, $1 \times 10^{-5}$\\
\hline
Morse-Smale Segmentation & $1 \times 10^{-6}$, $3.25 \times 10^{-6}$, $5.5 \times 10^{-6}$, $7.75 \times 10^{-6}$, $1 \times 10^{-5}$ \\
\hline
Local Order & $1 \times 10^{-5}$, $3.25 \times 10^{-5}$, $5.5 \times 10^{-5}$, $7.75 \times 10^{-5}$, $0.0001$ \\
Local Order / Tangaroa & $1 \times 10^{-6}$, $3.25 \times 10^{-6}$, $5.5 \times 10^{-6}$, $7.75 \times 10^{-6}$, $1 \times 10^{-5}$\\
Local Order / Miranda & $1 \times 10^{-6}$, $3.25 \times 10^{-6}$, $5.5 \times 10^{-6}$, $7.75 \times 10^{-6}$, $1 \times 10^{-5}$\\
Local Order / SCALE-LETKF & $1 \times 10^{-7}$, $3.25 \times 10^{-7}$, $5.5 \times 10^{-7}$, $7.75 \times 10^{-7}$, $1 \times 10^{-6}$\\
\hline
Critical Points & $0.0001$, $0.000325$, $0.00055$, $0.000775$, $0.001$ \\
Critical Points / Miranda & $1 \times 10^{-5}$, $3.25 \times 10^{-5}$, $5.5 \times 10^{-5}$, $7.75 \times 10^{-5}$, $0.0001$\\
Critical Points / SCALE-LETKF & $1 \times 10^{-5}$, $3.25 \times 10^{-5}$, $5.5 \times 10^{-5}$, $7.75 \times 10^{-5}$, $0.0001$\\
\hline
Isocontours & $0.0001$, $0.000325$, $0.00055$, $0.000775$, $0.001$ \\
\hline
Percentiles & $1 \times 10^{-5}$, $3.25 \times 10^{-5}$, $5.5 \times 10^{-5}$, $7.75 \times 10^{-5}$, $0.0001$ \\
Percentiles / Miranda & $1 \times 10^{-7}$, $3.25 \times 10^{-7}$, $5.5 \times 10^{-7}$, $7.75 \times 10^{-7}$, $1 \times 10^{-6}$\\
\hline
Combination & $1 \times 10^{-6}$, $3.25 \times 10^{-6}$, $5.5 \times 10^{-6}$, $7.75 \times 10^{-6}$, $1 \times 10^{-5}$ \\
Combination / Miranda & $1 \times 10^{-8}$, $3.25 \times 10^{-8}$, $5.5 \times 10^{-8}$, $7.75 \times 10^{-8}$, $1 \times 10^{-7}$\\
Combination / SCALE-LETKF & $1 \times 10^{-7}$, $3.25 \times 10^{-7}$, $5.5 \times 10^{-7}$, $7.75 \times 10^{-7}$, $1 \times 10^{-6}$\\
\hline
\end{tabular}
\end{center}
\label{tab:main-error-bounds}
\end{table*}


\end{document}